\documentclass{amsart}
\usepackage[utf8]{inputenc}
\usepackage{amssymb}
\usepackage{amsthm}
\usepackage{amsfonts}
\usepackage{amsmath}
\usepackage{hyperref}
\usepackage{multicol,multirow}
\usepackage{algorithm,algcompatible}
\usepackage[text={440pt,575pt},headheight=9pt,centering]{geometry}
\usepackage{comment}
\usepackage{enumitem}
\usepackage[font=small]{subfig}
\usepackage{tikz}
\usetikzlibrary{shapes, positioning}
\newcommand{\indep}{\perp \!\!\! \perp}
\newcommand{\exindep}{\perp_e}

\newtheorem{theorem}{Theorem}
\newtheorem{lemma}[theorem]{Lemma}
\newtheorem{proposition}[theorem]{Proposition}
\newtheorem{corollary}[theorem]{Corollary}
\newtheorem{definition}{Definition}

\newtheorem{assumption}{Assumption}
\newtheorem{remark}{Remark}

\usepackage{graphicx}
\usepackage{placeins}
\usepackage[comma, sort&compress]{natbib} % has a nice set of citation styles and commands
\newcommand{\R}{\mathbb{R}}
\newcommand{\proj}{\mathcal{P}}

\newcommand{\Gl}{\mathcal{J}}
\newcommand{\Gla}{\mathcal{J}^+}

\newcommand{\I}{\mathbb{I}^\star}
\newcommand{\Hs}{\mathbb{H}}

\newcommand{\M}{\mathcal{M}}
\newcommand{\HR}{H\"usler--Reiss}

\newcommand{\ones}{\mathbf{1}_p\mathbf{1}_p^\top}
\newcommand{\inc}{\mu}

\DeclareMathOperator{\Var}{Var}
\DeclareMathOperator*{\argmin}{argmin} % no space, limits underneath in displays
\usepackage{color}

 \usepackage[foot]{amsaddr}

\title{Extremal graphical modeling with latent variables via convex optimization}
\author[S.~Engelke]{Sebastian Engelke$^1$}
\email{sebastian.engelke@unige.ch}
\address{$^1$Research Center for Statistics, GSEM, University of Geneva, Switzerland}
\author[A.~Taeb]{Armeen Taeb$^{2}$}
\email{ataeb@uw.edu}
\address{$^2$Department of Statistics, University of Washington, U.S.}
\date{\today}
\usepackage{etoolbox}

\begin{document}
\maketitle

\begin{abstract}
Extremal graphical models encode the conditional independence structure of multivariate extremes and provide a powerful tool for quantifying the risk of rare events. Prior work on learning these graphs from data has focused on the setting where all relevant variables are observed. For the popular class of \HR{} models, we propose the \texttt{eglatent} method, a tractable convex program for learning extremal graphical models in the presence of latent variables. Our approach decomposes the \HR{} precision matrix into a sparse component encoding the graphical structure among the observed variables after conditioning on the latent variables, and a low-rank component encoding the effect of a few latent variables on the observed variables. We provide finite-sample guarantees of \texttt{eglatent} and show that it consistently recovers the conditional graph as well as the number of latent variables. We highlight the improved performances of our approach on synthetic and real data.
\end{abstract}

\begingroup
\def\uppercasenonmath#1{} % this disables uppercasing title
\let\MakeUppercase\relax % this disables uppercasing authors
\section{Introduction}
\endgroup
\label{sec:intro}

Floods, heat waves, and financial crashes illustrate the environmental and economic hazards primarily influenced by rare, yet significant, events.
Such catastrophic scenarios often result from the simultaneous occurrence of extreme values across multiple variables \citep{Zhou2009DependenceSO,Asadi2015ExtremesOR,Zscheischler17}.
To effectively measure and mitigate these disasters, it is essential to understand the dependencies between the various risk factors. From a mathematical perspective, this requires examining the tail dependence between the components of the random vector $X = (X_1, \dots, X_d)$.
Extreme value theory provides the theoretical foundation for extrapolations to the distributional tail of the random vector $X$. Within the multivariate setting, there are two different yet closely related approaches for modeling extremal data. The first method considers component-wise maxima of independent copies of $X$ and leads to the notion of max-stable distributions \citep{deHaan1977}. The second method relies on multivariate Pareto distributions that describe the random vector $X$ conditioned on the event that there is an extreme in one of the coordinates of~$X$ \citep{rootzen2006}.

Given the increasing complexity and dimensionality of contemporary data sets, identifying sparse representations for distributions of extreme events is critical for accurate modeling and risk assessment \citep{engelke2021a}. Graphical models serve as powerful tools in achieving such sparse representations, offering clear and interpretable models for understanding dependencies among variables \citep{lauritzen1996}. However, in the framework of max-stable distributions, \cite{papastathopoulos2016} highlighted limitations in developing non-trivial graphical models for their densities.
On the other hand, multivariate Pareto distributions do not face these limitations. Indeed, \cite{engelke2020} introduced extremal graphical models that factorize according to multivariate Pareto distributions and encode extremal conditional independence relationships, and \cite{segers2020} showed that extremal trees naturally arise as limits of Markov trees. For the popular \HR{} family \citep{hueslerReiss1989}, \cite{hen2022} showed that, similarly to the Gaussian case, the sparsity pattern of an extremal graphical model can be read off from a positive semi-definite precision matrix $\Theta$ with the all-ones vector in its null space. This precision matrix $\Theta$ is derived from a transformation of the variogram matrix $\Gamma$ that parameterizes a \HR{} distribution.
 Several recent papers have proposed methods to learn the extremal graphical structure from data \citep{engelke2022b, hu2022,engelke2022a, chang_allen_2023subbotin, wan2023graphical, lederer2023extremes}.

The study and techniques for modeling extremes have so far concentrated on scenarios where all relevant variables are directly observable. However, in many real-world situations, there exist latent variables that are not observable due to prohibitive costs or other practical constraints. Mathematically, the overall system of variables is then given by $X = (X_O, X_H)$, where $X_O$ are the observed and $X_H$ the latent variables, with $(O,H) = \{1,2,\dots,d\}$. The importance of accounting for latent factors becomes apparent in the example of a single latent variable $X_H = \{X_c\}$, where the data is generated through the one-factor model
\[ X_j = X_c + \varepsilon_j, \quad j\in O.\]
Here, $X_c$ is the common (unobserved) factor influencing all observed variables, and  $\varepsilon_j$, $j\in O$, are independent noise terms. Suppose that the exceedances of the random vector $X$ converge in distribution to a multivariate Pareto distribution $Y = (Y_O, Y_H)$; a concrete example where this is satisfied is when $X_H$ is standard exponential and the noise variables are normally distributed, in which case $Y$ has a \HR{} distribution, but many other combinations are possible \citep{eng2018a}. 
The joint vector $Y$ can be shown to be an extremal graphical model with respect to the star graph on the left-hand side of Figure~\ref{fig:graphs_mot}, where the observed variables $Y_O$ are conditionally independent given the latent variable $Y_H$. However, the sub-model model of $Y$ corresponding to the observed variables, that is, the limiting multivariate Pareto distribution arising from threshold exceedances of $X_O$, induces, in general, the fully connected extremal graph on the right-hand side of Figure~\ref{fig:graphs_mot}, where are all the variables are conditionally dependent.

This simple example illustrates that ignoring the effect of latent variables induces confounding dependencies among the observed variables: even for a sparse joint graph of observed and latent variables, any two observed variables are dependent when conditioning on the remaining observed variables. This phenomenon also appears in many real-world applications. In such cases, a latent extremal graphical model with possibly more than one latent variable $Y_H$ serves multiple purposes: i) it obtains the number of latent variables $h = |H|$ that summarize the effect of external phenomena on the observed variables, (ii) it identifies the residual graph structure among the observed variables after extracting away the effect of these external factors, (iii) it often yields a more sparsely represented and accurate statistical model than a model that ignores the latent variables. Latent extremal graphical models have only been studied when the graphical structure among the observed and latent variables is a tree, and where the tree structure is assumed to be known \citep{asenova2021extremes,roettger2023}.

\captionsetup[figure]{font=footnotesize}
\begin{figure}
\begin{center}
    \begin{minipage}{0.4\linewidth}
    \resizebox{4.2cm}{4.2cm}{
\begin{tikzpicture}[>=stealth,scale=0.20]
\node[draw, circle, fill=gray!20] (latent) at (0,0) {$H$};
% Observed Nodes
\node[draw, circle, above=of latent] (obs1) {$O_1$};
\node[draw, circle, right=of latent] (obs2) {$O_2$};
\node[draw, circle, below=of latent] (obs3) {$O_3$};
\node[draw, circle, left=of latent] (obs4) {$O_4$};
% Arrows with dashed lines between observed nodes
\draw[-, dashed] (latent) -- (obs1);
\draw[-, dashed] (latent) -- (obs2);
\draw[-, dashed] (latent) -- (obs3);
\draw[-, dashed] (latent) -- (obs4);
\end{tikzpicture}
}
\end{minipage}
\begin{minipage}{0.4\linewidth}
\centering
% Latent Node
   \resizebox{4.2cm}{4.2cm}{
\begin{tikzpicture}[>=stealth]
%\invisible{\node[draw, circle, fill=gray!20] (latent) at (0,0) {$L$};}
% Observed Nodes
\node[draw, circle] (obs1) {$O_1$};
\node[draw, circle, below right= of obs1] (obs2) {$O_2$};
\node[draw, circle, below left=of obs1] (obs4) {$O_4$};
\node[draw, circle, below right=of obs4] (obs3) {$O_3$};
\draw[-] (obs1) -- (obs2);
\draw[-] (obs1) -- (obs3);
\draw[-] (obs2) -- (obs4);
\draw[-] (obs1) -- (obs2);
\draw[-] (obs2) -- (obs3);
\draw[-] (obs3) -- (obs4);
\draw[-] (obs4) -- (obs1);
\end{tikzpicture}}
\end{minipage}
\caption{One-factor graph with one latent variable with four observed variables $O_1,\dots, O_4$ and one latent variable $H$ (left) and its marginalization on the observed variables (right).}
\label{fig:graphs_mot}
\end{center}
\end{figure}
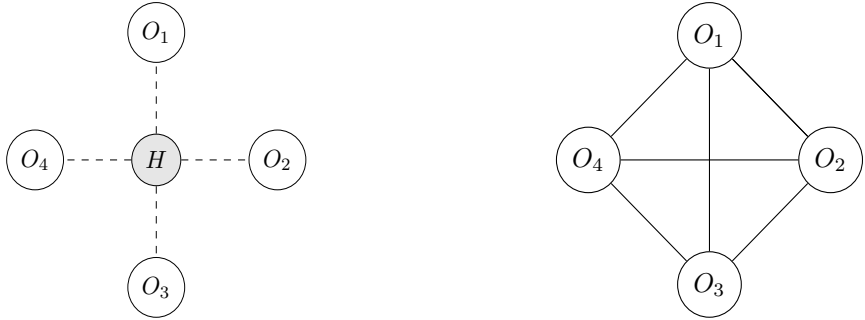

\subsection{Our contributions}
We introduce a general latent \HR{} graphical model where the graphical structure among the observed and latent variables as well as the number of latent variables may be arbitrary. Letting $\Theta \in \mathbb{R}^{d \times d}$ be the precision matrix, a key result that we establish is that the marginal precision matrix $\tilde{\Theta} \in \mathbb{R}^{p \times p}$ over the observed variables can be expressed in terms of blocks of $\Theta$ as
\begin{align*}
    \tilde{\Theta} = \Theta_O - \Theta_{OH}\Theta_{H}^{-1}\Theta_{HO}, \quad \text{where} \quad  \Theta = \begin{pmatrix}\Theta_{O}& \Theta_{OH} \\ \Theta_{HO} & \Theta_H\end{pmatrix}~,~ {\tilde{\Theta}\textbf{1}_p = 0, \text{ and } \Theta\textbf{1}_d = 0}.
\end{align*}
Here, $\textbf{1}_r$ is the all-ones vector with $r$ coordinates. The representation of $\Theta$ resembles the Schur complement in Gaussian latent variable graphical models \citep{Chand2012}. However, in the \HR{} case, the matrices $\Theta$ and $\tilde{\Theta}$ are not invertible since they have the all-one vector in their kernel, and the link between our representation and the Schur complement is therefore non-trivial. 

Assuming that the conditional graph among the observed variables is sparse and that there are a few latent variables influencing the observed variables, the marginal precision matrix $\tilde{\Theta}$ is decomposed as the sum of a sparse and a low-rank matrix, i.e., $\tilde{\Theta}^\star = S^\star-L^\star$. The sparse component $S^\star:= \Theta_O$ encodes the conditional graphical structure among the observed variables after conditioning on the latent variables and the low-rank component $L^\star := \Theta_{OH}\Theta_{H}^{-1}\Theta_{HO}$ encodes the effect of a few latent variables on the observed variables. Using this decomposition, we propose a \emph{convex optimization procedure} named \texttt{eglatent} that provides estimates $(S,L)$ for each term in the decomposition without knowledge of the underlying graphical structure or the number of latent variables. {Compared to the latent variable graphical modeling estimator in \cite{Chand2012}, \texttt{eglatent} has the additional constraint that the matrix $S-L$ has the all-ones vector in its kernel. Due to this structural constraint that arises in extremal models, in addition to assuming that the  
 number of latent variables is small (compared to the observed variables) and they affect many observed variables, we require new identifiability assumptions for recovering $S^\star$ and $L^\star$}. Under these identifiability assumptions, we provide finite-sample consistency guarantees for our estimator, showing that our procedure recovers the conditional graph and the number of latent variables. %Our identifiability conditions assume that the number of latent variables is small (compared to the observed variables) and they affect many observed variables. 

Figure~\ref{fig:intro} highlights the advantage of our method \texttt{eglatent} over the existing extremal graph learning method \texttt{eglearn} \citep{engelke2022b}, which does not account for latent variables. In this synthetic example, we generated $2000$ approximate observations from an extremal graphical model with $h=2$ latent variables and a cycle graph among $p=30$ observed variables, and fitted both methods for different values of the regularization parameters; see Section~\ref{sec:struc_recov} for details on the setup. Compared to \texttt{eglearn}, our \texttt{eglatent} produces a better model fit on validation data and more accurate graph estimates among the observed variables in terms of $F$-score. Indeed, due to the latent confounding, the marginal graph among the observed variables, encoded by the zero pattern in $\tilde{\Theta}$, is dense, and thus the sparsity that \texttt{eglearn} exploits is not appropriate: the best validated \texttt{eglearn} model has $252$ edges while the true graph has $30$ edges. On the other hand, conditional on the latent variables, the conditional graph among the observed variables, encoded by the zero pattern in $\Theta_O$, is sparse, and \texttt{eglatent} exploits this structure. Furthermore, \texttt{eglatent} estimates the correct number of latent variables and a near-perfect graph among the observed variables for regularization parameters with high validation likelihood. Note that in the left plot, the crosses for \texttt{eglearn} mean that the estimated graphical model is disconnected and therefore does not lead to a valid \HR{} model. In contrast, \texttt{eglatent} always yields a valid \HR{} model. More simulations and an application to large flight delays in the U.S.~are presented in Section~\ref{sec:experiments} that demonstrate the utility of our approach. 

{In summary, compared to the previous literature in extremal graphical modeling and Gaussian latent variable graphical modeling, our contributions are threefold. From a methodological perspective, we provide the first method to learn general extremal graphical models with latent variables. Our approach \texttt{eglatent} is based on a tractable convex optimization procedure that resembles the estimator in \cite{Chand2012} but involves an additional constraint due to the structural properties of extremal models. From a practical perspective, compared to existing extremal graphical modeling approaches that do not account for latent variables, \texttt{eglatent} often yields sparser and thus more interpretable graphical models with better fit to data. Theoretically, to arrive at our estimator \texttt{eglatent}, we prove a non-trivial Schur decomposition of the observed precision matrix of the observed variables. Further, since \texttt{eglatent} differs from the estimator in \cite{Chand2012}, we require new identifiability assumptions and conduct a more involved analysis to establish finite-sample consistency guarantees.} 

Our \texttt{eglatent} method is implemented in the R package \texttt{graphicalExtremes} \citep{graphicalExtremes2022} and all numerical results and figures can be reproduced using the code on \url{https://github.com/sebastian-engelke/extremal_latent_learning}. 

\begin{figure}
    \centering
    \includegraphics[width = .9\textwidth]{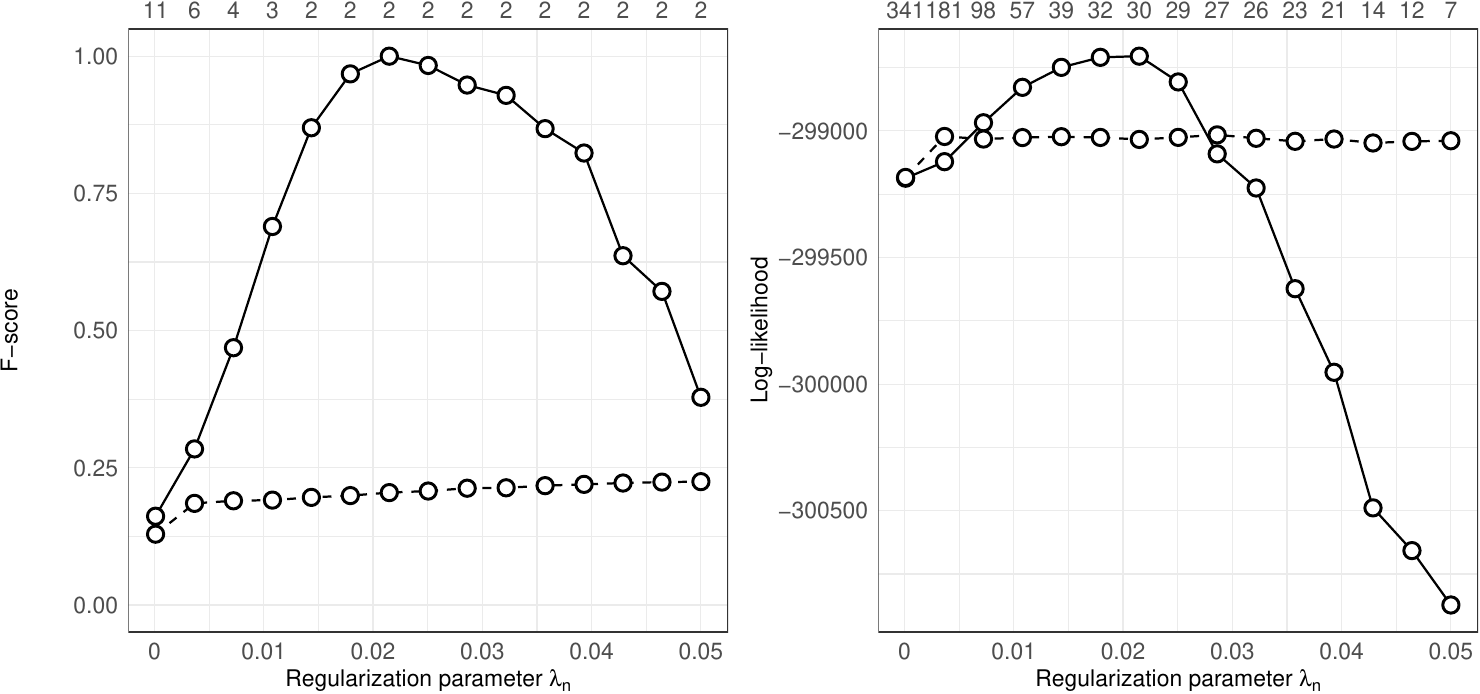}
    \caption{Left: $F$-score of our proposed method \texttt{eglatent} (solid line) and \texttt{eglearn} (dashed line) as function of the regularization parameter with larger $F$-scores being better; top axis shows the number of estimated latent variables. Right: the likelihood of the same methods evaluated on a validation data set; the top axis shows the number of estimated edges in the latent model.}
    \label{fig:intro}
\end{figure}

\subsection{Notation}
We denote $I_{r}$ as an $r \times r$ identity matrix and denote $\mathbf{1}_r$ as the all-ones vector with $r$ coordinates. The collection of $r \times r$ symmetric matrices is denoted by $\mathbb{S}^r$. The following matrix norms are employed throughout this paper: $\|M\|_2$ denotes the spectral norm, or the largest singular value of $M$; $\|M\|_\infty$ denotes the largest entry in the magnitude of $M$; $\|M\|_\star$ denotes the nuclear norm,
or the sum of the singular values of $M$ (this reduces to the trace for positive semidefinite matrices); and $\|M\|_1$ denotes the sum of the absolute values of the entries of $M$. Finally, we will denote $\sigma_\mathrm{min}(M)$ as the largest non-zero singular value of $M$.

\section{Background}\label{sec:background}

\subsection{Multivariate extreme value theory}
\label{sec:mevt}
Multivariate extreme value theory studies asymptotically motivated models for the largest observations of a random vector $X = (X_j: j \in V)$ with index set $V = \{1,\dots, d\}$. Since we concentrate on models for the extremal dependence structure, we assume that the marginal distributions of $X$ have been standardized to standard exponential distributions. In practice, this standardization can be achieved by using the marginal empirical distribution functions; see Section~\ref{sec:exmpirical_variogram}. 

A multivariate Pareto distribution models the multivariate tail of the distribution of $X$. It is defined as the limit in the distribution of the conditional exceedances over a high threshold $u$, that is, 
\begin{align}\label{def_pareto}
    Y = \lim_{u \to \infty} \left( X - u  \mid \max(X_1,\dots, X_d) > u \right),
\end{align}
if the limit exists \citep{rootzen2006}. Here the simple normalization by subtracting $u$ in each component of $X$ is due to the exponential marginals. The random vector $X$ is said to be in the domain of attraction of the multivariate Pareto distribution $Y$, which is supported on the space $\mathcal L = \{y \in \mathbb R^d : \max(y_1,\dots, y_d) > 0\}.$ Multivariate Pareto distributions are the only possible limits of threshold exceedances \citep{rootzen2018} and therefore a canonical model for extremes. 
If the convergence in~\eqref{def_pareto} holds, it is easy to see that for any non-empty subset $I \subset V$, the sub-vector $X_I = (X_j:j\in I)$ is itself in the domain of attraction of a $|I|$-dimensional Pareto distribution, which we call the $I$th sub-model of $Y$.

We now introduce the \HR{} model, which is the most popular parametric sub-class of multivariate Pareto distributions. It can be seen as the analog of Gaussian distributions in multivariate extreme value theory, a fact, that will become apparent when studying extremal graphical models in the next section.
\begin{definition}\label{def:HR}
    A multivariate Pareto distribution $Y = (Y_1,\dots, Y_d)$ is called a \HR{} distribution parameterized by the variogram matrix $\Gamma$ in the space of conditionally negative definite matrices
    \begin{align}
    \mathcal C^d
    =\{\Gamma \in [0,\infty)^{d\times d} :  \Gamma = \Gamma^\top,\;\;
        \rm{diag}({\Gamma}) = \mathbf 0
        ,\;\;
        v^\top\Gamma v < 0
        \;\forall\,
        \mathbf 0 \neq v \perp \mathbf 1
    \},
\end{align}
    if its density has the form 
    \begin{align}\label{HR_density}
        f(y; \Gamma)
        & = c_{\Gamma}
	\exp\left\{ -\frac12 (y - \mu_\Gamma)^\top \Theta (y - \mu_\Gamma) - \frac1d \sum_{i=1}^d y_i \right\}, \quad y \in \mathcal L,
\end{align}
where $c_\Gamma > 0$ is a normalizing constant, $\mu_\Gamma = \Pi(-\Gamma/2) \mathbf 1_d$, and $\Pi = I_d - \mathbf{1}_d\mathbf{1}_d^\top /d$ is the projection matrix onto the orthogonal complement of the all-ones vector in d-dimensions.
The matrix $\Theta = (\Pi (-\Gamma/2) \Pi)^+$ is the positive semi-definite \HR{} precision matrix \citep{hen2022}, where $A^+$ is the Moore--Penrose pseudoinverse of a matrix~$A$.    
\label{defn:mpd}
\end{definition}    

The \HR{} distribution is stable under marginalization, in the sense that for $I \subset V$, the \HR{} sub-model corresponding to the $I$th marginal is again \HR{} distributed with parameter matrix $\Gamma_I$.
While the density in~\eqref{HR_density} resembles the density of a multivariate normal distribution, we note that there are important differences. First, this function would not have finite integral on $\mathbb R^d$ because of the second term in the exponential, and the restriction to the subset $\mathcal L$ is crucial. Second, the precision matrix $\Theta$ is of rank $d-1$, which complicates theoretical and practical considerations.  

An important summary statistic of the dependence structure in multivariate Pareto distributions is the extremal variogram \citep{engelke2022a}. It takes a similar role as the covariance matrix in the non-extremal world.
\begin{definition}
    For a multivariate Pareto distribution $Y = (Y_j: j\in V)$ the extremal variogram rooted at node $m\in V$ is defined as the matrix $\Gamma^{(m)}$ with entries
    \[ \Gamma^{(m)}_{ij} = \Var\left\{Y_i - Y_j \mid Y_m > 1 \right\}, \quad i,j \in V,\]
    whenever the right-hand side is finite.
    \label{defn:rooted_variogram}
\end{definition}

If $Y$ follows a \HR{} distribution with parameter matrix $\Gamma$, it can be checked that the extremal variogram matrices coincide for all $m\in V$, and that they satisfy
\begin{align}
    \label{HR_vario} \Gamma = \Gamma^{(1)} = \dots = \Gamma^{(d)}.
\end{align} 
We use this fact later to combine empirical estimators of the extremal variograms rooted at the different nodes to obtain a more efficient joint estimator of $\Gamma$.

\subsection{Extremal graphical models}

Conditional independence for multivariate Pareto distributions $Y$ is non-standard since it is defined on the space $\mathcal L$, which is not a product space. \cite{engelke2020} therefore define a new notion of extremal conditional independence using the auxiliary vectors $Y^{(m)}$, for $m\in\{1,\dots, d\}$, defined as $Y$ conditioned on the event that $\{Y_m > 0\}$. For non-empty subsets $A, B, C \subset V$, we say that $Y_A$ is conditionally independent of $Y_B$ given $Y_C$, denoted by $Y_A \exindep Y_B \mid Y_C$, if for all auxiliary random vectors, we have the corresponding statement in the usual sense, that is,
\[ Y^{(m)}_A \indep Y^{(m)}_B \mid Y^{(m)}_C \quad \text{ for all } m\in V. \]
It can be shown that requiring the relation above is equivalent to requiring the existence of a single $m\in V$ for which $Y^{(m)}_A \indep Y^{(m)}_B \mid Y^{(m)}_C$  \citep{eng_iva_kir}.

Let $\mathcal{G}=(V,E)$ be an undirected graph with nodes $V = \{1,\dots, d\}$ and edge set $E \subset V\times V$.
Using the new notion of conditional independence, an {extremal graphical model} on $\mathcal{G}$ is a multivariate Pareto distribution $Y$ that satisfies the extremal pairwise Markov property on $\mathcal{G}$, that is,
\[Y_i \exindep Y_j \mid Y_{V\setminus \{i,j\}} \quad \text{ if } (i,j) \notin E.  \]
\cite{engelke2020} show that this definition is natural in the sense that it enables a Hammersley--Clifford theorem showing that densities factorize into lower-dimensional terms on the cliques of the graph. 

For a multivariate Gaussian distribution with covariance matrix $\Sigma$, the conditional dependence relationships, or equivalently the edges in the Gaussian graphical model can be identified from the nonzeros of the precision matrix $\Sigma^{-1}$. A similar property holds for extremal graphical models if $Y$ follows a \HR{} distribution, where the matrix $\Theta$ in Definition~\ref{defn:mpd} plays a key role.    
\begin{proposition}[Lemma 1 and Proposition 3 of \cite{engelke2020}]
   \label{prop:Theta_extreme} Let $Y \in \mathbb{R}^d$ follow a \HR{} distribution with precision matrix $\Theta$. Then,
   \begin{equation} Y_i \exindep Y_j \mid Y_{V\setminus \{i,j\}} \, \Leftrightarrow \, \Theta_{ij} = 0.
   \label{eqn:CI_prec}
   \end{equation}
\end{proposition}
 A consequence of Proposition~\ref{prop:Theta_extreme} is that for a \HR{} graphical model on an arbitrary connected graph $\mathcal{G}$, we can read off the graph structure from the zero pattern of the precision matrix $\Theta$.

Finally, we note that an important property of an extremal graphical model is that if $Y$ possesses a density that factorizes on the graph $\mathcal{G}$, then $\mathcal{G}$ must \emph{necessarily be connected} \citep{engelke2020}. The state-of-the-art structure learning methods for extremal data \citep{engelke2022b, wan2023graphical} can yield disconnected graphs that thus do not always yield a valid distribution (see the example in Figure~\ref{fig:graphs_mot}). 
For a detailed review of recent progress on extremal graphical models, we refer to \cite{engelke2024graphical}.
In the next section, we present our approach for structure learning, which can handle latent variables and always yields a valid distribution.

\section{Latent H\"usler--Reiss models and the \texttt{eglatent} method}
\label{sec:model_estimation}
\subsection{Latent \HR{} models}
\label{sec:latent_hr_graphs}
In the illustrative example in the introduction, we presented a \HR{} model with a single latent variable and a very simple graphical structure among the observed variables. We next introduce a latent \HR{} model with a general extremal graphical structure and any number of latent variables. In what follows, let $X_O \in \mathbb{R}^{p}$ be the collection of observed variables, $X_{H} \in \mathbb{R}^{h}$ be a collection of latent variables, and put $d := p+h$.
\begin{definition}[Latent \HR{} models]\label{def:LHR}
    Suppose that the random vector $X= (X_O\allowbreak, X_H) \in\mathbb R^{d}$, indexed by $V =(O,H)$, is in the domain of attraction of a \HR{} distribution $Y  \in \mathbb{R}^{d}$ in the sense of~\eqref{def_pareto} with variogram and precision matrices, and corresponding extremal graphical structure
    \begin{align*}\label{joint_gamma}
     \Gamma = \begin{pmatrix}\Gamma_{O}& \Gamma_{OH} \\ \Gamma_{HO} & \Gamma_H\end{pmatrix}, \quad \quad \Theta = \begin{pmatrix}\Theta_{O}& \Theta_{OH} \\ \Theta_{HO} & \Theta_H\end{pmatrix}, \quad \text{ and }\quad \mathcal{G} = (V,E),
\end{align*}
respectively.
Here $\Theta = (\Pi(-\Gamma/2)\Pi)^{+}$, $\Gamma_O$ and $\Theta_O$ are $p \times p$-dimensional symmetric matrices, and $E = \{(i,j): i,j \in V, i\neq j, \Theta_{ij}\neq 0\}$. We then say that $Y$ is a latent \HR{} model, and we note that the observed variables $X_O$ are in the domain of attraction of a \HR{} model with variogram $\Gamma_O$.
\end{definition}

Note that $\Theta_O$ and $\Theta_H$ in the above definition are positive definite since $\Gamma \in \mathcal C^d$; see Definition~\ref{def:HR} and \citet[][Appendix B]{engelke2020}. Latent \HR{} models have been studied only for very simple graphs, namely tree structures \citep{asenova2021inference, roettger2023} and block graphs \citep{asenova2021extremes}. All of the above methods assume the underlying graph structure among the observed and latent variables and the number of latent variables to be known, which is rarely realistic in practice. To handle more general graphs, we establish the following theorem, which relates the marginal distribution of the observed variables to components of the precision matrix $\Theta$. 

\begin{theorem}
Let $\tilde{\Pi} = I_{p}-\mathbf{1}_p\mathbf{1}_p^T/p$ be the projection matrix onto the orthogonal complement of the all-ones vector in $p$ dimensions. Then, the precision matrix $\tilde{\Theta} \in \mathbb{R}^{p \times p}$ of the observed variables of a latent \HR{} model with variogram matrix $\Gamma$ satisfies
\begin{equation}\label{precision_obs}
\begin{aligned}
    \tilde{\Theta} = (\tilde{\Pi}(-\Gamma_O/2)\tilde{\Pi})^+ &= \Theta_O - \Theta_{OH}\Theta_{H}^{-1}\Theta_{HO}.
\end{aligned}
\end{equation}
\label{thm:main1}
\end{theorem}
While it is not possible to observe the joint precision matrix $\Theta$ or any of its components directly, Theorem~\ref{thm:main1} provides a useful decomposition of the observable precision matrix $\tilde{\Theta}$ into the difference of two terms, each term involving the components of $\Theta$. By the property in \eqref{eqn:CI_prec}, we have for any  $i,j \in O$ that  
\[ Y_i \perp_e Y_j \mid Y_H,Y_{O\setminus\{i,j\}} \quad \Leftrightarrow \quad [\Theta_O]_{i,j} = 0.\] 
Thus, the first term $\Theta_O$ in decomposition \eqref{precision_obs} specifies the conditional independencies among the observed variables after conditioning on the latent variables. Moreover, the sparsity pattern of $\Theta_O$ encodes the residual graph $\mathcal{G}_O$ among the observed variables after extracting the influence of the latent variables. Here, $\mathcal{G}_O = (O,E_O)$ is a subgraph of $\mathcal{G}$ restricted to the observed variables where $E_O = \{(i,j) \in E, i,j \in O\}$. The second term $\Theta_{OH}\Theta_{H}^{-1}\Theta_{HO}$ in decomposition \eqref{precision_obs} serves as a summary of the marginalization of the latent variables $Y_H$ and encodes their effect on the observed variables. The rank of this matrix is equal to the number of latent variables. The overall term $\Theta_O - \Theta_{OH}\Theta_{H}^{-1}\Theta_{HO}$ is a Schur complement with respect to $\Theta_H$.

As an illustration, consider the extremal graph on the left-hand side of Figure~\ref{fig:graphs_mot}. Here, the matrix $\Theta_O$ is diagonal. Furthermore, the matrix $\Theta_{OH}\Theta_{H}^{-1}\Theta_{HO}$ has rank equal to one with all of its entries being nonzero. Note that $\tilde{\Theta}$ generally consists of all nonzero entries and hence the marginal graphical structure among the observed variables on the right-hand side of Figure~\ref{fig:graphs_mot} is fully connected.\\

\subsection{Sparse plus low-rank decomposition} In this paper, we consider a latent \HR{} graphical model where the subgraph $\mathcal{G}_O$ among the observed variables is sparse and the number of latent variables is small relative to the number of observed variables, that is, $h \ll p$. This modeling assumption is often natural in real-world applications. For example, \cite{Chand2012} and \cite{Taeb2016InterpretingLV} showed that a large fraction of the conditional dependencies among stock returns can be explained by a small number of latent variables and interpreted these to be correlated to exchange rate and government expenditures. In a similar spirit, \cite{Taeb2017ASG} demonstrated that the California reservoir network is sparsely connected after accounting for a few latent factors, and interpreted these latent factors to be highly correlated to environmental variables such as drought level and precipitation. %\se{or what did you mean with "correlated with environmental variables"}.%\se{more examples here? Say more about intepretation of latent variables? Maybe general economic indices in the stock market example?}

In the case of extremes, a sparse subgraph $\mathcal{G}_0$ and the presence of only a few latent variables in the model translate to a latent \HR{} model with matrix $\Theta_O$ being sparse, the matrix $\Theta_{OH}\Theta_{H}^{-1}\Theta_{HO}$ being low-rank, and thus the observed precision matrix $\tilde{\Theta}$ being decomposed as a sparse plus low-rank matrix having zero row sums. Notice that the matrix $\tilde{\Theta}$ will generally be dense due to the additional low-rank term $\Theta_{OH}\Theta_{H}^{-1}\Theta_{HO}$, highlighting how the latent variables induce many confounding dependencies among the observed variables (see Figure~\ref{fig:graphs_mot}), and how structure learning procedures that impose sparsity on the precision matrix $\tilde{\Theta}$ will generally not perform well. 

In summary, we can cast the problem of learning a latent \HR{} graphical model as obtaining a sparse plus low-rank decomposition of the precision matrix $\tilde{\Theta}$ of the observed variables. The sparse component provides the residual graphical structure of the observed variables after accounting for the latent variables, the rank of the low-rank component provides the number of latent variables, and the overall sum provides a compact model of the observed variables that can be used for downstream tasks. In the following section, we propose a convex optimization procedure to accurately estimate each of these components from data. 

Finally, we note that in the setting where the observed and latent variables are jointly Gaussian, \cite{Chand2012} also models the precision matrix among the observed variables as a sum of a sparse and a low-rank matrix. Analogous to our setting, the sparse component encodes the subgraph of the observed variables and the low-rank component encodes the effect of the latent variables on the observed variables. An important distinguishing feature with our extremal setting however is that in the Gaussian context, the resulting sum is {not constrained} to have zero row sum. As we describe in Section~\ref{sec:inference}, the additional subspace constraint in our extremal setting results in a different estimation procedure and assumptions for statistical consistency. 

\subsection{Inference for latent \HR{} graphical models}
\label{sec:inference}
Let $X = (X_O,X_H)$ be a collection of observed and latent variables in the domain of attraction of a latent \HR{} graphical model with a sparse subgraph among the observed variables and a small number of latent variables; we will specify the sparsity level and the number of latent variables in our theoretical results. Let $\Gamma^\star$ be the underlying population variogram matrix and $\Theta^\star$ be the population precision matrix with components $\Theta_O^\star,\Theta_{OH}^\star$ and $\Theta_{H}^\star$. Let $\tilde{\Theta}^\star$ be the precision matrix among the observed variables. From Theorem~\ref{thm:main}, we have that $\tilde{\Theta}^\star = S^\star - L^\star$ where $S^\star := \Theta_O^\star$ is a sparse matrix and $L^\star := \Theta_{OH}^\star{\Theta_{H}^\star}^{-1}\Theta_{HO}^\star$ is a low-rank matrix. Here, the support of $S^\star$ encodes the subgraph among the observed variables and the rank of $L^\star$ encodes the number of latent variables.  We will propose a convex optimization procedure to estimate the matrices $(S^\star,L^\star)$ from data. 

\subsubsection{Empirical extremal variogram matrix}
\label{sec:exmpirical_variogram}
An important ingredient of our procedure is an empirical estimate for the extremal variogram matrix $\Gamma^\star_O$. To arrive at our estimate, define for any $m \in O$, the population extremal variogram matrix $\Gamma^{\star,(m)}_O$ rooted at the node $m$; see Definition~\ref{defn:rooted_variogram}. Suppose we have $n$ independent and identically distributed samples $\{X_O^{(t)}\}_{t=1}^n \subseteq \mathbb{R}^p$  of the observed variables $X_O$. Then, a natural estimate $\hat{\Gamma}^{(m)}_O$ for $\Gamma^{\star,(m)}_O$ is given by
\begin{align*}
\hat \Gamma_{ij}^{(m)} &:= \widehat{\Var}\Big(\log (1 - \hat F_i(X_{i}^{(t)})) - \log(1 - \hat F_j(X^{(t)}_{j})) : \hat F_m(X_{m}^{(t)}) \geq 1 - k/n  \Big),\quad i,j \in O.
\end{align*}
Here, $\widehat{\Var}$ denotes the sample variance, and $k$ is the number of extreme samples considered in the conditioning event, which can be viewed as the \emph{effective sample size}. Since in Section~\ref{sec:background} we assumed that $X$ has standard exponential margins, for $i \in O$, $t\in\{1,\dots, n\}$, inside the variance we normalize the $i$-th entry of the $t$-th observation empirically by $-\log (1 - \hat F_i(X^{(t)}_i))$, where $\hat F_i$ denotes the empirical distribution function of $X_{i}^{(1)},\dots,X_{i}^{(n)}$. As \eqref{HR_vario} establishes that the empirical variogram matrix rooted at node $m$ coincides with the true variogram matrix $\Gamma^{\star}_O$ for every $m$, a natural empirical estimator of this matrix is
\begin{align}
    \label{emp_vario} \hat \Gamma_O := \frac1p \sum_{m=1}^p \hat \Gamma^{(m)}_O.
\end{align}
Under the assumption that $k\to \infty$ and $k/n \to 0$, and mild conditions on the underlying data generation, this estimator can be shown to be consistent for $\Gamma^\star_O$ \citep{engelke2022a}. Moreover, \cite{engelke2022b} derive finite sample concentration bounds for $\hat \Gamma_O$ that can be used for high-dimensional consistency results. We refer to Appendix \ref{sec:finite_sample_variogram}
 for details on the assumptions and results.

\subsubsection{Parameter estimation and structure learning} 
For structure learning in \HR{} models, formulating optimization problems in the precision domain leads to computationally efficient procedures. Indeed, the precision matrix estimate obtained from plugging in the empirical extremal variogram $\hat \Gamma_O$ in place of $\Gamma^\star_O$ in the expression $\tilde{\Theta}^\star = (\tilde{\Pi}(-\Gamma_O^\star/2)\tilde{\Pi})^+$ is the minimizer of the convex problem
\begin{equation}
\begin{aligned}
    \hat{\Theta} = \argmin_{\Theta \in \mathbb{S}^{p}} &~~-\log{\det}\left(U^T \Theta U\right) - \frac12 \mathrm{tr}(\Theta\hat{\Gamma}_O),\\
    \text{s.t.}&~~~\Theta \succeq 0~~,~~\Theta\mathbf{1}_p = 0,
\end{aligned}
\label{opt_prob}
\end{equation}
where the matrix $U \in \mathbb{R}^{p \times (p-1)}$ consists of the first $p-1$ left singular vectors of $\tilde{\Pi}$ so that $UU^T = \tilde{\Pi}$; see Appendix \ref{appendix:arriving} for a formal proof. The constraint $\succeq 0$ imposes positive semi-definiteness, $\mathbb{S}^p$ denotes the space of symmetric $p \times p$ matrices, and the constraint $\Theta\mathbf{1}_p=0$ ensures that $\Theta$ has zero row sum. The above optimization problem corresponds to the surrogate maximum likelihood estimator of the \HR{} distribution; for more details on this justification we refer to \cite{roettger2023}.

The formulation in terms of the precision matrix $\Theta$ opens the door to various regularized estimation methods.
\cite{roettger2023} solve \eqref{opt_prob} under the additional constraint that $\Theta_{ij} \leq 0$ for all $i,j\in V$ to ensure a from of positive dependence.
For a graph $\mathcal{G}= (V,E)$, in order to obtain a graph structured estimate of $\Gamma$, \cite{hen2022} solve a matrix completion problem that corresponds to~\eqref{opt_prob} under the constraint $\Theta_{ij} = 0$ for $(i,j)\notin E$. In the context of structure learning without latent variables, \cite{engelke2022b} and \cite{wan2023graphical} add an $\ell_1$ penalty to the loss function akin to the graphical Lasso. 

In the setting with latent variables, we rely on the sparse plus low-rank decomposition described in Section~\ref{sec:latent_hr_graphs}. We, therefore, search over the space of precision matrices $\Theta$ that can be decomposed as $\Theta = S-L$ to identify a sparse matrix $S$ and a low-rank matrix $L$, whose difference has zero row sum and yields a small surrogate negative likelihood. Motivated by the estimator for Gaussian latent variable graphical modeling \citep{Chand2012}, we introduce the \texttt{eglatent} method that solves the following regularized \emph{convex} likelihood problem for some $\lambda_n,\gamma \geq 0$:
\begin{equation}
\begin{aligned}
    (\hat{S},\hat{L}) = \argmin_{S \in \mathbb{S}^{p},L \in \mathbb{S}^{p}} &~~-\log{\det}(U^T(S-L)U) - \mathrm{tr}((S-L)\hat{\Gamma}_O/2) + \lambda_n(\|S\|_{1} + \gamma\mathrm{tr}(L)),\\
    \text{s.t.}&~~~S-L \succeq 0, L \succeq 0, (S-L)\mathbf{1}_p = 0.
\end{aligned}
\label{eqn:estimator}
\end{equation}
Here, $\hat{S}$ and $\hat{L}$ are estimates for the population quantities $S^\star$ and $L^\star$, respectively. The matrix $\hat{S}-\hat{L}$ represents an estimated precision matrix among the observed variables. By the constraints in \eqref{eqn:estimator} and the property of logdet functions, $\text{span}(\ones)$ is the null space of $\hat{S}-\hat{L}$ and $\hat{S}-\hat{L}$ always specifies a valid \HR{} model.

The function $\|\cdot\|_{1}$ denotes the $\ell_1$ norm that promotes sparsity in the matrix $S$ \citep{friedmanEtAl2007}. The role of the trace penalty on $L$ is to promote low-rank structure \citep{Fazel2004RankMA}. The regularization parameter $\gamma$ provides a trade-off between the graphical model component and the latent component. In particular, for very large values of $\gamma$, \texttt{eglatent} produces $\hat{L} = 0$ so that no latent variables are included in the model. As $\gamma$ decreases, the number of latent variables increases and correspondingly the number of edges in the residual graphical structure decreases. The regularization parameter $\lambda_n$ provides overall control of the trade-off between the fidelity of the model to the data and the complexity of the model, and thus naturally depends on the sample size. For $\lambda_n,\gamma \geq 0$, \texttt{eglatent} is a convex program that can be solved efficiently. 

{While our \texttt{eglatent} estimator \eqref{eqn:estimator} resembles the one in \cite{Chand2012}, there is an important distinguishing feature. Specifically, in contrast to the estimator in \cite{Chand2012}, our estimator imposes the constraint $(S-L)\textbf{1}_p = 0$ so that the resulting model is a valid \HR{} model. As a result of this additional constraint, the log-determinant term in our objective is also different in that it projects $S-L$ onto the space of matrices that have zero row/column sum. Since our estimator is different, we need additional assumptions and more involved analysis to establish consistency guarantees; see Section~\ref{sec:consistency} for more details.}

\begin{remark}
A challenge with the optimization in \eqref{opt_prob}, both theoretically and numerically, is the fact that the matrices range in the space of positive semi-definite matrices with zero row sum. This factor indeed seems to prohibit direct structure learning without latent variables (i.e., setting $L = 0$ in \eqref{eqn:estimator} to obtain a graphical lasso analog) where the estimated graphical structure can be rather different than the true graphical structure; see the discussion in \citet[][Section 7]{engelke2022b}. To circumvent this issue, \cite{engelke2022b} and \cite{wan2023graphical} solve slightly different problems to obtain accurate graph estimation, although their estimated graphs do not always yield valid \HR{} models. Remarkably, the addition of the low-rank component $L$ in the \texttt{eglatent} estimator \eqref{eqn:estimator} solves these issues. Indeed, we will show that \texttt{eglatent} consistently recovers the subgraph among the observed variables and the number of latent variables, and matches the performance of existing procedures \citep{engelke2022b,wan2023graphical} for learning an accurate model when no latent variables are present.
\label{remark:no_latent}
\end{remark}

\section{Consistency guarantees for \texttt{eglatent}}
\label{sec:consistency}
Recall from Section~\ref{sec:inference} that we denote by $S^\star$ the population matrix encoding the graphical structure among the observed variables conditioned on the latent variables, and by $L^\star$ the population matrix encoding the effect of a few latent variables on the observed variables. Further, $\tilde{\Theta}^\star = S^\star-L^\star$ represents the marginal precision matrix in the \HR{} model over the observed variables. In this section, we state a theorem to prove that the estimates of \texttt{eglatent} in \eqref{eqn:estimator} provide, with high probability, the correct graphical structure among the observed variables, the correct number of latent variables, and an accurate extremal model. Stated mathematically, we show with high probability that (i) the sign-pattern of $\hat{S}$ is the same as that of $S^\star$, i.e., $\text{sign}(\hat{S}) = \text{sign}(S^\star)$, where $\text{sign}(0)=0$; (ii) the rank of $\hat{L}$ is the same as that of $L^\star$, i.e.,  $\text{rank}(\hat{L}) = \text{rank}(L^\star)$; and (iii) the estimated precision model $\hat{S}-\hat{L}$ closely approximates the true precision matrix $\tilde{\Theta}^\star$, i.e., $\hat{S}-\hat{L} \approx \tilde{\Theta}^\star$. Our analysis requires assumptions on the population model so that the matrices $S^\star$ and $L^\star$ are identifiable from their sum, and that the number of effective samples $k$ is of order $k \gtrsim  p^2\log(p)$. 

\subsection{Technical setup}
\label{sec:technical_setup}
 As \texttt{eglatent} is solved in the precision matrix parameterization, the conditions for our theorems are naturally stated in terms of the precision matrix $S^\star-L^\star$. The assumptions are similar in spirit to convex relaxation methods for Gaussian latent-variable graphical model selection \citep{Chand2012}, {although some conditions are new due to the zero row and column sum structure of the observed precision matrix $S^\star-L^\star$}.

To ensure correct graph recovery and correct number of latent variables, we seek an estimate $(\hat{S},\hat{L})$ from \texttt{eglatent} such that $\text{support}(\hat{S}) = \text{support}(S^\star)$ and $\text{rank}(\hat{L}) = \text{rank}(L^\star)$. Building on both classical statistical estimation theory, as well as the recent literature on high-dimensional statistical inference, a natural set of conditions for accurate parameter estimation, is to assume that the curvature of $S^\star-L^\star$ is bounded in certain directions. The curvature is governed by the modified Hessian of the surrogate log-likelihood loss at $S^\star-L^\star$:
$$ \I := \left(S^\star-L^\star + \frac{1}{p}\ones\right)^{-1} \otimes \left(S^\star-L^\star + \frac{1}{p}\ones\right)^{-1},$$
where $\otimes$ denotes a Kronecker product between matrices, and $\I$ may be viewed as a map from $\mathbb{S}^{p}$ to $\mathbb{S}^p$. The matrix $\I$ modifies the Hessian of the surrogate log-likelihood loss $(S^\star - L^\star)^{+} \otimes (S^\star - L^\star)^{+}$, where the addition of the term $\frac{1}{p}\ones$ (a dual parameter of the program \eqref{eqn:estimator}) helps to compactify the assumptions we place in our population model.

We impose conditions so that $\I$ is well-behaved when applied to matrices of the form ${S}-S^\star-({L}-L^\star+t\ones)$. Here, $S$ is in the neighborhood of $S^\star$ restricted to sparse matrices, $L$ is in the neighborhood of $L^\star$ restricted to low-rank matrices, and $t\ones$ is a dual parameter for some $t \in \mathbb{R}$ due to the constraint $(S-L)\textbf{1}_p = 0$ that appears in the analysis of \eqref{eqn:estimator}. These local properties of $\I$ around $S^\star-(L^\star+t\ones)$ are conveniently stated in terms of tangent spaces to algebraic varieties of sparse and low-rank matrices. In particular, the tangent space of a matrix $M$ with $r$ non-zero entries with respect to the algebraic variety of $p \times p$ matrices with at most $r$ non-zeros is given by
$$\Omega(M) := \{N \in \mathbb{R}^{p \times p}: \text{support}(N) \subseteq \text{support}(M)\}.$$
Moreover, the tangent space at a rank-$r$ matrix $M$ with respect to the algebraic variety of $p \times p$ matrices with rank less than or equal to $r$ is given by:
\begin{eqnarray*}
    \begin{aligned}
    T(M) &:= \{N_R+N_C: N_R,N_C \in \mathbb{R}^{p \times p},\\& \quad\quad  \text{row-space}(N_R)\subseteq \text{row-space}(M), \text{col-space}(N_C)\subseteq \text{col-space}(M)\}.
    \end{aligned}
\end{eqnarray*}
For more discussion on the tangent spaces of sparse and low-rank matrices, see \cite{Chand2012}. In the next section, we describe conditions on the population Hessian $\I$ in terms of tangent spaces $\Omega(S^\star)$ and $T(L^\star)$. Under these conditions, we present a theorem in Section~\ref{sec:theoretical_results} showing that the convex program provides accurate estimates. For notational simplicity, we let $\Omega^\star := \Omega(S^\star)$ and $T^\star:= T(L^\star)$. Finally, the linear operators $\mathcal{A}: \mathbb{S}^p \times \mathbb{S}^p \to \mathbb{S}^p$ and its adjoint $\mathcal{A}^\dagger: \mathbb{S}^p \to \mathbb{S}^p \times \mathbb{S}^p$ are defined as:
\begin{eqnarray}
\mathcal{A}(M,N) := (M-N),~~~~ \mathcal{A}^\dagger(Q) := (Q,Q).
\label{eqn:linear_maps}
\end{eqnarray}
{\subsection{Conditions on the Hessian $\mathbb{I}^\star$}}
\label{sec:fisher_conds}
Given a norm $\|\cdot\|_{\Psi}$ on $\mathbb{S}^p \times \mathbb{S}^p$, we first consider a classical condition in statistical estimation literature, which is to control the minimum gain of the Hessian $\mathbb{I}^\star$ restricted to a subspace $\mathbb{Q} \subseteq \mathbb{S}^p \times \mathbb{S}^p$ as follows:
\begin{eqnarray}
\chi(\mathbb{Q},\|\cdot\|_{\Psi}) := \min_{\substack{Z \in \mathbb{H}\\\|Z\|_{\Psi}=1}} \|\mathcal{P}_{\mathbb{Q}}\mathcal{A}^\dagger\mathbb{I}^\star\mathcal{A}\mathcal{P}_{\mathbb{Q}}(Z)\|_{\Psi},
\label{eqn:chi}
\end{eqnarray}
where $\mathcal{P}_{\mathbb{Q}}$ denotes the projection operator onto the subspace $\mathbb{Q}$ and the linear maps $\mathcal{A}$ and $\mathcal{A}^\dagger$ are defined in \eqref{eqn:linear_maps}. The quantity $\chi(\mathbb{Q},\|\cdot\|_{\Psi})$ insures that the Hessian is well-conditioned restricted to the image $\mathcal{A}\mathbb{Q}$. The remaining condition we impose on $\mathbb{I}^\star$ are in the spirit of irrepresentability-type conditions that are frequently employed in high-dimensional estimation problems \citep{Meinshausen2006HighdimensionalGA,Wainwright2009SharpTF,Zhao2006OnMS,Ravikumar2008HighdimensionalCE,rechtCandesMatrix,Chand2012}. Specifically, we control the inner-product between elements in $\mathcal{A}\mathbb{Q}$ and $\mathcal{A}\mathbb{Q}^\perp$ as quantified by the metric induced by $\mathbb{I}^\star$ via the following quantity:
\begin{equation}
\varphi(\mathbb{Q},\|\cdot\|_{\Psi}) := \max_{\substack{Z \in \mathbb{Q}\\\|Z\|_{\Psi}=1}}\|\mathcal{P}_{\mathbb{Q}^\perp}\mathcal{A}^\dagger\mathbb{I}^\star\mathcal{A}\mathcal{P}_{\mathbb{Q}}(\mathcal{P}_{\mathbb{Q}}\mathcal{A}^\dagger\mathbb{I}^\star\mathcal{A}\mathcal{P}_{\mathbb{Q}})^{-1}(Z)\|_{\Psi}.
\label{eqn:varphi}
\end{equation}
The operator $(\mathcal{P}_{\mathbb{Q}}\mathcal{A}^\dagger\mathbb{I}^\star\mathcal{A}\mathcal{P}_{\mathbb{Q}})^{-1}$ in \eqref{eqn:varphi} is well-defined if $\chi(\mathbb{Q},\|\cdot\|_{\Psi}) > 0$, since this latter condition implies that $\mathbb{I}^\star$ is injective restricted to $\mathcal{A}\mathbb{Q}$. The quantity $\varphi(\mathbb{Q},\|\cdot\|_{\Psi})$ being small implies that any element of $\mathcal{A}\mathbb{Q}$ and any element of $\mathcal{}\mathbb{Q}^\perp$ have a small inner-product (in the metric induced by $\mathbb{I}^\star$). 

A natural approach to controlling the condition of the Hessian $\mathbb{I}^\star$ around $S^\star-L^\star+t\ones$ is to bound the quantities $\chi(\mathbb{Q}^\star,\|\cdot\|_{\Psi})$ and $\varphi(\mathbb{Q}^\star,\|\cdot\|_{\Psi})$ for $\mathbb{Q}^\star = \Omega^\star \times (T^\star\oplus\mathrm{span}(\ones))$. However, a complication that arises with tangent spaces to low-rank varieties is that they are locally smooth. To account for this curvature, we bound distances of nearby tangent spaces via the following induced norm:
$$ \rho(T_1,T_2) :=\max_{\|N\|_2 \leq 1} \|(\proj_{T_1}-\proj_{T_2})(N)\|_2.$$
The quantity $\rho(T_1,T_2)$ measures the sine of the largest angle between $T_1$ and $T_2$. Using this approach for bounding nearby tangent spaces, we consider subspaces $\mathbb{Q}' = \Omega^\star \times (T'\oplus \mathrm{span}(\ones))$ for all $T'$ close to $T^\star$ as measured by $\rho$. For $\omega \in (0,1)$, we bound $\chi(\mathbb{Q}',\|\cdot\|_{\Psi})$ and $\varphi(\mathbb{Q}',\|\cdot\|_{\Psi})$ in the sequel for all subspaces $\mathbb{Q}'$ in the following set:
\begin{eqnarray*}
U(\omega) = \{\Omega^\star \times (T'\oplus \mathrm{span}(\ones)) ~|~ \rho(T',T^\star) \leq \omega\}.
\end{eqnarray*}
We control the quantities $\chi(\mathbb{Q}',\|\cdot\|_{\Psi})$ and $\varphi(\mathbb{Q}',\|\cdot\|_{\Psi})$ using the dual norm of the regularizer $\|S\|_1 + \gamma\mathrm{tr}(L^\star)$:
$$\Phi_\gamma(S,L) := \max\left\{\|S\|_1,\frac{\|L\|_2}{\gamma}\right\}.$$
As the dual norm $\max\{\|S\|_1,\frac{\|L\|_2}{\gamma}\}$ plays a central role in the optimality conditions of \eqref{eqn:estimator}, controlling the quantities $\chi(\mathbb{Q}',\|\cdot\|_{\Phi_\gamma})$ and $\varphi(\mathbb{Q}',\|\cdot\|_{\Phi_\gamma})$ leads to a natural set of conditions that guarantee the consistency of the estimates produced by \eqref{eqn:estimator}. In summary, given a fixed set of parameters $(\omega,\gamma) \in (0,1) \times \mathbb{R}_{+}$, we assume that $\mathbb{I}^\star$ satisfies the following conditions, where $F = \proj_{{T^\star}^\perp}(1/p\ones)/\|\proj_{{T^\star}^\perp}(1/p\ones)\|_2$ and $\|\mathbb{I}^\star\|_2$ denotes the spectral norm of the operator $\mathbb{I}^\star$. 

\begin{assumption} $\inf_{\mathbb{Q}'\in U(\omega)} \chi(\mathbb{Q}',\Phi_\gamma) \geq \alpha$ for some $\alpha> 8\omega\max\{\gamma,1\}(\|\mathbb{I}^\star(F)\|_2+\|\mathbb{I}^\star\|_2\omega+1)$.
\label{ass:1}
\end{assumption}
\begin{assumption} $\sup_{\mathbb{Q}'\in U(\omega)} \varphi(\mathbb{Q}',\Phi_\gamma) \leq 1-\nu$ for some $\nu \in [4\omega,1)$.
\label{ass:2}
\end{assumption}
\cite{Chand2012} impose a sufficient set of conditions, and prove that they imply conditions similar to Assumptions~\ref{ass:1}-\ref{ass:2} (see Proposition 3.3 in \cite{Chand2012}). A key distinction between our conditions and the implied conditions in \cite{Chand2012} is that our subspace $\mathbb{Q}'$ also contains the directions $\mathrm{span}(\ones)$. This distinction arises from the additional zero row-sum constraint in our estimator which introduces the dual parameter $t\ones$. Moreover, we require the following condition for how far $\mathrm{span}(\ones)$ deviates from $T^\star$:
\begin{assumption} $\kappa^\star:= \|\proj_{{T^\star}^\perp}(\ones/p)\|_2 \in \left(\omega,\min\left\{4\nu, \frac{\alpha}{8\max\{\gamma,1\}(\|\mathbb{I}^\star(F)\|_2+\|\mathbb{I}^\star\|_2\omega+1)}-\omega\right\}\right)$.
\label{ass:3}
\end{assumption}
Assumption~\ref{ass:3} is also a new condition relative to \cite{Chand2012}. This assumption ensures that $k^\star$ not so small so that $L^\star$ and the dual parameter $t\ones$ can be distinguished from one another. Assumption~\ref{ass:3} also ensures that $\kappa^\star$ is not too large. This condition comes from the optimality conditions of \eqref{eqn:estimator}, which involve controlling the size of the inner product of elements in $\text{span}(\ones)$ and in ${T^\star}^\perp$. Further, bounding $\kappa^\star$ allows the size of $t$ to be controlled. 

\begin{remark}[Dependency on $h$ and graph structure] The dependence on the number of latent variables $h$ and the density of the graphical structure among the observed variables conditioned on the latent variables does not appear explicitly in Assumptions~\ref{ass:1}--\ref{ass:3}, but is implicit in the quantities $\alpha,\nu$. Indeed, as larger $h$ and denser graphical structures increase the dimensions of the tangent spaces $T^\star$ and $\Omega^\star$, respectively, they result in smaller $\alpha,\nu$. In Appendix \ref{sufficient_hessian}, we provide conditions on the Hessian $\mathbb{I}^\star$ that do not depend on $\gamma$ and measure the behavior of $\mathbb{I}^\star$ restricted to individual subspaces $\Omega^\star$ and $T^\star$ (rather than their coupling as in Assumptions 1--2). With these conditions and when the latent variables affect most of the observed variables (see also the discussion in Section \ref{section:numerical_assumptions}), we prove in Appendix~\ref{sufficient_hessian} that as long as $d^\star\sqrt{{h}/{p}} = \mathcal{O}(1)$, there exists a choice of $\gamma$ that satisfies Assumptions~\ref{ass:1}-\ref{ass:3}. 
{Here, 
\begin{align}\label{d_star}
    d^\star := \max_{i}\sum_{j}\mathbb{I}[S^\star_{ij} \neq 0]
\end{align}
is the maximum degree of the graphical structure among the observed variables.} For instance, for the following  nontrivial classes of models the above condition holds:
\begin{itemize}
    \item {Polynomial degree}: the maximum degree $d^\star$ grows at most polynomially with $p$, that is, $d^\star = \mathcal{O}(p^q)$, and the number of latent variables satisfies $h = \mathcal{O}\left({p}^{1-q}\right),$ where $q \in (0,1)$. Here, consistent estimation is possible even when the graph structure is complex.
    \item {Bounded degree}: we have $d^\star = \mathcal{O}(1)$ so that $h = \mathcal{O}(p).$ Here again, consistent estimation of the underlying graphical structure among the observed variables is possible even when the number of latent variables is in the same order as the number of observed variables. 
\end{itemize}

\end{remark}
\begin{remark}[Choice of $\gamma$] We make two observations. First, a smaller range of values of $\gamma$ naturally leads to larger $\alpha$ and $\nu$. Second, intuitively, the choice of $\gamma$ should decrease with a larger $h$ so that less penalty is imposed on the rank of $\hat{L}$, and it should increase with larger $d^\star$ so that $\hat{L}$ does not contain some of the components of $S^\star$. To formalize this intuition, we consider the setting described in the previous paragraph. We show in Appendix~\ref{sufficient_hessian} that the lower-bound on the range of values of $\gamma$ that satisfy Assumptions 1--3 scales with $d^\star$, and the upper-bound is in the order $\sqrt{{p}/{h}}$. 
% Furthermore, recalling that $\inc^\star \geq \sqrt{h/p}$, and that the upper-bound for $\gamma \sim 1/\inc^\star$ (according to Theorem~\ref{thm:main}), as expected, the upper-bound on appropriate values of $\gamma$ decreases when $h$ increases.
\label{remark:choice_gamma}
\end{remark}

\subsection{When do Assumptions~\ref{ass:1}--\ref{ass:3} on the Hessian hold? Connections to identifiability}
\label{section:numerical_assumptions}
In this section, we provide concrete examples of latent extremal models that satisfy Assumptions~\ref{ass:1}--\ref{ass:3} for some choices of $\alpha > 0$, $\nu \in [0,1) $, $\omega \in (0,1)$ and $\gamma > 0$. To arrive at such models, we must intuitively understand when the matrices $S^\star$, $L^\star$, and $t\ones$ are identifiable from their sum for some $t \in \mathbb{R}$ (recall, the term $t\ones$ arises from the zero row-sum constraint). Since the matrix $t\ones$ has rank equal to one and is thus also low-rank, we consider a combined term $L^\star+t\ones$. Two identifiability issues arise: the first is to distinguish $S^\star$ from $L^\star+t\ones$ and the second is to distinguish $L^\star$ from $L^\star+t\ones$. 

To address the first identifiability issue, we appeal to the previous literature on sparse-plus-low rank decompositions, which states that the matrices $S^\star$ and $L^\star+t\ones$ are identifiable from their sum if the row and columns of the matrix $S^\star$ are sufficiently sparse and the matrix $L^\star+t\ones$ is sufficiently low-rank with most of its entries non-zero and similar in magnitude \citep{Candes11RobustPCA,Chandrasekaran2011RankSparsityIF,Recht2010GuaranteedMS}. {Sparsity of $S^\star$ corresponds to small $d^\star$ in~\eqref{d_star} so that no observed variable is directly connected to ``many'' other observed variables.}
Since the matrix $t\ones$ has equal entries and is rank one, the structural constraint on $L^\star+t\ones$ can be interpreted as the number of latent variables being small (as compared to the ambient dimension {$p$}) with their effects spread across all the observed variables. To measure the ``diffuseness" of the latent effects, we consider the following quantity for any linear subspace $Z \subseteq \mathbb{R}^p$ \citep{Candes11RobustPCA,rechtCandesMatrix,Chandrasekaran2011RankSparsityIF,Chand2012}: $\mu[Z] :=\max_i \|\mathcal{P}_{Z}({\mathbf{e}}_i)\|_2$, 
where $\mathcal{P}_{Z}$ is the projection onto the subspace $Z$ and $\mathbf{e}_i$ is a standard coordinate basis. The quantity $\inc[Z]$ is also known as the ``incoherence parameter'' \citep{rechtCandesMatrix,Chandrasekaran2011RankSparsityIF}. It measures how aligned the subspace $Z$ is with respect to standard basis elements and is lower-bounded by $\sqrt{\text{dim}(Z)/p}$ and upper-bounded by one. In our setting, the relevant subspace is the row or column space of $L^\star$ and {so we define 
\begin{align}\label{mu_star}
    \inc^\star:=\mu[\text{col-space}(L^\star)].
\end{align}}
Thus, a lower bound for $\inc^\star$ is $\sqrt{{h}/{p}}$, which is achieved when the effect of latent variables on the observed variables is equally spread out. A small value of $\inc^\star$ ensures the matrix $L^\star$ has a small rank and is far from being sparse. 

To address the second identifiability issue of disentangling  $L^\star$ and $t\ones$ from their sum, as described earlier, we want the deviation of the subspaces $T^\star$ and $\mathrm{span}(\ones)$ to not be too large (i.e., the lower-bound condition in Assumption~\ref{ass:3}). This deviation can be conveniently measured by $\kappa^\star:-\|\proj_{{T^\star}^\perp}(\ones/p)\|_2$ which is equivalent to $\|\proj_{\text{col-space}(L^\star)^\perp}({1}/{\sqrt{p}}\textbf{1}_p)\|_F^2$.

Having these identifiability concerns in mind, we give stylized extremal graphical models and numerically check that the Hessian conditions in Assumptions~\ref{ass:1}--\ref{ass:3} are satisfied for appropriate choice of parameters. Specifically, we set $p = 30$, $h = 1$ and specify the sub-graph $\mathcal{G}_O = (E_O,O)$ among the observed variables to be an Erd\H{o}s--R{\"e}nyi graph with edge probability $\tau \in \{0.001,0.005\}$ and set $\Theta_{ij}^\star$ to $0.2$ for every $(i,j) \in E_O$ and zero otherwise. We connect the latent variable to each observed variable and select the corresponding entries $\Theta^\star_{p+1,j}$ uniformly at random from the interval $[1/\sqrt{j},1.1/\sqrt{j}]$ for all $j \in O$. Notice that larger values of $\tau$ lead to larger sparsity parameter $d^\star$. We let $\omega = 0.003$ so that the largest angle between tangent spaces $T'$ and $T^\star$ is less than $0.0005$ degrees. Employing a numerical procedure described in Appendix~\ref{sec:numerical_approach}, we obtain a range of values of $\gamma,\alpha,\nu$ that satisfy Assumptions~\ref{ass:1}--\ref{ass:3}. The values of $\alpha$ and $\nu$ that are computed using this procedure serve as a lower bound for the optimal $\alpha,\nu$, respectively. Indeed, an exciting direction for future research is to develop numerical or analytical techniques to precisely characterize the optimal values of $\alpha,\nu$. Table~\ref{conditions} illustrates $d^\star$ and the corresponding values of $\gamma,\alpha,\nu$ that satisfy Assumptions~\ref{ass:1}--\ref{ass:3}. Examining Table~\ref{conditions}, we can make two observations. First, for each value of $\tau$, a larger range of $\gamma$ results in smaller $\alpha$ and $\nu$. Second, larger graph density (i.e., larger $\tau$) reduces the range of values of $\gamma$ that satisfy Assumptions~\ref{ass:1}-\ref{ass:3}. These two observations are consistent with theory; see Remark~\ref{remark:choice_gamma}.

\FloatBarrier
\begin{table}[ht]
\centering
\begin{tabular}{c c c c c } % centered columns (4 columns)
%inserts double horizontal lines
$\tau$ & $d^\star$ & $\gamma$ & $\alpha \geq $ & $\nu\leq $  \\[0.2ex]
\hline % inserts single horizontal line
0.001&1&(1.7,3.6)& 0.91 & 0.004 \\[0.1ex]
0.001&1&(2.15,3)& 1.14 & 0.150\\[0.1ex]
0.005&2&(2.8,3.45)& 1.26 & 0.007 \\ [0.1ex]% inserting body of the table
0.005&2&(3,3.3)& 1.3 & 0.04 \\ [0.1ex] 
\hline %inserts single line
\end{tabular}
\caption{{Different values of the edge probability $\tau$, the maximum node degree~\eqref{d_star}, and the corresponding ranges of the regularization $\gamma$ in~\eqref{eqn:estimator} and values of $\alpha$, $\nu$ that satisfy Assumptions~\ref{ass:1}--\ref{ass:3}.}}
\label{conditions}
\end{table}
\FloatBarrier

%\newpage
%In addition to the quantities $(\mu^\star,d^\star,\kappa^\star)$, we must control the behavior of the Hessian $\I$ restricted to elements in the direct sums  $\Omega^\star\oplus T^\star$ as well as $\Omega^\star\oplus (T^\star \oplus \mathrm{span}(\ones))$. A more interpretable set of assumptions is to measure the behavior of $\I$ restricted to the spaces $\Omega^\star$ and $T^\star$ \emph{separately} and use the bound on the quantities $(\mu^\star,d^\star,\kappa^\star)$ to ``couple" them to control the behavior of $\I$ restricted to the combined spaces. We next present the decoupled conditions and describe in Appendix \ref{sec:coupled_conds} how they combine.

\subsection{{Theorem statement}}
\label{sec:theoretical_results}
We now describe the performance of \texttt{eglatent} under suitable conditions on the quantities from the previous section. We state the theorem based on essential aspects of the conditions required for the success of our convex relaxation (i.e., the Hessian conditions) and omit complicated constants. We specify these constants in Appendix \ref{sec:proof_thm2}. Our results depend on a second-order parameter $\xi > 0$ that determines the rate of convergence of a random vector $X$ in the domain of attraction of a \HR{} distribution to its limit, with larger values corresponding to faster convergence; see Appendix~\ref{sec:finite_sample_variogram}.

\begin{theorem}
Suppose that we have $n$ independent and identically distributed samples in the domain of attraction of a latent \HR{} model as described in Section~\ref{sec:inference} with second-order parameter $\xi >0$ in  Assumption~\ref{assumption2} in Appendix~\ref{sec:finite_sample_variogram}. Assume that there exists $\alpha >0$, $\nu \in (0,1]$, $\omega\in (0,1)$ and the choice of the parameter $\gamma$ so that the Hessian $\mathbb{I}^\star$ corresponding to this latent \HR{} model satisfies Assumptions~\ref{ass:1}--\ref{ass:3}. Let $m := \max\{1,{1}/{\gamma}\}$ and $\bar{m} := \max\{1,\gamma\}$. Let the effective sample size $k$ be chosen such that $k < n^{2\xi/(2\xi+1)}$. Let $h := \mathrm{rank}(L^\star)$ be the true number of latent variables and {$d^\star$ the maximum degree of the true graph structure among the observed variables conditioned on the latent variables as in~\eqref{d_star}}. Suppose:
\begin{enumerate}
    \item $k \gtrsim \frac{m^5h{d^\star}^2}{\alpha^6}{p}^2\log({p})$, i.e.,  effective sample size is sufficiently large; 
    \item $\lambda_n \sim \frac{m}{\nu}\sqrt{\frac{p^2\log(p)}{k}}$, i.e.,  $\lambda_n$ is appropriately chosen; 
    \item $\sigma_\mathrm{min}(L^\star) \gtrsim \frac{m^4\bar{m}h}{\nu\alpha^4}\sqrt{\frac{p^2\log(p)}{k}}$, i.e., the minimum nonzero singular value of $L^\star$ is sufficiently bounded away from zero; 
    \item $|S^\star_{ij}| \gtrsim \frac{m^3\bar{m}\sqrt{h}}{\nu\alpha^2}\sqrt{\frac{p^2\log(p)}{k}}$ for every $(i,j) \text{ with }$ $|S^\star_{ij}|>0$, i.e., the minimum nonzero entry of $S^\star$ is sufficiently bounded away from zero. 
\end{enumerate}
Then, the estimate $(\hat{S},\hat{L})$ defined as the unique minimizer of \texttt{eglatent} in~\eqref{eqn:estimator} with empirical variogram in~\eqref{emp_vario} satisfies
$$ \mathbb P\left(\mathrm{sign}(\hat{S}) = \mathrm{sign}(S^\star), \mathrm{rank}(\hat{L}) = \mathrm{rank}(L^\star), \|(\hat{S}-\hat{L})-\tilde{\Theta}^\star\|_2 \lesssim \frac{m^3\sqrt{h}}{\nu\alpha^2}\sqrt{\frac{p^2\log(p)}{k}} \right) \geq 1-\frac{1}{p}.$$
\label{thm:main}
\end{theorem}

\begin{remark}\label{rem:max_stable}
    The class of distributions $X$ in the domain of attraction of a \HR{} distribution is very large. For instance, the max-stable \HR{} distribution is one member of this class. Note that since we are considering threshold exceedances, the right-hand side of~\eqref{defn:mpd} changes with the threshold $u$ even if $X$ is a max-stable distribution. Indeed, \citet[][Proposition S.6]{engelke2022b} showed that the rate of convergence is governed by a second-order parameter $\xi$ that can be chosen as any value in $(0,1)$. In this case, the effective sample size $k$ in Theorem~\ref{thm:main} must satisfy $k = o(n^{0.66})$. Thus, in our simulations with max-stable distribution, we use $k = n^{0.65}$.    
\end{remark}

We prove Theorem~\ref{thm:main} in Appendix~\ref{sec:proof_thm2}. Due to the zero row-sum constraint in the \texttt{eglatent} estimator \eqref{eqn:estimator}, the proof of Theorem~\ref{thm:main} is more involved than the consistency analysis in \cite{Chand2012}. Specifically, we need additional technical arguments to deal with the dual parameter $t\ones$ that arises from the zero row-sum constraint. We highlight these technical arguments in Appendix~\ref{sec:lemmas_zero_row_sum_constraint}.

 {Theorem~\ref{thm:main} essentially states that if Assumptions 1--3 hold, $(\lambda,\gamma)$ are chosen appropriately, the effective sample size $k$ is sufficiently large, the minimum nonzero singular value of the low-rank term $L^\star$ and the minimum nonzero entry of the sparse piece $S^\star$ are bounded away from zero, then, with high probability, \texttt{eglatent} provides accurate estimates for the subgraph among the observed variables, the number of latent variables, and a marginal extremal model.}

The quantities $(\alpha,\nu,\omega)$ as well as the choices of the parameters $\lambda_n$ and $\gamma$ play a prominent role in the result. Indeed, larger values of $\alpha,\omega,\nu$ lead to a better conditioned Hessian $\I$ around the tangent spaces $\Omega^\star$ and $T^\star \oplus \mathrm{span}(\ones)$. The better conditioning of the Hessian $\I$ then results in less stringent requirements on sample complexity, the minimum nonzero singular value of $L^\star$, and the magnitude of the minimum nonzero entry of $S^\star$. Notice that the complexity of the true graph structure among the observed variables $d^\star$ and the true number of latent variables $h$ appears explicitly in the bounds in Theorem~\ref{thm:main}. We also note the dependence
on $d^\star$ and $h$ is implicit in the dependence on $\alpha,\nu$ and $\gamma$. Indeed, as larger $d^\star$ and $h$ increases the dimension of the tangent space $\Omega^\star$ and $T^\star$, respectively,  they result in smaller $\alpha,\nu$. Furthermore, as described in Remark~\ref{remark:choice_gamma}
, the range of values of $\gamma$ decrease with larger graph complexity and number of latent variables.

{
\begin{remark} \cite{engelke2022b} prove that $k \geq \mathcal{O}(\log(p))$ suffices for consistent estimation of extremal graphical models without latent variables. Further, \cite{Chand2012} prove that $k \geq \mathcal{O}(p)$ suffices for consistent estimation of Gaussian latent variable graphical model. According to Theorem~\ref{thm:main}, we require $k \geq \mathcal{O}(p^2\log(p))$ in our setting. This requirement is determined by the deviation $\|\hat{\Gamma}_O - \Gamma^\star_O\|_2$, namely how fast the empirical variogram matrix $\hat{\Gamma}_O$ converges in \emph{spectral norm} to the true variogram matrix $\Gamma^\star_O$. \cite{engelke2022b} carried out extensive mathematical arguments to obtain the following concentration of the empirical variogram matrix in $\ell_\infty$ norm $\|\hat{\Gamma}_O - \Gamma^\star_O\|_\infty \leq \mathcal{O}(\sqrt{\log(p)/k})$. In our analysis, we use this result and the equivalence of norms relation 
\[\|\hat{\Gamma}_O - \Gamma^\star_O\|_2 \leq p\|\hat{\Gamma}_O - \Gamma^\star_O\|_\infty \leq \mathcal{O}(\sqrt{p^2\log(p)/k})\]
to obtain a convergence rate in the spectral norm. We suspect that a tighter convergence result of $\|\hat{\Gamma}_O - \Gamma^\star_O\|_2 \leq \mathcal{O}(\sqrt{p/k})$ holds. Such a tighter convergence guarantee would then imply $k \geq \mathcal{O}(p)$ is sufficient for consistency guarantees of our estimator. 

Finally, we should expect a more stringent sample size requirement for the latent extremal model than for the extremal model without latent variables. In particular, a larger sample size allows us to guarantee spectral norm consistency of the low-rank component, ensuring accurate estimates for the number of latent variables and their effects; see also the discussion in \cite{Chand2012}.

\end{remark}}

\section{Experimental demonstrations}
\label{sec:experiments}
In our numerical experiments, we use \texttt{eglatent} as a model selection procedure and perform a second refitting step on the selected model structure to estimate the model parameters; see Appendix~\ref{sec:refit} for details. Code to reproduce our results can be found at \url{https://github.com/sebastian-engelke/extremal_latent_learning}.

\subsection{Synthetic simulations}
We illustrate the utility of our method for recovering the subgraph among the observed variables and the number of latent variables on synthetic data. We compare the performance of our \texttt{eglatent} method to \texttt{eglearn} by \cite{engelke2022b} for learning extremal graphical models. {(In Appendix~\ref{sec:additional_exp_gaussian}, we provide comparisons with the Gaussian latent variable estimator in \cite{Chand2012}. As expected, our estimator is better at capturing dependency structure in the extremes and outperforms the Gaussian estimator.)} To evaluate the accuracy of the estimated graphs with edges $\hat E$ relative to the true subgraph among the observed variables with edges $E = E_O$, we use the $F$-score 
\begin{align*}\label{Fscore}
  F = \frac{|E \cap \hat E| }{|E \cap \hat E| + \frac12 (|E^c \cap \hat E| + |E \cap \hat E^c|)}.
\end{align*}
Larger $F$-scores thus indicate more accurate graph recovery.
\begin{figure}
    \centering {\includegraphics[width = 1\textwidth]{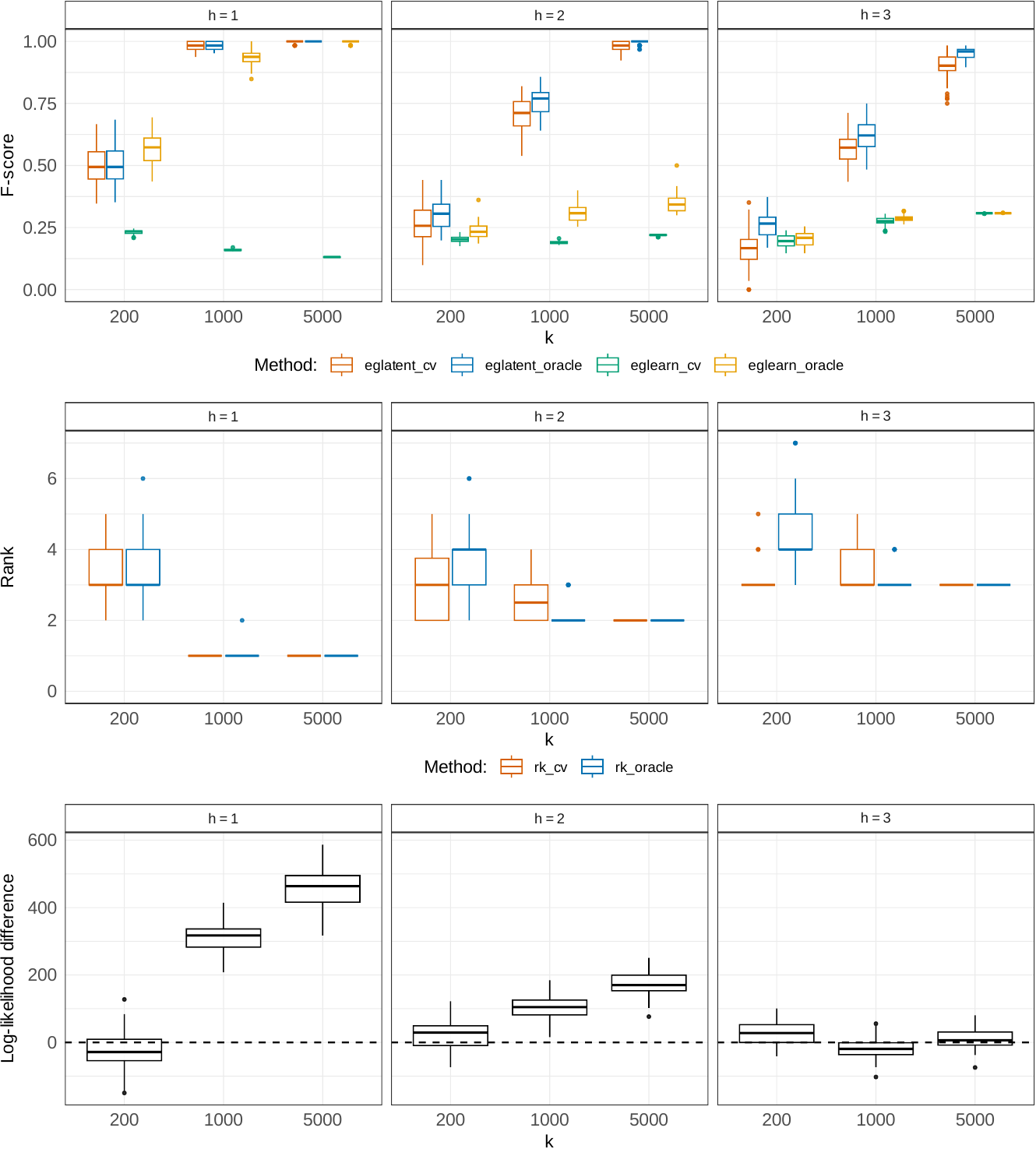}}
    \caption{$F$-score (top row) and estimated number of latent variables (middle row) of \texttt{eglatent} method with the selection of the tuning parameter based on the oracle and validation on the $F$-score for the cycle graph with $h=1,2,3$ latent variables and different effective sample sizes $k=200,1000,5000$. The bottom row shows the difference between best \texttt{eglatent} and best \texttt{eglearn} log-likelihoods on the validation set.}
    \label{fig:ads2}
\end{figure}

\subsubsection{Structure recovery}\label{sec:struc_recov}
In order to evaluate the performance of our new method, we generate data from a random vector $X = (X_O, X_H)$ in the domain of attraction of a latent \HR{} multivariate Pareto distribution $Y$ with the precision matrix $\Theta^\star \in \mathbb{R}^{p+h \times p+h}$, $p$ observed variables $O = \{1,2,\dots,p\}$ and $h$ latent variables $H = \{p+1,\dots,p+h\}$.
We choose to simulate $X$ from the \HR{} max-stable distribution with the same precision matrix $\Theta^\star$, which is well-known to be in the domain of attraction of $Y$; see \cite{res2008} for details. The simulation can be done efficiently with the method in \cite{dombry2016}.
 
 We specify the sub-graph $\mathcal{G}_O = (E_O,O)$ among the observed variables to be a cycle graph and set $\Theta_{ij}^\star$ to $-2$ for every $(i,j) \in E_O$ and zero otherwise. The latent variables are not connected in the joint graph, so $\Theta^\star_{ij} = 0$ for every $i,j \in H, i\neq j$. We connect each latent variable node $i \in H$ to every $k \in O$ satisfying $k = i-(p+1)+\zeta{h}$ for some positive integer $\zeta$ (thus every latent variable is connected to a distinct set of observed variables in the graph). The corresponding entries $\Theta^\star_{ik}$ in the precision matrix are chosen uniformly at random from the interval $[50/\sqrt{p+h},75/\sqrt{p+h}]$. Finally, we set the diagonal entries of $\Theta^\star$ to have the all-ones vector in its null space. Appendix \ref{sec:additional_exp_different_graph_structure} shows results for a setting where the subgraph among the observed variables is generated according to an Erd\H{o}s--R\'enyi graph.

We let $p = 30$, $h \in\{1,2,3\}$, and we set the number of marginal exceedances to $k = \lfloor n^{0.65} \rfloor$. Following Remark~\ref{rem:max_stable}, this choice satisfies the assumptions of Theorem~\ref{thm:main} since we simulate from a max-stable distribution. Altering $k$ in a reasonable range does not change the qualitative results of the simulation study. A more detailed discussion of the choice of $k$ can be found in the real data application in Section~\ref{sec:application}.
We generate $n$ samples from the max-stable \HR{} distribution parameterized by $\Theta^\star$ so that we obtain $k \in \{200,1000,5000\}$ effective extreme samples. {When deploying our \texttt{eglatent} estimator in \eqref{eqn:estimator}, we fix $\gamma = 4$ to a reasonable default value; In Appendix~\ref{sec:additional_exp_gamma}, we demonstrate the robustness of our results to different values of $\gamma$}. Concerning the regularization parameter $\lambda_n$, which also appears in the \texttt{eglearn} method, in both methods, it is chosen either by validation likelihood on a separate dataset of size $n$ or by an oracle approach maximizing the $F$-score for the sub-graph among observed variables. 

Figure \ref{fig:ads2} summarizes the performance of the methods on 50 independent trials for the different sample sizes and different numbers of latent variables. We observe that our proposed approach outperforms \texttt{eglearn} in several ways. Indeed, the top row shows that the graph learned by $\texttt{eglearn}$ only poorly recovers the graphical structure among observed variables. This reveals a limitation of this method, namely that in the presence of latent variables, the marginal graph of observed variables is dense and sparsity cannot be well detected by methods that ignore this fact. Clearly, this problem becomes more pronounced with a larger number of latent variables. On the other hand, our new \texttt{eglatent} method exploits the latent structure for learning the sparse graph among the observed variables conditional on the latent variables. It recovers the graphical structure among the observed variables increasingly well with a growing sample size. In fact, the results for the tuning parameter $\lambda_n$ chosen through validation likelihood are almost as good as those based on the oracle. The middle row of Figure \ref{fig:ads2} shows that \texttt{eglatent} is able to identify the correct number of latent variables, especially for larger sample sizes. 

We can also compare the model in terms of their likelihood on the validation data. Again, our \texttt{eglatent} method generally attains a better validation likelihood and is thus more representative of the data. As an exception, we observe that if the effective sample size is small ($k=100$), then \texttt{eglearn} performs better. The reason is that \texttt{eglatent} is a more flexible model with more parameters to learn, and it therefore benefits more from additional data.

\subsubsection{Robustness to zero latent variables}\label{sec:robust} We now evaluate the performance of \texttt{eglatent} when there are no latent variables present and compare its performance to \texttt{eglearn}. We first specify a graph structure using a Barab\'asi--Albert model denoted by $\mathrm{BA}(d,m)$, which is a preferential attachment model with $d$ notes and a degree parameter $m$ \citep{Albert2001StatisticalMO}. We set $d = 20$ and $m = 2$. We then define a \HR{} precision matrix $\Theta^\star \in \mathbb{R}^{d \times d}$ with entries sampled uniformly at random from the interval $[-5,-2]$. The diagonal entries of $\Theta^\star$ are chosen so that it has the all-ones vector in its null space. We generate $n$ samples from the max-stable \HR{} distribution parameterized by $\Theta^\star$ such that there are $k = \lfloor n^{0.65} \rfloor = 200$ effective marginal extreme samples. We also generate a separate dataset of size $n$ for validation. For the method \texttt{eglatent}, for each value of the regularization parameter $\gamma= 1,4,8,20$ the regularization parameter $\lambda_n$ is chosen based on the validation set. The regularization parameter $\lambda_n$ in \texttt{eglearn} is chosen similarly.

Figure~\ref{fig:zero_latent} presents the $F$-scores and validation log-likelihood scores of \texttt{eglatent} as $\gamma$ varies and for  50 independent trials. We also display the average numbers of edges and latent variables, as well as the performance of \texttt{eglearn}. As expected, larger values of $\gamma$ lead to smaller estimates for the number of latent variables. We observe that when $\gamma = 4$, \texttt{eglatent} obtains an accurate graphical structure ($F$-score close to one) with a similar validation likelihood as \texttt{eglearn}. Here, \texttt{eglearn} yields a sparse graph since, unlike the previous settings, there are no unobserved confounding. Interestingly, the average number of estimated latent variables in this case is not close to zero. In particular, we observe that when $\gamma$ is chosen so that \texttt{eglatent} yields nearly zero latent variables (i.e., $\hat{L} \approx 0$), the $F$-scores scores obtained by \texttt{eglatent} drop significantly. For such $\gamma$, our estimator \eqref{eqn:estimator} resembles the analog of the graphical lasso which is known to yield inaccurate models \citep{engelke2022b}; see also Remark~\ref{remark:no_latent}.  

In summary, when the sample size is sufficiently large, \texttt{eglatent} yields a similar model fit and graph recovery as \texttt{eglearn} even when there are no latent variables. It is worth emphasizing that \texttt{eglatent} achieves this favorable performance by estimating some latent variables. This shows the robustness of our method to model misspecification. %{{In Appendix~\ref{sec:additional_zero_latent_vars}, We show similar behaviors for the smaller effective sample size of $k = 200$.}}

\begin{figure}
    \centering
    \includegraphics[width = .8\textwidth]{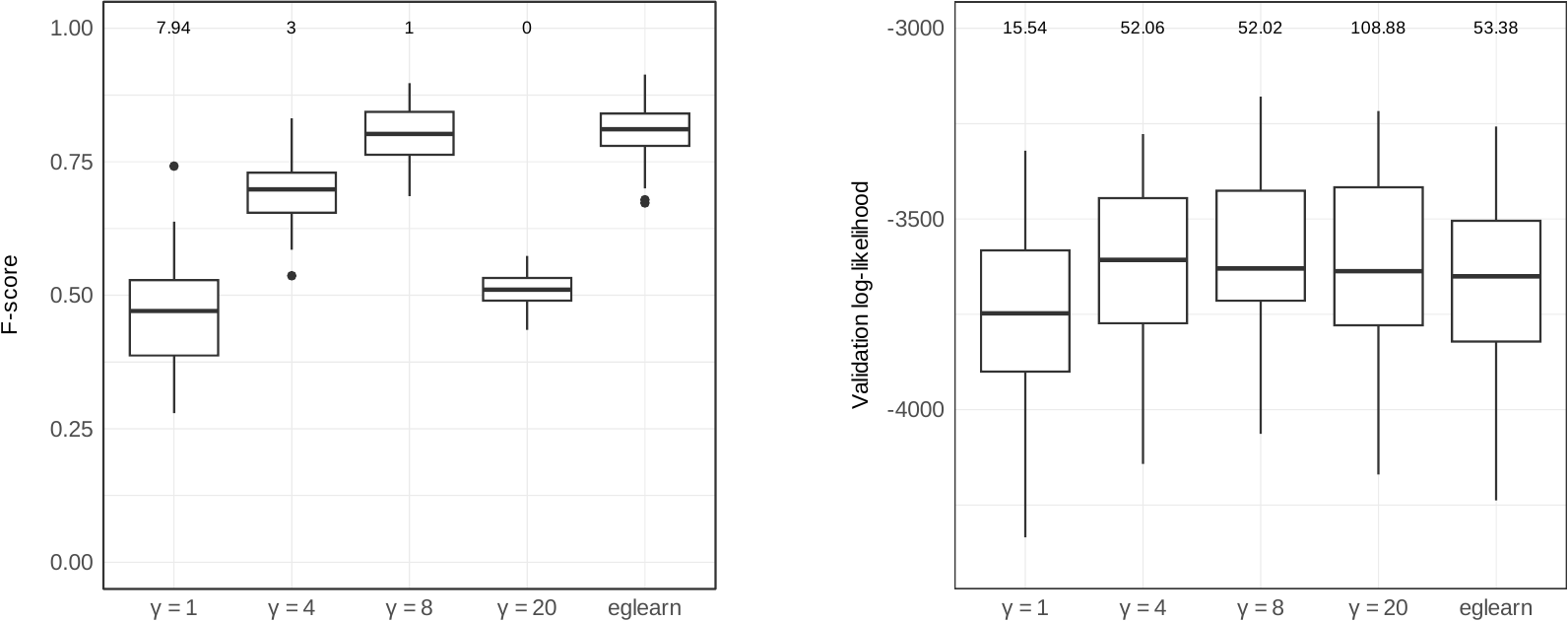}
    \caption{Left: $F$-score of \texttt{eglatent} for different regularization parameters $\gamma \in \{1,4,8,20\}$ and \texttt{eglearn}; top axis shows the average number of estimated latent variables in \texttt{eglatent}. Right: the log-likelihood of the same methods evaluated on a validation data set; the top axis shows the average number of estimated edges in each model.}
    \label{fig:zero_latent}
\end{figure}

\subsection{Real data application}\label{sec:application}

We apply our latent \HR{} model to analyze large flight delays. We use a data set from the R package \texttt{graphicalExtremes} \citep{graphicalExtremes2022} with $p = 29$ airports in the southern U.S.~shown in the left panel of Figure~\ref{fig:graphs}.
Large flight delays cause huge financial losses and lead to congestion of critical airport infrastructure. Our method provides an improved model for the dependence of such excessive delays at different airports, and can eventually be used for stress testing of the system; see \cite{hen2022} for details on this application. 
Unless otherwise noted, we fit the models in the whole dataset consisting of $n = 3603$ observations from 2005-01-01 to 2020-12-31.
We compare our \texttt{eglatent} method for latent \HR{} models with the \texttt{eglearn} algorithm by \citep{engelke2022a} that estimates a graphical structure without latent variables. 
We report here the results for the exceedance threshold of be $q = 0.90$ (i.e., $1-k/n = 0.90$) resulting in $k = 360$ marginal exceedances for the computation of the empirical variogram $\hat \Gamma_O$; see Section~\ref{sec:exmpirical_variogram}. The latter is the input for the different structure learning methods. Different choices of the threshold, or equivalently, of $k$, are discussed below.

\begin{figure}
    \centering
    \includegraphics[width = .32\textwidth]{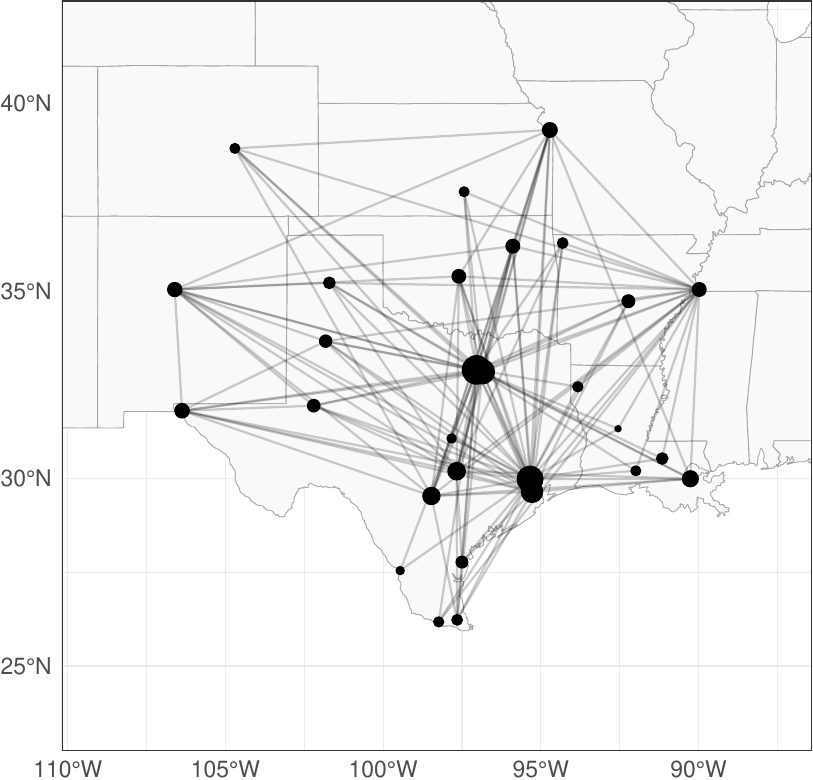}
    \includegraphics[width = .32\textwidth]{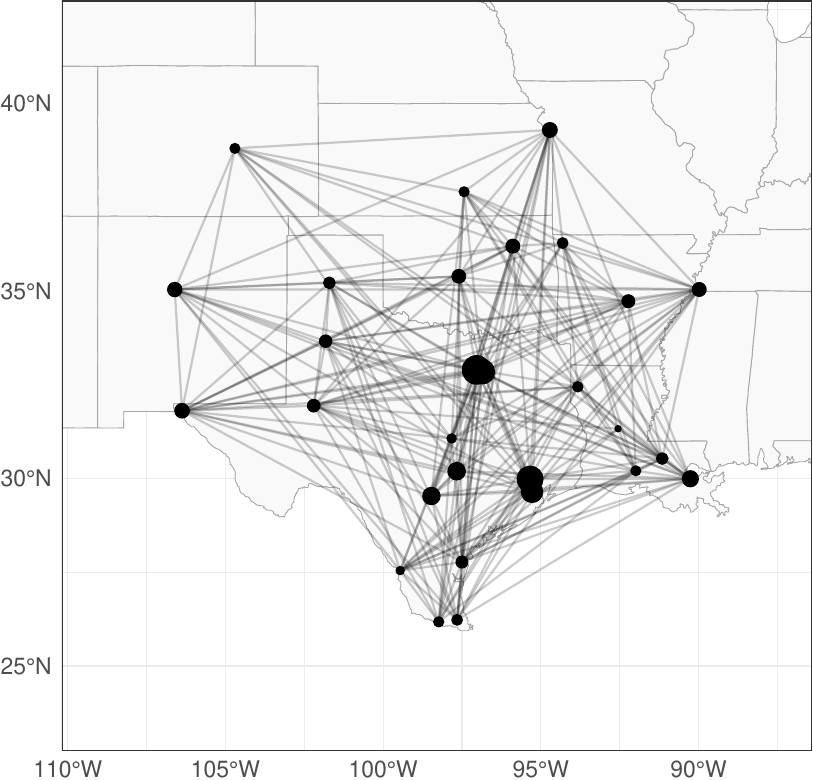}
    \includegraphics[width = .32\textwidth]{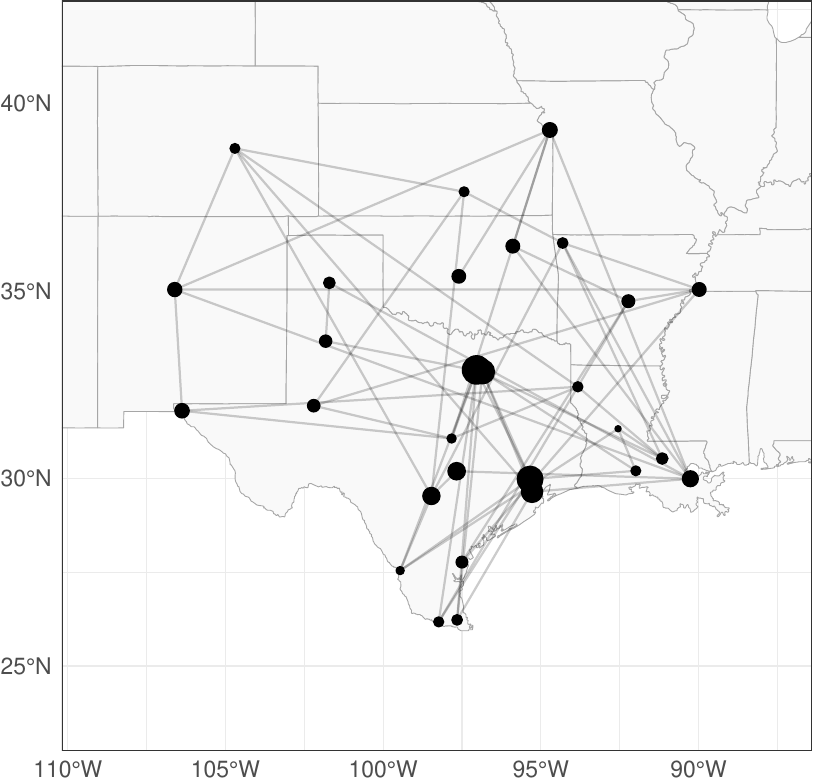}
    \caption{Airports in the Southern U.S.~(dots) and flight connections, where the thickness of the nodes indicates the average number of daily flights at the airports. Left: flight connection graph with an edge between any pair of airports with daily flights. Center: estimated graph of optimal \texttt{eglearn} model. Right: estimated sub-graph corresponding to observed variables of optimal \texttt{eglatent} model.}
    \label{fig:graphs}
\end{figure}

The left-hand side of Figure~\ref{fig:likelihood} shows the number of edges of \texttt{eglatent} and of \texttt{eglearn} as a function of the tuning parameter $\lambda_n$, where the parameter $\gamma$ related to the latent variable selection in \texttt{eglatent} is fixed to the default choice $\gamma = 4$; different values of $\gamma$ give similar results and are omitted here. We see that for both methods, larger values of $\lambda_n$ result in sparser graphs. It is important to note that for \texttt{eglearn}, we count the edges of the usual estimated graph. For our \texttt{eglatent} method we count the edges of the residual graph among the observed variables. The latent graphs generally have fewer edges and are therefore more easily interpretable.

To compare the different model fits and to select the optimal value for the tuning parameter~$\lambda_n$, we must compute the likelihood of the fitted models on an independent validation set. To this end, we split the data chronologically into five equally large folds and perform cross-validation by leaving one fold out (validation data) and fitting on the remaining four folds (training data). The results for model performance on the validation sets are then averaged. 
The right-hand side of Figure~\ref{fig:likelihood} shows the averaged log-likelihood values on the validation sets that were not used for model fitting. For both methods, we see that for too small values of $\lambda_n$, the graphs are too dense and overfit to the training data.
In fact, for $\lambda_n =0$, both models correspond to the fully connected graph whose performance (horizontal line) is much worse than the models enforcing sparsity. For too large values of $\lambda_n$, the graph becomes too sparse and the model is not flexible enough. Clearly, the latent model outperforms \texttt{eglearn}, indicating that latent variables are present in this data set. In this particular application, they can be thought of as confounding factors such 
as meteorological variables or strikes in the aviation industry that affect many airports simultaneously.

Figure~\ref{fig:graphs} compares the estimated graphs of \texttt{eglatent} (center) and \texttt{eglearn} (right) fitted on the whole data set, where the regularization parameter $\lambda_n$ in both methods is chosen as the maximizers of the respective validation likelihoods. We observe that the latent graph is much sparser and therefore highlights more clearly certain features of the system. For instance, it seems that hubs, such as the Fort Worth International Airport in Dallas (the thickest point on the map), are more central in the graph since they have more connections than smaller airports.  

\begin{figure}
    \centering
    \includegraphics[width = .9\textwidth]{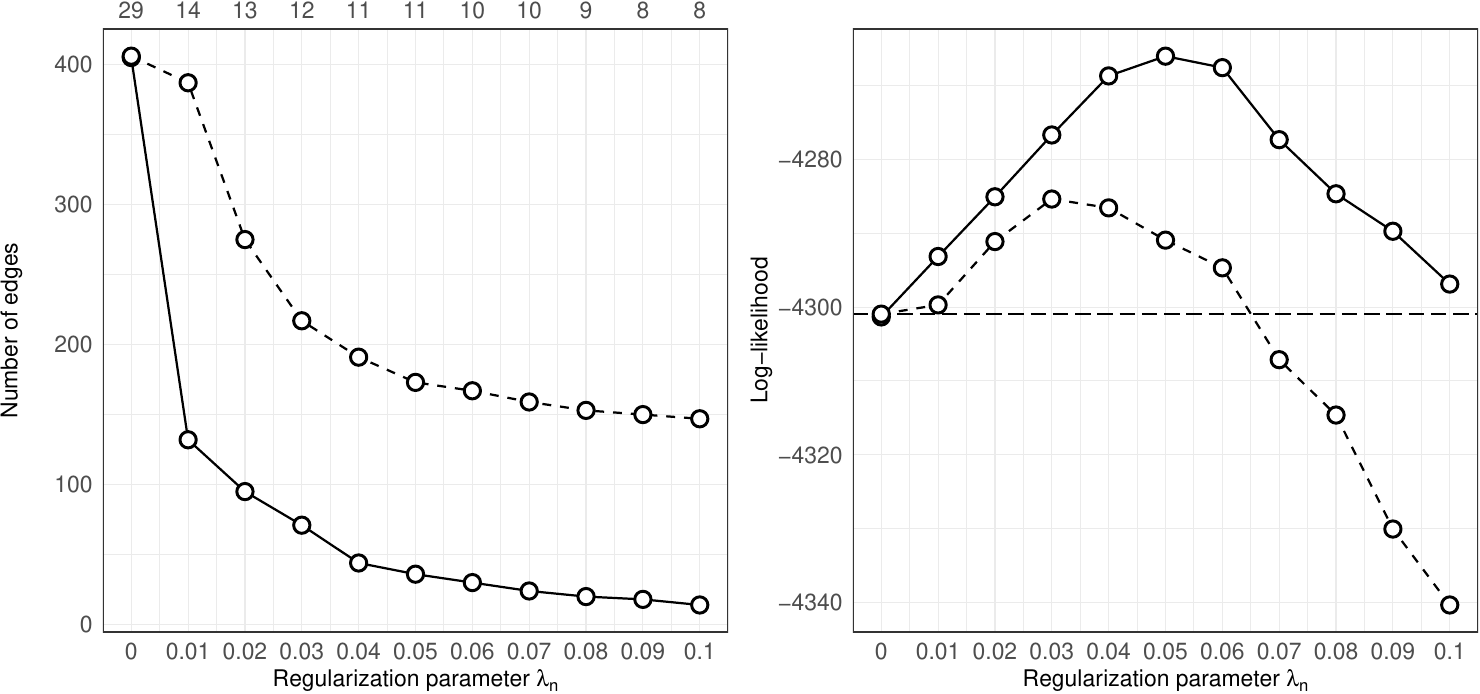}
    \caption{Left: number of edges of the estimated graph of \texttt{eglearn} (dashed line) and the estimated sub-graph of observed variables of \texttt{eglatent} (solid line) as functions of the regularization parameter $\rho$; top axis shows the number of latent variables in \texttt{eglatent}. Right: corresponding log-likelihoods; horizontal line is the validation log-likelihood of the fully connected graph.}
    \label{fig:likelihood}
\end{figure}

{
The number of exceedances $k$ used in the analysis, or equivalently, the probability threshold $q = 1-k/n$, is a tuning parameter appearing in virtually all extreme value analyses. In theory and simulation studies with knowledge of the underlying distributions, there is an optimal choice of the asymptotic order of $k$ compared to the sample size $n$; see for instance Remark~\ref{rem:max_stable}. In real data, we typically neither know the data generating distribution nor the second-order parameter $\xi$ that determines the rate of $k$ in Theorem~\ref{thm:main}. Therefore, it is common practice to run the analysis for different choices of reasonable values of $k$ and compare the results in terms of stability.
In addition to the threshold $q = 0.90$ ($k = 360$), we rerun the above application with thresholds $q = 0.85$ ($k = 540$) and $q = 0.95$ ($k = 180$); see Appendix~\ref{sec:additional_application} for the results. Similarly to Figure~\ref{fig:likelihood}, Figures~\ref{fig:likelihood_85} and~\ref{fig:likelihood_95} show that also for these threshold choices, \texttt{eglatent} outperforms \texttt{eglearn} significantly. Moreover, Figure~\ref{fig:graphs_comparison} compares the different estimated graphs among the observed variables for the three thresholds. We see that the results are very stable and the graphs only have a few edges that differ.
}

\section{Future work}
\label{sec:discussion}
Our work on latent variables in the analysis of extremal dependence opens several future research directions. First, as described in Section \ref{sec:theoretical_results}, our sample size requirement is driven by a spectral norm concentration result on the empirical variogram matrix. This result was derived by translating the $\ell_\infty$ concentration result of \cite{engelke2022b} to the spectral norm setting using equivalence of norms. To obtain tighter convergence results, one must obtain direct concentration bounds on the spectral norm; such a result would be of independent interest in the multivariate extremes literature. Second, solving \texttt{eglatent} can be challenging for large problems. Building on the work of \cite{Ma2012AlternatingDM} in the Gaussian setting, faster solvers can be developed using alternating direction method of multipliers \citep{Boyd2011DistributedOA}. Moreover, we observed in Section~\ref{sec:robust} that \texttt{eglatent} estimates a few latent variables to accurately recover the underlying graphical structure when there are no latent variables present. It would be of interest to develop a theoretical justification for this phenomenon. Also additional structure on the dependency structures among the observed and latent variables, such as multivariate total positivity of order 2 \citep{roettger2023, rod2024} or colored graphs \citep{roettger2023parametric}, may be exploited to develop more powerful extremal graphical models with latent variables. The recent connection between extremal graphical models and graphical models for L\'evy processes could allow us to use \texttt{eglatent} also for modeling latent variables in stochastic processes dependence structure \citep{engelke2024levygraphicalmodels}.

\section{Acknowledgments}
We thank the referees for their valuable comments that improved the paper. We thank Nicola Gnecco and Manuel Hentschel for their help with creating the figures in this paper.
The authors acknowledge funding from the Swiss National Science Foundation (Sebastian Engelke), NSF grant DMS-2413074 and the Royalty Research Fund at the University of Washington (Armeen Taeb).

\bibliographystyle{elsarticle-harv}
\bibliography{main.bib}

\newpage
\appendix
\section*{supplementary material}

\section{Useful lemmas for proving Theorem~\ref{thm:main1}}
Our analysis of Theorem~\ref{thm:main1} relies on some lemmas.

\begin{lemma} Let $A \in \mathbb{S}^{d}$ and $B \in \mathbb{S}^{d}$ be two symmetric matrices with $A + B$ being nonsingular and row/column spaces of $A$ and $B$ being orthogonal to one another. Then, $(A+B)^{-1} = A^+ + B^+$.
\label{lemma:pseudo}
\end{lemma}
\begin{proof}[Proof of Lemma~\ref{lemma:pseudo}] Let $U_A{D}_A{U}_A^T$ and $U_B{D}_B{U}_B^T$ be the reduced SVD of $A$ and $B$. Then, since $A+B$ is non-singular, and the subspaces spanned by the columns of $U_A$ and $U_B$ are orthogonal, we have that $(U_A,U_B)$ forms an orthogonal matrix. Therefore, 
$$(A+B) = \begin{pmatrix}U_A & U_B\end{pmatrix}\begin{pmatrix}D_A & 0\\0 & D_B\end{pmatrix}\begin{pmatrix}U_A & U_B\end{pmatrix}^T,$$ 
and thus $(A+B)^{-1} = U_A{D}_A^{-1}U_A^T +U_B{D}_B^{-1}U_B^T = A^+ + B^+.$
\end{proof}

\begin{lemma} Suppose that $UU^TMUU^T = M$. Then, $U(U^TMU)^{-1}U^T = M^+$.
\label{lemma:new_2}
\end{lemma}
\begin{proof}[Proof of Lemma~\ref{lemma:new_2}] Let $UDU^T$ be the reduced-SVD of $M$. Then, $U(U^MU)^{-1}U^T = UD^{-1}U^T$, which is equivalent to $M^+$.  
\end{proof}

\begin{lemma} \label{lemma: schur.Comp}
 Let $\Tilde{\Pi} = (I_p-\ones/p)$. Let $\Theta = \begin{pmatrix}\Theta_{O} & \Theta_{OH} \\ \Theta_{HO} & \Theta_{H} \end{pmatrix} \in \mathbb{R}^{d\times{d}}$ with $\Theta_O \in \mathbb{R}^{p\times{p}}$, $\Theta_H \in \mathbb{R}^{h\times{h}}$ and $d = h+p$.  Suppose $\Theta$ is a positive semi-definite matrix with its null-space being the span of the all-ones vector. Then:   
    \begin{equation*}
         \Tilde{\Pi}(\Theta_O - \Theta_{OH}{\Theta_H}^{-1}\Theta_{HO})\Tilde{\Pi} = \Theta_O - \Theta_{OH}{\Theta_H}^{-1}\Theta_{HO}.
    \end{equation*}
\end{lemma}
\begin{proof}
    Since $\Theta\mathbf{1}_{d}=0$, we have
    \begin{align}
        \Theta_O\mathbf{1}_{p} + \Theta_{OH}\mathbf{1}_{h} & = 0, \label{subequa: OOH}\\
        \Theta_{HO}\mathbf{1}_{p} + \Theta_{H}\mathbf{1}_{h} & = 0. \label{subequa: HOH}
    \end{align}
    Consider $\Theta_O- \Theta_{OH}{\Theta_H}^{-1}\Theta_{HO}$, we have 
    \begin{align*}
        (\Theta_O - \Theta_{OH}{\Theta_H}^{-1}\Theta_{HO})\mathbf{1}_{p} & = 
        \Theta_O\mathbf{1}_{p} - \Theta_{OH}{\Theta_H}^{-1}\Theta_{HO}\mathbf{1}_{p}, \\
        &\stackrel{by \eqref{subequa: HOH}}{=}\Theta_O \mathbf{1}_{p}+\Theta_{OH}{\Theta_H}^{-1} (\Theta_{H}\mathbf{1}_{h}), \\
        &= \Theta_O \mathbf{1}_{p}+\Theta_{OH}\mathbf{1}_{h}
         \stackrel{by \eqref{subequa: OOH}}{=} 0.
    \end{align*}
    Thus, $\mathbf{1}_{p} \in \text{ker}(\Theta_O - \Theta_{OH}{\Theta_H}^{-1}\Theta_{HO})$. To complete the proof, we will show that $\text{dim}(\text{ker}(\Theta_O- \Theta_{OH}{\Theta_H}^{-1}\Theta_{HO})) = 1$. Suppose there exist non-zero vector $v \in \text{ker}(\Theta_O - \Theta_{OH}{\Theta_H}^{-1}\Theta_{HO})$, and let $u = {\Theta_H}^{-1}\Theta_{HO} v$. Since $v \ne 0$, $u \ne 0$, and then it follows that $\Theta_O v - \Theta_{OH}{u}  = 0$ and $\Theta_{HO}v - \Theta_{H}{u} = 0$ yielding $       \Theta^\star\begin{pmatrix}
            v \\
            - u
        \end{pmatrix} = 0$.
    Since $\begin{pmatrix}
            v \\
            - u
        \end{pmatrix} \in \text{ker}(\Theta)$, $v = \alpha' \mathbf{1}_p$ for some $\alpha' \in R$, which implies that $\text{dim}(\text{ker}(\Theta_O- \Theta_{OH}{\Theta_H}^{-1}\Theta_{HO})) = 1$.
\end{proof}

\begin{lemma} \label{lemma:ptildep} Let $\Tilde{\Pi} = I_p - \ones/p$  and  ${\Pi} = I_{d}-\textbf{1}_{d}\textbf{1}_{d}^T/d$ with $d = p+h$. For any matrix $M \in \R^{d\times d}$, $\Tilde{\Pi} (\Pi M \Pi)_{1:p, 1:p} \Tilde{\Pi} = \Tilde{\Pi} M_{1:p, 1:p}\Tilde{\Pi}$
\end{lemma}
\begin{proof}
    Note that $(\Pi M \Pi)_{1:p, 1:p} = \begin{pmatrix}
        I_p & \mathbf{0} 
    \end{pmatrix} 
    \Pi M \Pi \begin{pmatrix}
        I_p \\ 
        \mathbf{0} 
    \end{pmatrix}$. Then it follows that
    \begin{equation}
        \begin{aligned}
        \Tilde{\Pi} (\Pi M \Pi)_{1:p, 1:p} \Tilde{\Pi} & = \Tilde{\Pi} \begin{pmatrix}
        I_p & \mathbf{0} 
    \end{pmatrix} 
    \Pi M \Pi \begin{pmatrix}
        I_p \\ 
        \mathbf{0} 
    \end{pmatrix} \Tilde{\Pi}= \begin{pmatrix}
        \Tilde{\Pi} & \mathbf{0} 
    \end{pmatrix} 
    \Pi M \Pi \begin{pmatrix}
        \Tilde{\Pi} \\ 
        \mathbf{0} 
    \end{pmatrix}.
    \end{aligned}
    \label{result_1}
    \end{equation}
Notice that:
\begin{eqnarray}
\begin{aligned}
\begin{pmatrix}
        \Tilde{\Pi} & \mathbf{0} 
    \end{pmatrix} 
    \Pi &=  \begin{pmatrix}(I_p - \ones/p)(I_p-\ones/d) & (I_p - \ones/p)\textbf{1}_{p}/d\end{pmatrix} \\
    &= \begin{pmatrix}(I_p - \ones/p)(I_p-\ones/p + \ones/(p)-\ones/d) &\mathbf{0}\end{pmatrix} \\
    &= \begin{pmatrix}\tilde{\Pi}(\tilde{\Pi} + \ones/(p)-\ones/d) &\mathbf{0}\end{pmatrix} = \begin{pmatrix} \tilde{\Pi} & \mathbf{0}\end{pmatrix}.
\end{aligned}
\label{result_2}
\end{eqnarray}   
Putting \eqref{result_1} and \eqref{result_2} together, we have the desired result.

\end{proof}

\section{Proof of Theorem~\ref{thm:main1}}

\begin{proof}[Proof of Theorem~\ref{thm:main1}] For notational simplicity, we let $M = -\Gamma^\star/2$. Let $\Pi = I_d - \textbf{1}_d\textbf{1}_d^T/d$. We have from \cite{hen2022} that $(\Pi{M}\Pi)^+ = \Theta^\star$ or equivalently $\Pi{M}\Pi = (\Theta^\star)^+$. Since $\Theta^\star$ has zero row/column sums and thus its row/column spaces are orthogonal to the all-ones vector,  we have by Lemma~\ref{lemma:pseudo} that for any $t > 0$, $({\Theta^\star}+t\textbf{1}_d\textbf{1}_d^T)^{-1} = {\Theta^\star}^{+} + (t\textbf{1}_d\textbf{1}_d^T)^+ = {\Theta^\star}^{+} + \frac{1}{td^2}(\textbf{1}_d\textbf{1}_d^T)$. As ${\Pi}\textbf{1}_d\textbf{1}_d^T{\Pi} = 0$, we have that:
$${\Pi}M{\Pi} = \Pi(\Theta^\star + t\textbf{1}_d\textbf{1}_d^T)^{-1}\Pi.$$
The equation above implies $\tilde{\Pi}[{\Pi}M{\Pi}]_{1:p,1:p}\tilde{\Pi} = \tilde{\Pi}[\Pi(\Theta^\star + t\textbf{1}_d\textbf{1}_d^T)^{-1}\Pi]_{1:p,1:p}\tilde{\Pi}.$
Using Lemma~\ref{lemma:ptildep}, we have that:
\begin{equation}\tilde{\Pi}M_{1:p,1:p}{\tilde{\Pi}} = \tilde{\Pi}[(\Theta^\star + t\textbf{1}_d\textbf{1}_d^T)^{-1}]_{1:p,1:p}\tilde{\Pi}.
\label{eqn:important}
\end{equation}
We will now analyze the term $[(\Theta^\star + t\textbf{1}_d\textbf{1}_d^T)^{-1}]_{1:p,1:p}$ inside \eqref{eqn:important}. From Schur's complement, we have that:
\begin{eqnarray}
\begin{aligned}
    &[(\Theta^\star + t\textbf{1}_d\textbf{1}_d^T)^{-1}]_{1:p,1:p} =\\& \left[\Theta^\star_{O}+t\textbf{1}_p\textbf{1}_p^T - (\Theta_{OH}^\star+t\textbf{1}_{p}\textbf{1}_{h}^T)(\Theta_{H}^\star+t\textbf{1}_{h}\textbf{1}_{h}^T)^{-1}(\Theta_{HO}^\star+t\textbf{1}_{h}\textbf{1}_{p}^T)\right]^{-1}.
\end{aligned}
\label{eqn:temp_1}
\end{eqnarray}
By the Woodbury inversion lemma, we have that:
\begin{eqnarray}
(\Theta^\star_{H}+t\textbf{1}_{h}\textbf{1}_{h}^T)^{-1} = (\Theta_H^\star)^{-1} - (\Theta_H^\star)^{-1}\textbf{1}_{h}\left(\frac{1}{t}+\textbf{1}_h^T(\Theta_H^\star)^{-1}\textbf{1}_h\right)^{-1}\textbf{1}_h^{T}(\Theta_H^\star)^{-1}.
\label{eqn:temp_2}
\end{eqnarray}
Plugging the result of \eqref{eqn:temp_2} into \eqref{eqn:temp_1}, we have that:
\begin{eqnarray*}
\begin{aligned}
   &[(\Theta^\star + t\textbf{1}_d\textbf{1}_d^T)^{-1}]_{1:p,1:p}\\&=\Theta^\star_{O}+t\textbf{1}_p\textbf{1}_p^T - (\Theta_{OH}^\star+t\textbf{1}_{p}\textbf{1}_{h}^T)(\Theta_{H}^\star+t\textbf{1}_{h}\textbf{1}_{h}^T)^{-1}(\Theta_{HO}^\star+t\textbf{1}_{h}\textbf{1}_{p}^T)=  A+B+C
   \end{aligned}
\end{eqnarray*}
where 
\begin{equation*}
\begin{aligned}
A &= \Theta^\star_{O}-\Theta^\star_{OH}(\Theta_H^\star)^{-1}\Theta^\star_{HO},\\
B &=  t\textbf{1}_p\textbf{1}_p^T + t^2\textbf{1}_p\textbf{1}_h^T(\Theta^\star_{H}+t\textbf{1}_{h}\textbf{1}_{h}^T)^{-1}\textbf{1}_h\textbf{1}_p^T,\\
C &= t\textbf{1}_{p}\textbf{1}_{h}^T(\Theta_{H}^\star+t\textbf{1}_{h}\textbf{1}_{h}^T)^{-1}\Theta_{HO}^\star+t\Theta_{OH}^\star(\Theta_{H}^\star+t\textbf{1}_{h}\textbf{1}_{h}^T)^{-1}\textbf{1}_h\textbf{1}_{p}^T\\&+\Theta^\star_{OH}(\Theta_H^\star)^{-1}\textbf{1}_{h}\left(\frac{1}{t}+\textbf{1}_h^T(\Theta_H^\star)^{-1}\textbf{1}_h\right)^{-1}\textbf{1}_h^{T}(\Theta_H^\star)^{-1}\Theta^\star_{HO}.
\end{aligned}
\end{equation*}

From Lemma~\ref{lemma: schur.Comp}, we have that: $\tilde{\Pi}A\tilde{\Pi} = A$. Furthermore, notice that $B$ lies in the all-ones subspace, i.e. $\tilde{\Pi}B\tilde{\Pi}=0$ and is a positive semi-definite matrix for $t > 0$. Thus, the matrix $A + B$ is invertible.  Notice
\begin{eqnarray*}
\begin{aligned}
&\lim_{t \to 0}t\textbf{1}_{p}\textbf{1}_{h}^T(\Theta_{H}^\star+t\textbf{1}_{h}\textbf{1}_{h}^T)^{-1}\Theta_{HO}^\star = \lim_{t \to 0}t\textbf{1}_{p}\textbf{1}_{h}^T(\Theta_{H}^\star)^{-1}\Theta_{HO}^\star = 0, \\
&\lim_{t\to{0}}t\Theta_{OH}^\star(\Theta_{H}^\star+t\textbf{1}_{h}\textbf{1}_{h}^T)^{-1}\textbf{1}_h\textbf{1}_{p}^T = \lim_{t\to{0}}t\Theta_{OH}^\star(\Theta_{H}^\star)^{-1}\textbf{1}_h\textbf{1}_{p}^T = 0,\\
&\lim_{t\to{0}}\Theta^\star_{OH}(\Theta_H^\star)^{-1}\textbf{1}_{h}\left(\frac{1}{t}+\textbf{1}_h^T(\Theta_H^\star)^{-1}\textbf{1}_h\right)^{-1}\textbf{1}_h^{T}(\Theta_H^\star)^{-1}\Theta^\star_{HO} = \lim_{t\to{0}}t\Theta^\star_{OH}(\Theta_H^\star)^{-1}\textbf{1}_{h}(\Theta_H^\star)^{-1}\Theta^\star_{HO} = 0,
\end{aligned}
\end{eqnarray*}
so that $\lim_{t\to 0}C = 0$. Notice on the other hand that $\lim_{t \to 0}A+B \neq 0$. By the Woodbury inversion lemma, we have that: $(A+B+C)^{-1} = (A+B)^{-1}-(A+B)^{-1}C(I + (A+B)^{-1}C)^{-1}(A+B)^{-1}$. Thus:
$$\lim_{t \to 0}\tilde{\Pi}(A+B+C)^{-1}\tilde{\Pi} = \tilde{\Pi}\lim_{t \to 0}(A+B)^{-1}\tilde{\Pi} - 
\lim_{t \to 0}\tilde{\Pi}(A+B)^{-1}C(I + A^{-1}C)^{-1}A^{-1}\tilde{\Pi}.$$
Since $\lim_{t \to 0}{C} = 0$, we have that:
$$ \lim_{t\to \infty}\tilde{\Pi}[(\Theta^\star + t\textbf{1}_d\textbf{1}_d^T)^{-1}]_{1:p,1:p}\tilde{\Pi} = \lim_{t \to 0}\tilde{\Pi}(A+B+C)^{-1}\tilde{\Pi} = \lim_{t \to 0}\tilde{\Pi}(A+B)^{-1}\tilde{\Pi} = \tilde{\Pi}A^{+}\tilde{\Pi} = A^{+}.$$
Here, the second equality follows from noting that the row/column spaces of $A$ and $B$ are orthogonal to one another and so by Lemma~\ref{lemma:pseudo}, $(A+B)^{-1} = A^{+} + B^{+}$. Furthermore, since $B$ is a multiple of all-ones matrix, $\tilde{\Gamma}B^{+}\tilde{\Gamma} = 0$. The last equality follows from Lemma~\ref{lemma: schur.Comp}. Noting that $M_{1:p,1:p} =-\Gamma_O^\star/2$ and plugging in $A^+$ for $\tilde{\Pi}[(\Theta^\star + t\textbf{1}_d\textbf{1}_d^T)^{-1}]_{1:p,1:p}\tilde{\Pi}$ in \eqref{eqn:important}, we conclude that:
$$(\tilde{\Pi}(-{\Gamma_O^\star}/2)\tilde{\Pi})^{+} = \Theta_O^\star - \Theta_{OH}^\star(\Theta_{H}^\star)^{-1}\Theta_{HO}^\star.$$
Taking pseudo-inverses of both sides, we have the desired result. In Lemma~\ref{lemma: schur.Comp}, we also showed that $\Theta_O^\star - \Theta_{OH}^\star(\Theta_{H}^\star)^{-1}\Theta_{HO}^\star = \tilde{\Pi}(\Theta_O^\star - \Theta_{OH}^\star(\Theta_{H}^\star)^{-1}\Theta_{HO}^\star)\tilde{\Pi}$.
 
\end{proof}

\section{Arriving at estimator \eqref{eqn:estimator}}
\label{appendix:arriving}
Recall that $\tilde{\Theta}^\star = (\tilde{\Pi}(-\Gamma_O^\star/2)\tilde{\Pi})^{+}$, where $\tilde{\Pi} = UU^T$. Furthermore, the null-space of $\tilde{\Theta}^\star$ is the subspace $\text{span}(\ones)$. In other words, $UU^T\tilde{\Theta}^\star{UU}^T = \tilde{\Theta}^\star$. We arrive at our estimator by noting that $\tilde{\Theta}^\star$ is the unique minimizer of the convex program:
\begin{equation}
\label{eqn:pop_opt}
\begin{aligned}
    \hat{\Theta} = \argmin_{\Theta \in \mathbb{S}^{p}} &~~-\log{\det}\left(U^T \Theta U\right) - \frac12 \mathrm{tr}(\Theta{\Gamma}^\star_O),\\
    \text{s.t}&~~~\Theta \succeq 0~~,~~\Theta\mathbf{1}_p = 0.
\end{aligned}\end{equation}
To see why that is, first note that the constraint $\Theta \succeq 0$ can be removed since the log-det function forces $U^T\Theta{U}$ to be positive definite and together with the constraint $\Theta\mathbf{1}_p$ forces $\Theta \succeq 0$ and additionally $UU^T\Theta{UU}^T = \Theta$. Note that $\mathrm{tr}(\Theta{\Gamma}^\star_O) = \mathrm{tr}(UU^T\Theta{UU}^T\Gamma^\star_O) = \mathrm{tr}(\Theta{UU}^T\Gamma^\star_OUU^T)$. Thus, an equivalent optimization to \eqref{eqn:pop_opt} is 
\begin{equation}
\label{eqn:pop_opt_equiv}
\begin{aligned}
    \hat{\Theta} = \argmin_{\Theta \in \mathbb{S}^{p}} &~~-\log{\det}\left(U^T \Theta U\right) - \frac12 \mathrm{tr}(\Theta{UU^T}{\Gamma}^\star_OUU^T),\\
    \text{s.t}&~~~\Theta \in \mathrm{span}(\ones)^\perp.
\end{aligned}\end{equation}
Using Lagrangian duality theory, we have that $\hat{\Theta}$ must satisfy for some $t\in \mathbb{R}$
$$-U(U^T\hat{\Theta}U^{-1})U^T - \frac12{UU^T}{\Gamma}^\star_OUU^T + t\ones = 0.$$
Note that $t = 0$ since the first two terms live in the space spanned by the columns of $U$ and the last term lies in the orthogonal subspace. Similarly, $-U(U^T\hat{\Theta}U^{-1})U^T - \frac12{UU^T}{\Gamma}^\star_OUU^T = 0$. Since $UU^T\hat{\Theta}UU^T = \hat{\Theta}$, we appeal to Lemma~\ref{lemma:new_2} to conclude that $\hat{\Theta}^+ = -\frac12 \frac12{UU^T}{\Gamma}^\star_OUU^T$. Some simple manipulations allow us to conclude that $\hat{\Theta} = \tilde{\Theta}^\star$. 

\section{Useful lemmas for proof of consistency}
Our analysis will depend on the following quantities for any pair of subspaces $\Omega,T \subseteq \mathbb{R}^{p \times p}$:
\begin{eqnarray*}
\begin{aligned}
\theta(\Omega) :=\max_{N \in \Omega, \|N\|_\infty = 1}\|N\|_2~~~;~~~ \xi(T) :=\max_{N \in T, \|N\|_2 = 1}\|N\|_\infty.
\end{aligned}
\end{eqnarray*}
When $\Omega = \Omega^\star$ and $T = T^\star$, these quantities are closely connected to the maximal degree $d^\star$ and the incoherence parameter $\inc^\star$ (defined in Section~\ref{sec:fisher_conds}). In particular, \cite{Chand2012} showed that $\mu(\Omega^\star) \in [0,d^\star]$ and $\xi(T^\star) \in [\inc^\star,2\inc^\star]$. 

\subsection{Some auxillary lemmas}

\begin{lemma}[Lemma 3.1 of \cite{Chand2012}] For any tangent spaces $T_1,T_2$ of same dimension with $\rho(T_1,T_2)<1$, we have that: $\xi(T_2) \leq \frac{\xi(T_1)+\rho(T_1,T_2)}{1-\rho(T_1,T_2)}$.
\label{lemma:1}
\end{lemma}

\begin{lemma}Consider a tangent space $T'$ of a symmetric matrix with $\rho(T^\star,T') \leq \omega$ with $\omega<1$. Let $\mathcal{C}'$ and $\mathcal{C}^\star$ be the column spaces that form the tangent spaces $T'$ and $T^\star$ respectively. Then, we have that:
$\|\proj_{\mathcal{C}'}-\proj_{\mathcal{C}^\star}\|_2 \leq \omega$. 
\label{lemma:2}
\end{lemma}
\begin{proof}[Proof of Lemma~\ref{lemma:2}]
Since $\omega<1$, $T^\star$ and $T'$ are of the same dimension. Let $\sigma_{s}(\cdot)$ be the $s$-th largest singular value of the input matrix. Notice that 
\begin{eqnarray*}
\begin{aligned}
\|\proj_{\mathcal{C'}}-\proj_{\mathcal{C}^\star}\|_2 = \|\proj_{\mathcal{C'}^\perp}-\proj_{{\mathcal{C}^\star}^\perp}\|_2 = \sqrt{1-\sigma_{p-k}\left(\proj_{{\mathcal{C}'}^\perp}\proj_{{\mathcal{C}^\star}^\perp}\right)^2} &= \sqrt{1-\sigma_{(p-k)^2}\left(\proj_{{{T}'}^\perp}\proj_{{T^\star}^\perp}\right)}\\
&\leq \sqrt{1-\sigma_{(p-k)^2}\left(\proj_{{{T}'}^\perp}\proj_{{T^\star}^\perp}\right)^2}\\
&= \|\proj_{{T'}^\perp}-\proj_{{T^\star}^\perp}\|_2=\|\proj_{{T'}}-\proj_{{T^\star}}\|_2.
\end{aligned}
\end{eqnarray*}

\end{proof}

{\subsection{Lemmas to account for the zero row-sum constraint}}
\label{sec:lemmas_zero_row_sum_constraint}
To deal with the additional dual parameter $t\ones$ introduced by the zero row-sum constraint $(S-L)\textbf{1}_p$, our analysis requires the following lemmas.

\begin{lemma}Let $\mathcal{C}_1,\mathcal{C}_2 \subseteq \mathbb{R}^p$ be a pair of subspaces. Then, for any $z\in \mathbb{R}^p$:
\begin{eqnarray*}
    \begin{aligned}\max_{v \in \mathcal{C}_1 \oplus \mathcal{C}_2, \|v\|_2=1} \langle z,v\rangle &\leq 2\min\left\{\max_{v \in \mathcal{C}_1, \|v\|_2=1} \langle z,v\rangle, \max_{v \in \mathcal{C}_2, \|v\|_2=1} \langle z,v\rangle\right\}\\&+\max\left\{\max_{v \in \mathcal{C}_1, \|v\|_2=1} \langle z,v\rangle, \max_{v \in \mathcal{C}_2, \|v\|_2=1} \langle z,v\rangle\right\}.
    \end{aligned}
    \end{eqnarray*}
\label{lemma:3}
\end{lemma}
\begin{proof}[Proof of Lemma~\ref{lemma:3}] 
Suppose without loss of generality that $\max_{\substack{u_1 \in \mathcal{C}_1,\|u_1\|_2=1}} u_1^Tz \leq \max_{\substack{u_2 \in \mathcal{C}_2,\|u_2\|_2=1}}u_2^Tz$. Thus 
\begin{eqnarray*}
\begin{aligned}
\max_{v \in \mathcal{C}_1 \oplus \mathcal{C}_2, \|v\|_2=1}\langle z,v\rangle &=\max_{\substack{u_1 \in \mathcal{C}_1,u_2\in\mathcal{C}_2,\|u_1\|_2=\|u_2\|_2 = 1\\v = c_1u_1+c_2u_2}}  |v^Tz|/\|v\|_2, \\
&= \max_{\substack{u_1 \in \mathcal{C}_1,u_2\in\mathcal{C}_2,\|u_1\|_2=\|u_2\|_2 = 1\\ u_3 = u_2 - (u_2^Tu_1)u_1 \\ v = c_1u_1+c_2u_3}}  |v^Tz|/\|v\|_2, \\
&\leq \max_{\substack{u_1 \in \mathcal{C}_1,u_2\in\mathcal{C}_2,\|u_1\|_2=\|u_2\|_2 = 1\\ u_3 = u_2 - (u_2^Tu_1)u_1 \\ v = c_1u_1+c_2u_3}}  \frac{|c_1|}{\sqrt{c_1^2+c_2^2}}|u_1^Tz| + \frac{|c_2|}{\sqrt{c_1^2+c_2^2}}|u_3^Tz|, \\
&\leq \max_{\substack{u_1 \in \mathcal{C}_1,\|u_1\|_2=1}} 2|u_1^Tz|+\max_{\substack{u_2 \in \mathcal{C}_2,\|u_2\|_2=1}}|u_2^Tz|.\\
    \end{aligned}
\end{eqnarray*}
\end{proof}

\begin{lemma}Let $Z \in T' \oplus \mathrm{span}(\ones)$ with $\rho(T',T^\star)\leq\omega$ and $\|Z\|_2 = 1$. Then, $1+2(\kappa^\star+\omega)\geq \|\proj_{T'}(Z)\|_2 \geq 1-2(\kappa^\star+\omega)$ and thus $\|\proj_{{T'}^\perp}(Z)\|_2 \leq 2({\kappa^\star}+\omega)$.
\label{lemma:4}
\end{lemma}
\begin{proof}[Proof of Lemma~\ref{lemma:4}] %Reverse triangle inequality yields the bound $\|Z_2\|_2 \leq 1+\|Z_1\|$. Some algebraic manipulations yield $\|Z_1+\proj_{T'}(Z_2)\|_2 \leq1$. Using reverse triangle inequality and the definition of $\kappa$, we have that $\|Z_1\|_2 \leq 1+(1-\kappa+\min\{\kappa,\xi(T^\star)\}/2)\|Z_2\|_2$. Plugging this into the earlier bound, we have that $\|Z_2\|_2 \leq \frac{1}{\kappa-\min\{\kappa,\xi(T^\star)\}/2}$ and consequently $\|Z_1\|_2 \leq \frac{1}{\kappa-\min\{\kappa,\xi(T^\star)\}/2}$. \\
Note that $\|Z\|_2+\|\proj_{{T'}^\perp}(Z)\|_2\geq \|\proj_{T'}(Z)\|_2 \geq \|Z\|_2-\|\proj_{{T'}^\perp}(Z)\|_2$. Let $T'$ be a tangent space with associated row and column spaces $\mathcal{C}'$ and $\mathcal{R}'$. Let $\tilde{C}= \mathcal{C}' \oplus \mathrm{span}(\textbf{1}_p)$ and $\tilde{R}= \mathcal{R}' \oplus \mathrm{span}(\textbf{1}_p)$. Since $Z \in T' \oplus \mathrm{span}(\textbf{1}_p)$, it is straightforward to show that $Z = \proj_{\tilde{\mathcal{C}}}Z\proj_{\tilde{\mathcal{R}}^\perp} + Z\proj_{\tilde{\mathcal{R}}}$. Therefore, we have that $\proj_{{T'}^\perp}(Z) = \proj_{{\mathcal{C}'}^\perp}\left[\proj_{\tilde{\mathcal{C}}}Z\proj_{\tilde{\mathcal{R}}^\perp} + Z\proj_{\tilde{\mathcal{R}}}\right]\proj_{{\mathcal{R}'}^\perp}$.  Thus, $\|\proj_{{T'}^\perp}(Z)\|_2 \leq \|\proj_{{\mathcal{C}'}^\perp}\proj_{\tilde{\mathcal{C}}}\|_2 + \|\proj_{\tilde{\mathcal{R}}}\proj_{{\mathcal{R}'}^\perp}\|_2.$ Letting $\mathcal{C}_1 = \mathcal{C}'$ and $\mathcal{C}_2 = \mathrm{span}(\textbf{1}_p)$, we appeal to Lemma~\ref{lemma:3} to conclude that:

\begin{eqnarray*}
\begin{aligned}
    \max_{v\in \tilde{\mathcal{C}},\|v\|_2=1} \|\proj_{{\mathcal{C}'}^\perp}(v)\|_2  &\leq \max_{\substack{z\in {\mathcal{C}'}^\perp\\\|z\|_2=1}}\max_{\substack{u_1\in \mathcal{C}',\|u_1\|_2=1}}2|\langle z, u_1\rangle| + \max_{\substack{z\in {\mathcal{C}'}^\perp\\\|z\|_2=1}}\max_{\substack{u_2\in \mathrm{span}(\textbf{1}),\|u_2\|_2=1}}|\langle z, u_2\rangle|\\&= \|\proj_{{\mathcal{C}'}^\perp}(\textbf{1}_p/\sqrt{p})\|_2.
\end{aligned}
\end{eqnarray*}
Again, appealing to Lemma~\ref{lemma:3}, 
\begin{eqnarray*}
\begin{aligned}
    \max_{v\in {\mathcal{C}'}^\perp, \|v\|_2=1} \|\proj_{\tilde{\mathcal{C}}}(z)\|_2 &\leq \max_{\substack{z\in {\mathcal{C}'}^\perp\\\|z\|_2=1}}\max_{\substack{u_1\in \mathcal{C}',\|u_1\|_2=1}}2|\langle z, u_1\rangle| + \max_{\substack{z\in {\mathcal{C}'}^\perp\\\|z\|_2=1}}\max_{\substack{u_2\in \mathrm{span}(\textbf{1}),\|u_2\|_2=1}}|\langle z, u_2\rangle|\\&= \|\proj_{{\mathcal{C}'}^\perp}(\textbf{1}_p/\sqrt{p})\|_2.
\end{aligned}
\end{eqnarray*}
So we have concluded that $\|\proj_{{\mathcal{C}'}^\perp}\proj_{\tilde{\mathcal{C}}}\|_2 \leq  \|\proj_{{\mathcal{C}'}^\perp}(\textbf{1}/\sqrt{p})\|_2$. Thus, appealing to Lemma~\ref{lemma:2}
, $\|\proj_{{\mathcal{C}'}^\perp}\proj_{\tilde{\mathcal{C}}}\|_2 \leq \kappa^\star + \omega$. Similarly, we have that: $\|\proj_{{\mathcal{R}'}^\perp}\proj_{\tilde{\mathcal{R}}}\|_2 \leq  \|\proj_{{\mathcal{R}'}^\perp}(\textbf{1}/\sqrt{p})\|_2$ and thus $\|\proj_{{\mathcal{R}'}^\perp}\proj_{\tilde{\mathcal{R}}}\|_2 \leq \kappa^\star + \omega$. Putting things together, we have the desired bound. 

\end{proof}

\begin{lemma}Let $T' \subseteq \mathbb{R}^{p \times p}$ be a tangent space to a low-rank variety. Then,\\ $\|\proj_{({T' \oplus \mathrm{span}(\ones))}^\perp}(L)\|_2\leq \|\proj_{{T'}^\perp}(L)\|_2$ for any matrix $L \in \mathbb{R}^{p \times p}$.
\label{lemma:5}
\end{lemma}
\begin{proof}[Proof of Lemma~\ref{lemma:5}] Let $\mathcal{R}',\mathcal{C}'$ be row/column space pair that form the tangent space $T'$. Let $\tilde{C} = \mathrm{span}(\mathcal{C}',\textbf{1})$ and $\tilde{\mathcal{R}} = \mathrm{span}(\mathcal{R}',\textbf{1})$. Then, it is straightforward to see that $T' \oplus \mathrm{span}(\ones)$ is itself a tangent space formed by column space $\tilde{C}$ and row space $\tilde{R}$. Thus, $\|\proj_{({T' \oplus \mathrm{span}(\ones)})^\perp}(L^\star)\|_2 = \|\proj_{\tilde{\mathcal{C}}^\perp}L^\star\proj_{\tilde{\mathcal{R}}^\perp}\|_2$. Since $\mathcal{C}' \subseteq \tilde{\mathcal{C}}$, we have that: $\|\proj_{\tilde{\mathcal{C}}^\perp}L\proj_{\tilde{\mathcal{R}}^\perp}\|_2 \leq \|\proj_{{\mathcal{C}'}^\perp}L\proj_{{\mathcal{R}'}^\perp}\|_2 = \|\proj_{{T'}^\perp}(L)\|_2$.
\end{proof}

\begin{lemma} Suppose that $\kappa^\star > \omega$. Then, $\mathrm{span}(\ones) \cap (T' \oplus T^\star)= \{0\}$ for every tangent space $T'$ with $\rho(T',T^\star)\leq \omega$. 
\label{lemma:6}
\end{lemma}
\begin{proof}[Proof of Lemma~\ref{lemma:6}] It suffices to show that $\|\proj_{({T' \oplus T^\star})^\perp}(\ones/p)\|_2>0$. Let $\mathcal{C}'$ be the column space associated with the tangent space $T'$ at a symmetric matrix. Note that $T' \oplus T^\star$ is another tangent space with column space $\mathcal{C}' \oplus \mathcal{C}^\star$. Then, $\|\proj_{({T' \oplus T^\star})^\perp}(\ones/p)\|_2 = \|\proj_{(\mathcal{C}'\oplus \mathcal{C}^\star)^\perp}(\textbf{1}/\sqrt{p})\|_2^2$. So it suffices to show that $\|\proj_{\mathcal{C}'\oplus \mathcal{C}^\star}(\textbf{1}/\sqrt{p})\|_2 < 1$. Note additionally that $\|\proj_{\mathcal{C}'\oplus \mathcal{C}^\star}(\textbf{1}/\sqrt{p})\|_2 \leq \|\proj_{\mathcal{C}'\oplus \mathcal{C}^\star}\proj_{\mathcal{C}^\star}(\textbf{1}/\sqrt{p})\|_2 + \|\proj_{\mathcal{C}'\oplus \mathcal{C}^\star}\proj_{{\mathcal{C}^\star}^\perp}(\textbf{1}/\sqrt{p})\|_2 \leq \|\proj_{\mathcal{C}^\star}(\textbf{1}/\sqrt{p})\|_2 + \|\proj_{\mathcal{C}'\oplus \mathcal{C}^\star}\proj_{{\mathcal{C}^\star}^\perp}\|_2$. We have that: $\|\proj_{\mathcal{C}^\star}(\textbf{1}/\sqrt{p})\|_2 = 1-\kappa^\star$. Using Lemma \ref{lemma:3}, it is straightforward to conclude that $\|\proj_{\mathcal{C}'\oplus \mathcal{C}^\star}\proj_{{\mathcal{C}^\star}^\perp}\|_2 \leq \|\proj_{\mathcal{C}'}\proj_{{\mathcal{C}^\star}^\perp}\|_2 \leq \|\proj_{\mathcal{C}'}-\proj_{\mathcal{C}^\star}\|_2$. Appealing to Lemma~\ref{lemma:2}
, and putting everything together, we conclude that: $\|\proj_{\mathcal{C}'\oplus \mathcal{C}^\star}(\textbf{1}/\sqrt{p})\|_2 \leq (1-\kappa^\star)+\omega$. As $\kappa^\star > \omega$, we have the desired result.  
\end{proof}

\begin{lemma}Let $Z = T' \oplus \mathrm{span}(\ones)$ with $\|Z\|_2 = 1$ and $\rho(T',T^\star) \leq \omega$. Then, assuming $\kappa^\star >{\omega}$, $Z$ can be decomposed uniquely as follows $Z = Z_1 + Z_2$ where $Z_1 \in T'$, $Z_2 \in \mathrm{span}(\ones)$ with $\max\{\|Z_1\|_2,\|Z_2\|_2\} \leq  \frac{2\sqrt{5h}}{1-\sqrt{1- ({\kappa^\star}-\omega)^2}}$.
\label{lemma:new}
\end{lemma}
\begin{proof}[Proof of Lemma~\ref{lemma:new}] The unique decomposition follows from Lemma~\ref{lemma:6}. Since $\omega<1$, we have that $T'$ and $T^\star$ have the same dimension. Since $Z_1 \in T' \oplus T^\star$, then $\text{rank}(Z_1) \leq 4h$ (this follows from noting that every matrix inside $T'$ or $T^\star$ has rank at most $2h$ and rank of a sum of matrices is less than the sum of the ranks). Further, $\text{rank}(Z_2) \leq 1$, so that $\text{rank}(Z) \leq 5h$. Therefore, $\|Z\|_F \leq \sqrt{5h}$. Notice that: $\|Z\|_F^2 = \|Z_1 + \proj_{T'}(Z_2) + \proj_{{T'}^\perp}(Z_2)\|_F^2 = \|Z_1+\proj_{T'}(Z_2)\|_F^2+ \|\proj_{{T'}^\perp}(Z_2)\|_F^2$. Thus, $\|Z_1+\proj_{T'}(Z_2)\|_F \leq \sqrt{5h}$. Using reverse triangle inequality, we conclude that $\|Z_1\|_F \leq \sqrt{5h}+\|\proj_{T'}(Z_2)\|_F$. Now notice that: $\|Z_2\|_F^2 = \|\proj_{{T'}^\perp}(Z_2)\|_F^2+ \|\proj_{T'}(Z_2)\|_F^2$, so that: $\sqrt{\|Z_2\|_F^2 -  \|\proj_{{T'}^\perp}(Z_2)\|_F^2} = \|\proj_{T'}(Z_2)\|_F$. Since $Z_2$ is rank-1, we have then that: $\|\proj_{T'}(Z_2)\|_F = \|Z_2\|_2\sqrt{1- \|\proj_{{T'}^\perp}(\ones/p)\|_2^2}$. Combining things, we conclude that $\|Z_1\|_F \leq \sqrt{5h} +\|Z_2\|_2\sqrt{1- \|\proj_{{T'}^\perp}(\ones/p)\|_2^2}$. Notice that $\|\proj_{{T'}^\perp}(\ones/p)\|_2 \geq \|\proj_{{T^\star}^\perp}(\ones/p)\|_2-\omega = {\kappa^\star}-\omega $. Reverse triangle inequality also gives $\|Z_2\|_F \leq \|Z_1\|_F+\sqrt{5h}$. Putting the last bounds together, we have that: $\|Z_2\|_F \leq \frac{2\sqrt{5h}}{1-\sqrt{1- ({\kappa^\star}^2-\omega)^2}}$. Plugging this into a previous bound, we also find that $\|Z_1\|_F \leq  \frac{2\sqrt{5h}}{1-\sqrt{1- ({\kappa^\star}^2-\omega)^2}}$.

\end{proof}
\begin{lemma}
Let $T'\subseteq \mathbb{R}^{p \times p}$ be a tangent space to a low-rank variety. Then:\\ $\max_{N \in T' \oplus \mathrm{span}(\ones), \|N\|_2=1}\allowbreak \|N\|_\infty \leq {3\xi(T')}.$
    \label{lemma:3p}
\end{lemma}
\begin{proof}[Proof of Lemma~\ref{lemma:3p}] Let $(\mathcal{R}',\mathcal{C}')$ be the row/column space pair associated with $T'$. Let $\tilde{C}= \mathcal{C}' \oplus \mathrm{span}(\textbf{1}_p)$ and $\tilde{R}= \mathcal{R}' \oplus \mathrm{span}(\textbf{1}_p)$. Since $Z \in T' \oplus \mathrm{span}(\textbf{1})$, it is straightforward to show that $Z = \proj_{\tilde{\mathcal{C}}}Z\proj_{\tilde{\mathcal{R}}^\perp} + Z\proj_{\tilde{\mathcal{R}}}$. Therefore, $\|Z\|_\infty \leq \max_{i}\|\proj_{\tilde{\mathcal{C}}}(e_i)\|_2 + \max_{i}\|\proj_{\tilde{\mathcal{R}}}(e_i)\|_2$. Letting $\mathcal{C}_1 = \mathcal{C}'$ and $\mathcal{C}_2 = \mathrm{span}(\textbf{1}_p)$, and appealing to Lemma~\ref{lemma:3}, we have that:
\begin{eqnarray*}
\begin{aligned}
\max_{i}\|\proj_{\tilde{\mathcal{C}}}(e_i)\|_2 &\leq 2\max_i\max_{\substack{u_1 \in \mathrm{span}(\textbf{1}),\|u_1\|_2=1}} 2|u_1^Te_i|+\max_i\max_{\substack{u_2 \in \mathcal{C}',\|u_2\|_2=1}}|u_2^Te_i|\leq 2/\sqrt{p} + \inc[\mathcal{C}'].
\end{aligned}
\end{eqnarray*}
Analogously, letting $\mathcal{C}_1 = \mathcal{R}'$ and $\mathcal{C}_2 = \mathrm{span}(\textbf{1})$, and appealing to Lemma~\ref{lemma:3}, we have that:
\begin{eqnarray*}
\begin{aligned}
\max_{i}\|\proj_{\tilde{\mathcal{R}}}(e_i)\|_2 &\leq 2/\sqrt{p} + \inc[\mathcal{R}'].
\end{aligned}
\end{eqnarray*}
Since $\xi(T') \geq \max\{[\inc[\mathcal{C}'],\inc[\mathcal{R}']\}$ and $2\xi(T') \geq \frac{2}{\sqrt{p}}$, we conclude the desired result.
\end{proof}

\begin{lemma}Let $T'$ be a tangent space to the low-rank matrix variety with $\rho(T',T^\star)\leq \omega$ for some $\omega \in (0,1)$. Let $\mathbb{H}' = \Omega^\star \times T'$ and $\mathbb{Q}' = \Omega^\star \times (T' + \mathrm{span}(\ones))$. Then for any matrix $N \in \mathbb{R}^{p \times p}$, we have that $|\|\mathcal{P}_{\mathbb{H}'}(N)\|_2 - \|\mathcal{P}_{{\mathbb{Q}'}}(N)\|_2| \leq 2(\kappa^\star+\omega)$ and $|\|\mathcal{P}_{{\mathbb{H}'}^\perp}(N)\|_2 - \|\mathcal{P}_{{\mathbb{Q}'}^\perp}(N)\|_2| \leq 2(\kappa^\star+\omega)$.
\label{lemma:TQ_relation}
\end{lemma}
\begin{proof}
Decompose $N = N_1 + N_2$ where $N_1 \in \mathbb{Q}'$ and $N_2 \in {\mathbb{Q}'}^\perp$. Thus, $\mathcal{P}_{\mathbb{Q}'}(N) = N_1$. Furthermore, since $\mathbb{H}' \subseteq \mathbb{Q}'$, $\mathcal{P}_{\mathbb{H}'}(N) = \mathcal{P}_{\mathbb{H}'}(N_1)$. From Lemma \ref{lemma:4}, we have that $\|\mathcal{P}_{\mathbb{H}'}(N_1)\|_2 \geq \|N_1\|_2(1-2(\kappa^\star+\omega))$. Thus, $|\|\mathcal{P}_{\mathbb{Q}'}(N)\|_2 - \|\mathcal{P}_{\mathbb{H}'}(N)\|_2| \leq 2(\kappa^\star+\omega)$. Since for any tangent space to a low-rank variety $F \subseteq \mathbb{R}^{p \times p}$, $\|\mathcal{P}_{F^\perp}(N)\|_2 = \|N\|_2-\|\mathcal{P}_F(N)\|_2$, we can also conclude that $|\|\mathcal{P}_{{\mathbb{H}'}^\perp}(N)\|_2 - \|\mathcal{P}_{{\mathbb{Q}'}^\perp}(N)\|_2| \leq 2(\kappa^\star+\omega)$.
\end{proof}

\begin{lemma}Let $\mathbb{H}' = \Omega^\star \times T'$ where $\rho(T',T^\star)\leq \omega$. Let $ \kappa^\star = \|\proj_{{T^\star}^\perp}(1/p\ones)\|_2$. Suppose that $\min_{\mathbb{Q}'\in U(\omega)}\chi(\mathbb{Q}',\|\cdot\|_{\Phi_\gamma}) > 2(\kappa^\star+\omega)$. Let $F= \proj_{{T^\star}^\perp}(1/p\ones)/\|\proj_{{T^\star}^\perp}(1/p\ones)\|_2$. Then, we have the following results:
\begin{eqnarray*}
\begin{aligned}
&\min_{\substack{Z \in \mathbb{H}'\\\|Z\|_{\Phi_\gamma} = 1\\\rho(T',T^\star)\leq \omega}}\|\mathcal{P}_{\mathbb{H}'}\mathcal{A}^\dagger\mathbb{I}^\star\mathcal{A}\mathcal{P}_{\mathbb{H}'}(Z)\|_{\Phi_\gamma} \geq \min_{\mathbb{Q}'\in U(\omega)}\chi(\mathbb{Q}',\|\cdot\|_{\Phi_\gamma})-2(\kappa^\star+\omega),\\
&\max_{\substack{Z \in \mathbb{H}'\\\|Z\|_{\Phi_\gamma} = 1\\\mathbb{Q}'\in U(\omega)}}\|\mathcal{P}_{{\mathbb{H}'}^\perp}\mathcal{A}^\dagger\mathbb{I}^\star\mathcal{A}\mathcal{P}_{\mathbb{Q}'}(\mathcal{P}_{{\mathbb{H}'}}\mathcal{A}^\dagger\mathbb{I}^\star\mathcal{A}\mathcal{P}_{\mathbb{Q}'})^{-1}(Z)\|_{\Phi_\gamma} \leq \max_{\mathbb{Q}'\in U(\omega)}\varphi(\mathbb{Q}',\|\cdot\|_{\Phi_\gamma})+2(\kappa^\star+\omega),\\
&\max_{\substack{Z \in \mathbb{Q}'\\\|Z\|_{\Phi_\gamma} = 1\\\mathbb{Q}'\in U(\omega)}}\|(\mathcal{P}_{{\mathbb{H}'}}\mathcal{A}^\dagger\mathbb{I}^\star\mathcal{A}\mathcal{P}_{\mathbb{H}'})^{-1}\mathcal{P}_{{\mathbb{H}'}}\mathcal{A}^\dagger\mathbb{I}^\star\mathcal{A}\mathcal{P}_{{\mathbb{H}'}^\perp}(Z)\|_{\Phi_\gamma} \leq \frac{4(\kappa^\star+\omega)\max\{\gamma,1\}(\|\mathbb{I}^\star(F)\|_2+\|\mathbb{I}^\star\|_2\omega)}{\min_{\mathbb{Q}'\in U(\omega)}\chi(\mathbb{Q}',\|\cdot\|_{\Phi_\gamma})-2(\kappa^\star+\omega)}.
\end{aligned}
\end{eqnarray*}
where the linear operators $\mathcal{A},\mathcal{A}^\dagger$, the norm $\Phi_\gamma$, the set $U(\omega)$, and the functions $\chi$ and $\varphi$ are defined in Section \ref{sec:consistency}.
\label{lemma:Hessian_cond}
\end{lemma}
\begin{proof}
To prove the first part, consider $Z \in \mathbb{H}'$. Thus, $\mathcal{P}_{\mathbb{Q}'}(Z) = Z$. Combining this with Lemma~\ref{lemma:TQ_relation} and noting that the first components of $\mathbb{H}'$ and $\mathbb{Q}'$ are identical, we find that $|\|\mathcal{P}_{\mathbb{H}'}\mathcal{A}^\dagger\mathbb{I}^\star\mathcal{A}\mathcal{P}_{\mathbb{H}'}(Z)\|_{\Phi_\gamma}-\|\mathcal{P}_{\mathbb{Q}'}\mathcal{A}^\dagger\mathbb{I}^\star\mathcal{A}\mathcal{P}_{\mathbb{Q}'}(Z)\|_{\Phi_\gamma}| \leq 2(\kappa^\star+\omega)$. This results allows us to conclude the first paper. 

To prove the second part, let $Z \in \mathbb{H}'$ with $\|Z\|_{\Phi_\gamma} \leq 1$. We first notice that by appealing to Lemma~\ref{lemma:TQ_relation}, we have that: $\|\mathcal{P}_{{\mathbb{H}'}}\mathcal{A}^\dagger\mathbb{I}^\star\mathcal{A}\mathcal{P}_{\mathbb{Q}'}(Z)\|_{\Phi_\gamma} \geq \|\mathcal{P}_{{\mathbb{Q}'}}\mathcal{A}^\dagger\mathbb{I}^\star\mathcal{A}\mathcal{P}_{\mathbb{Q}'}(Z)\|_{\Phi_\gamma}-2(\kappa^\star+\omega) > 0$. Thus, the operator $\mathcal{P}_{{\mathbb{H}'}}\mathcal{A}^\dagger\mathbb{I}^\star\mathcal{A}\mathcal{P}_{\mathbb{Q}'}$ is invertible. Furthermore, suppose that there exists $N_1 \in \mathbb{Q}',N_2\in \mathbb{Q}'$ such that $\mathcal{P}_{{\mathbb{H}'}}\mathcal{A}^\dagger\mathbb{I}^\star\mathcal{A}\mathcal{P}_{\mathbb{Q}'}(N_1) = \mathcal{P}_{{\mathbb{Q}'}}\mathcal{A}^\dagger\mathbb{I}^\star\mathcal{A}\mathcal{P}_{\mathbb{Q}'}(N_2)$. Since $\mathbb{H}' \subseteq \mathbb{Q}'$, we have that then: $\mathcal{P}_{{\mathbb{H}'}}\mathcal{A}^\dagger\mathbb{I}^\star\mathcal{A}\mathcal{P}_{\mathbb{Q}'}(N_1-N_2) = 0$, which allows us to conclude that for any $Z \in \mathbb{H}'$,  $(\mathcal{P}_{{\mathbb{H}'}}\mathcal{A}^\dagger\mathbb{I}^\star\mathcal{A}\mathcal{P}_{\mathbb{Q}'})^{-1}(Z) = (\mathcal{P}_{{\mathbb{Q}'}}\mathcal{A}^\dagger\mathbb{I}^\star\mathcal{A}\mathcal{P}_{\mathbb{Q}'})^{-1}(Z)$. Appealing to Lemma~\ref{lemma:TQ_relation}, we have that for any $N \in \mathbb{Q}'$, $|\mathcal{P}_{{\mathbb{H}'}^\perp}\mathcal{A}^\dagger\mathbb{I}^\star\mathcal{A}\mathcal{P}_{\mathbb{Q}'}(N)-\mathcal{P}_{{\mathbb{Q}'}^\perp}\mathcal{A}^\dagger\mathbb{I}^\star\mathcal{A}\mathcal{P}_{\mathbb{Q}'}(N)| \leq 2(\kappa^\star+\omega)$, which allows us to conclude the desired result. 

To prove the third part, Consider any $Z \in \mathbb{Q}'$ with $Z_2$ denoting its second component which is contained in $T' \oplus \mathrm{span}(\ones)$. Let $Z_2 = Z_{21}+Z_{22}$ where $Z_{21} \in T'$ and $Z_{22} \in \mathrm{span}(\ones)$. Notice that $\proj_{{H'}^\perp}(Z) = \proj_{{{T}'}^\perp}(Z_{22})$. By Lemma~\ref{lemma:4}, $\|Z_{22}\|_2 \leq \kappa^\star+\omega$. Furthermore, $\proj_{{T'}^\perp}(Z_{22}) = (\proj_{{T'}^\perp}-\proj_{{T^\star}^\perp})(Z_{22})+\proj_{{T^\star}^\perp})(Z_{22})$. Thus, using the fact that $\|\proj_{T'}(M)\|_2 \leq 2\|M\|_2$ for any matrix $M$, we have that: $\mathcal{P}_{{\mathbb{H}'}}\mathcal{A}^\dagger\mathbb{I}^\star\mathcal{A}\mathcal{P}_{{\mathbb{H}'}^\perp}(Z)\|_{\Phi_\gamma} \leq 4(\kappa^\star+\omega)\max\{\gamma,1\}(\|\mathbb{I}^\star(F)\|_2+\|\mathbb{I}^\star\|_2\omega)$. Then, appealing to the first part of the Lemma, we have the desired result.

\end{proof}

\begin{lemma}Let $\mathbb{Q}^\star = \Omega^\star \times (T^\star + \mathrm{span}(\ones))$ and $\mathbb{Q}'= \Omega^\star \times (T'+ \mathrm{span}(\ones))$ where $\rho(T',T^\star)\leq \omega)$. Then, for any $Z \in \mathbb{R}^{p \times p} \times \mathbb{R}^{p \times p}$ with $\|Z\|_{\Phi_\gamma}=1$, 
$$|\|\proj_{\mathbb{Q}'}(Z)\|_{\Phi_\gamma}-\|\proj_{\mathbb{Q}^\star}(Z)\|_{\Phi_\gamma}| \leq 5\omega+4\kappa^\star.$$
\begin{proof}
Notice that:
\begin{eqnarray*}
\begin{aligned}
\|\proj_{\mathbb{Q}^\star}(Z)\|_{\Phi_\gamma} - \|(\proj_{\mathbb{Q}'}(Z)-\proj_{\mathbb{Q}^\star})(Z)\|_{\Phi_\gamma} \leq \|\proj_{\mathbb{Q}'}(Z)\|_{\Phi_\gamma} \leq  \|\proj_{\mathbb{Q}^\star}(Z)\|_{\Phi_\gamma} +\|(\proj_{\mathbb{Q}'}(Z)-\proj_{\mathbb{Q}^\star})(Z)\|_{\Phi_\gamma}.
\end{aligned}
\end{eqnarray*}
Further, letting $Z = (Z_1,Z_2)$ with $\|Z_2\|_2/\gamma \leq 1$:
\begin{eqnarray*}
\begin{aligned}
\|(\proj_{\mathbb{Q}'}-\proj_{\mathbb{Q}^\star})(Z)\|_{\Phi_\gamma} &= \frac{1}{\gamma}\|(\proj_{T'\oplus \mathrm{span}(\ones)}-\proj_{T^\star\oplus \mathrm{span}(\ones)})(Z_2)\|_{2}\\
&\leq 4(\kappa^\star+\omega) + \frac{1}{\gamma}\|(\proj_{T'}-\proj_{T^\star})(Z_2)\|_{2} \leq 4\kappa^\star+5\omega.
\end{aligned}
\end{eqnarray*}
Combining the results proves our result. 
\end{proof}
\label{lemma:Q_Q'}
\end{lemma}

\begin{lemma}Let $\mathbb{Q}^\star = \Omega^\star \times (T^\star + \mathrm{span}(\ones))$ and $\mathbb{Q}'= \Omega^\star \times (T'+ \mathrm{span}(\ones))$ where $\rho(T',T^\star)\leq \omega)$. Suppose
\begin{eqnarray*}
\min_{Z \in \mathbb{Q}^\star, \|Z\|_{\Phi_\gamma} = 1}\|\proj_{\mathbb{Q}^\star}\Gla\I\Gl\proj_{\mathbb{Q}^\star}(Z)\|_{\Phi_\gamma} \geq \beta.
\end{eqnarray*}
Then, 
\begin{eqnarray*}
\begin{aligned}
\min_{Z \in \mathbb{Q}', \|Z\|_{\Phi_\gamma} = 1}\|\proj_{\mathbb{Q}'}\Gla\I\Gl\proj_{\mathbb{Q}'}(Z)\|_{\Phi_\gamma} &\geq \beta(1-(4\kappa^\star+5\omega)) \\&- 2(5\omega+4\kappa^\star)\|\I\|_2\max\{\gamma,1\}\\&-(4\kappa^\star+5\omega)\|\I\|_2(d^\star/\gamma+1).
\end{aligned}
\end{eqnarray*}
Additionally, for any $Z \in \mathbb{R}^{p \times p} \times \mathbb{R}^{p \times p}$ with $\|Z\|_{\Phi_\gamma} = 1$:
\begin{eqnarray*}
\begin{aligned}
|\|\proj_{\mathbb{Q}'}\Gla\I\Gl\proj_{\mathbb{Q}'}(Z)\|_{\Phi_\gamma}-
\|\proj_{\mathbb{Q}^\star}\Gla\I\Gl\proj_{\mathbb{Q}^\star}(Z)\|_{\Phi_\gamma}| &\leq  2(5\omega+4\kappa^\star)\|\I\|_2\max\{\gamma,1\}\\&+(4\kappa^\star+5\omega)\|\I\|_2(d^\star/\gamma+1).
\end{aligned}
\end{eqnarray*}
\label{lemma:Q'_Qstar}
\end{lemma}

\begin{proof}Consider any $Z$ with $\|Z\|_{\Phi_\gamma} = 1$. Then, 
\begin{eqnarray*}
\begin{aligned}
\|\proj_{\mathbb{Q}'}\Gla\I\Gl\proj_{\mathbb{Q}'}(Z)\|_{\Phi_\gamma} &\geq \|\proj_{\mathbb{Q}'}\Gla\I\Gl\proj_{\mathbb{Q}^\star}(Z)\|_{\Phi_\gamma}-\|\proj_{\mathbb{Q}'}\Gla\I\Gl(\proj_{\mathbb{Q}'}-\proj_{\mathbb{Q}^\star})(Z)\|_{\Phi_\gamma}\\
&\geq \|\proj_{\mathbb{Q}^\star}\Gla\I\Gl\proj_{\mathbb{Q}^\star}(Z)\|_{\Phi_\gamma}-\|\proj_{\mathbb{Q}'}\Gla\I\Gl(\proj_{\mathbb{Q}'}-\proj_{\mathbb{Q}^\star})(Z)\|_{\Phi_\gamma}\\&-\|(\proj_{\mathbb{Q}'}-\proj_{\mathbb{Q}^\star})\Gla\I\Gl\proj_{\mathbb{Q}^\star}(Z)\|_{\Phi_\gamma}
\end{aligned}
\end{eqnarray*} 
Some algebra and appealing to Lemma~\ref{lemma:Q_Q'} leads to the conclusions that
\begin{eqnarray*}
\begin{aligned}
\|\proj_{\mathbb{Q}'}\Gla\I\Gl(\proj_{\mathbb{Q}'}-\proj_{\mathbb{Q}^\star})(Z)\|_{\Phi_\gamma} &\leq 2(5\omega+4\kappa^\star)\|\I\|_2\max\{\gamma,1\},
\end{aligned}
\end{eqnarray*}
and that $\|\proj_{\mathbb{Q}^\star}(Z)\|_{\Phi_\gamma} \geq (1-(4\kappa^\star+5\omega))$. Furthermore, denote $Z_1 = \proj_{\Omega^\star}(Z)$ and $Z_2 =  \proj_{T^\star\oplus \mathrm{span}(\ones)}(Z)$. Notice that $\|Z_1\|_2 \leq \|Z_1\|_\infty\theta(\Omega^\star)\leq \|Z_1\|_\infty{d^\star}$. Again appealing to Lemma~\ref{lemma:Q_Q'}:
\begin{eqnarray*}
\begin{aligned}
\|(\proj_{\mathbb{Q}'}-\proj_{\mathbb{Q}^\star})\Gla\I\Gl\proj_{\mathbb{Q}^\star}(Z)\|_{\Phi_\gamma} &\leq (4\kappa^\star+5\omega)\|\I\|_2(d^\star/\gamma+1)
\end{aligned}
\end{eqnarray*}
Putting things together, we have the first desired result. The second desired result follows from a similar analysis as the first part.
\end{proof}

\begin{lemma}We begin with the following lemmas where we let $\mathbb{Q}^\star = \Omega^\star \times (T^\star + \mathrm{span}(\ones))$ and $\mathbb{Q}'= \Omega^\star \times (T'+ \mathrm{span}(\ones))$ where $\rho(T',T^\star)\leq \omega)$. Suppose
\begin{eqnarray*}
\begin{aligned}
\min_{Z \in \mathbb{Q}^\star, \|Z\|_{\Phi_\gamma} = 1}\|\proj_{\mathbb{Q}^\star}\Gla\I\Gl\proj_{\mathbb{Q}^\star}(Z)\|_{\Phi_\gamma} \geq \beta &> \frac{2(5\omega+4\kappa^\star)\|\I\|_2\max\{\gamma,1\}}{1-(4\kappa^\star+5\omega)}\\&+\frac{(d^\star+\gamma)(4\kappa^\star+5\omega)\|\I\|_2\max\{\gamma,1\}}{1-(4\kappa^\star+5\omega)},
\end{aligned}
\end{eqnarray*}
and 
\begin{eqnarray*}
\max_{Z \in \mathbb{Q}^\star, \|Z\|_{\Phi_\gamma} = 1}\|\proj_{{\mathbb{Q}^\star}^\perp}\Gla\I\Gl\proj_{\mathbb{Q}^\star}(\proj_{{\mathbb{Q}^\star}}\Gla\I\Gl\proj_{\mathbb{Q}^\star})^{-1}(Z)\|_{\Phi_\gamma} \leq \zeta,
\end{eqnarray*}
and 
\begin{eqnarray*}
\max_{Z \in \mathbb{Q}^\star, \|Z\|_{\Phi_\gamma} = 1}\|\proj_{{\mathbb{Q}^\star}^\perp}\Gla\I\Gl\proj_{\mathbb{Q}^\star}(Z)\|_{\Phi_\gamma} \leq \delta.
\end{eqnarray*}
Then, 
\begin{eqnarray*}
\begin{aligned}
&\max_{Z \in \mathbb{Q}', \|Z\|_{\Phi_\gamma} = 1}\|\proj_{{\mathbb{Q}'}^\perp}\Gla\I\Gl\proj_{\mathbb{Q}^\star}(\proj_{{\mathbb{Q}'}}\Gla\I\Gl\proj_{\mathbb{Q}'})^{-1}(Z)\|_{\Phi_\gamma} \leq \\& (2\delta + 2\|\I\|\max\{\gamma,1\}(4\kappa^\star+5\omega)+\|\I\|(d^\star/\gamma+1)(4\kappa^\star+5\omega))\frac{\|\Delta\|_{\Phi_\gamma}+4\kappa^\star+5\omega}{\beta}\\&+ \zeta(1+5\omega+4\kappa^\star)+ \|\I\|_2\max\{\gamma,1\}(5\omega+4\kappa^\star)\frac{1+5\omega+4\kappa^\star}{\beta}\\&+  \|\I\|_2\max\{\gamma,1\}(5\omega+4\kappa^\star)(1+d^\star/\gamma)\frac{1+5\omega+4\kappa^\star}{\beta},\\
\end{aligned}
\end{eqnarray*}
where, 
$$\|\Delta\|_{\Phi_\gamma} \leq \frac{1}{\tilde{\beta}}\left(2(5\omega+4\kappa^\star)\|\I\|_2\max\{\gamma,1\}+(d^\star+\gamma)(4\kappa^\star+5\omega)\|\I\|_2\max\{\gamma,1\}\right),$$
and 
$$\tilde{\beta}:= \beta-{2(5\omega+4\kappa^\star)\|\I\|_2\max\{\gamma,1\}}-{(d^\star+\gamma)(4\kappa^\star+5\omega)\|\I\|_2\max\{\gamma,1\}}.$$
\label{lemma:Q'_Q'perp}
\end{lemma} 

\begin{proof}
Take $Z \in \mathbb{Q}'$. Let $Z = \proj_{{\mathbb{Q}'}}\Gla\I\Gl\proj_{\mathbb{Q}'}(A)$. Define $\Delta:= \proj_{{\mathbb{Q}^\star}}\Gla\I\Gl\proj_{\mathbb{Q}^\star}(A)-Z$. Then, by Lemma~\ref{lemma:Q'_Qstar}
$$\|\Delta\|_{\Phi_\gamma} \leq \|A\|_{\Phi_\gamma}\left(2(5\omega+4\kappa^\star)\|\I\|_2\max\{\gamma,1\}+(4\kappa^\star+5\omega)\|\I\|_2(d^\star/\gamma+1)\right),$$
and that $\|A\|_{\Phi_\gamma} \leq 1/\tilde{\beta}$ where 
\begin{eqnarray*}
\begin{aligned}
\tilde{\beta}&:= \beta(1-(4\kappa^\star+5\omega)) \\&- 2(5\omega+4\kappa^\star)\|\I\|_2\max\{\gamma,1\}\\&-(4\kappa^\star+5\omega)\|\I\|_2(d^\star/\gamma+1).
\end{aligned}
\end{eqnarray*}
Let $B = (\proj_{{\mathbb{Q}^\star}}\Gla\I\Gl\proj_{\mathbb{Q}^\star})^{-1}\proj_{\mathbb{Q}^\star}(Z)$. Then, $\proj_{{\mathbb{Q}^\star}}\Gla\I\Gl\proj_{\mathbb{Q}^\star}(A-B) = \Delta+\proj_{{\mathbb{Q}^\star}^\perp}(Z)$, which implies that:
$$\|A-B\|_{\Phi_\gamma} \leq \frac{\|\Delta\|_{\Phi_\gamma}+\|\proj_{{\mathbb{Q}^\star}^\perp}(Z)\|_{\Phi_\gamma}}{\beta} \leq \frac{\|\Delta\|_{\Phi_\gamma}+4\kappa^\star+5\omega}{\beta}.$$
Note that:
\begin{eqnarray*}
\begin{aligned}
&\|\proj_{{\mathbb{Q}'}^\perp}\Gla\I\Gl\proj_{\mathbb{Q}'}(\proj_{{\mathbb{Q}'}}\Gla\I\Gl\proj_{\mathbb{Q}'})^{-1}(Z)\|_{\Phi_\gamma} \leq\\
&\underbrace{\|\proj_{{\mathbb{Q}'}^\perp}\Gla\I\Gl\proj_{\mathbb{Q}'}(\proj_{{\mathbb{Q}^\star}}\Gla\I\Gl\proj_{\mathbb{Q}^\star})^{-1}\proj_{\mathbb{Q}^\star}(Z)\|_{\Phi_\gamma}}_{T1}\\
&+ \underbrace{\|\proj_{{\mathbb{Q}'}^\perp}\Gla\I\Gl\proj_{\mathbb{Q}'}\left((\proj_{{\mathbb{Q}'}}\Gla\I\Gl\proj_{\mathbb{Q}'})^{-1}(Z) - (\proj_{{\mathbb{Q}^\star}}\Gla\I\Gl\proj_{\mathbb{Q}^\star})^{-1}\proj_{\mathbb{Q}^\star}(Z)\right)\|_{\Phi_\gamma}}_{T2}.
\end{aligned}
\end{eqnarray*}
Notice that for any $M \in \mathbb{R}^{p \times p} \times \mathbb{R}^{p \times p}$, appealing to Lemma~\ref{lemma:Q_Q'}
\begin{eqnarray*}
\begin{aligned}
\|\proj_{{\mathbb{Q}'}^\perp}\Gla\I\Gl\proj_{\mathbb{Q}'}(M)\|_{\Phi_\gamma} &\leq \|\proj_{{\mathbb{Q}^\star}^\perp}\Gla\I\Gl\proj_{\mathbb{Q}^\star}(M)\|_{\Phi_\gamma},\\&+  \|\proj_{{\mathbb{Q}'}^\perp}\Gla\I\Gl(\proj_{\mathbb{Q}'}-\proj_{\mathbb{Q}^\star})(M)\|_{\Phi_\gamma}\\&+\|(\proj_{{\mathbb{Q}'}^\perp}-\proj_{{\mathbb{Q}^\star}^\perp})\Gla\I\Gl\proj_{\mathbb{Q}^\star}(M)\|_{\Phi_\gamma},\\
&\leq 2\delta + 2\|\I\|\max\{\gamma,1\}(4\kappa^\star+5\omega)+\|\I\|(d^\star/\gamma+1)(4\kappa^\star+5\omega).
\end{aligned}
\end{eqnarray*}

\begin{eqnarray*}
T_2 \leq (2\delta + 2\|\I\|\max\{\gamma,1\}(4\kappa^\star+5\omega)+\|\I\|(d^\star/\gamma+1)(4\kappa^\star+5\omega))\|A-B\|_{\Phi_\gamma}.
\end{eqnarray*}
To control $T_1$, notice that for any $M \in \mathbb{R}^{p \times p} \times \mathbb{R}^{p \times p}$, appealing to Lemma~\ref{lemma:Q_Q'}
\begin{eqnarray*}
\begin{aligned}
\|\proj_{{\mathbb{Q}'}^\perp}\Gla\I\Gl\proj_{\mathbb{Q}'}(M)\|_{\Phi_\gamma} &\leq \|\proj_{{\mathbb{Q}'}^\perp}\Gla\I\Gl\proj_{\mathbb{Q}^\star}(M)\|_{\Phi_\gamma} + \|\I\|_2\max\{\gamma,1\}(5\omega+4\kappa^\star)\\
&\leq \|\proj_{{\mathbb{Q}^\star}^\perp}\Gla\I\Gl\proj_{\mathbb{Q}^\star}(M)\|_{\Phi_\gamma}+ \|(\proj_{{\mathbb{Q}'}^\perp}-\proj_{{\mathbb{Q}^\star}^\perp})\Gla\I\Gl\proj_{\mathbb{Q}^\star}(M)\|_{\Phi_\gamma} \\ &+ \|\I\|_2\max\{\gamma,1\}(5\omega+4\kappa^\star)\|\proj_{\mathbb{Q}^\star}(M)\|_{\Phi_\gamma}\\
&\leq \|\proj_{{\mathbb{Q}^\star}^\perp}\Gla\I\Gl\proj_{\mathbb{Q}^\star}(M)\|_{\Phi_\gamma}\\&+ \|\I\|_2\max\{\gamma,1\}(5\omega+4\kappa^\star)(1+d^\star/\gamma)\|\proj_{\mathbb{Q}^\star}(M)\|_{\Phi_\gamma}\\ &+ \|\I\|_2\max\{\gamma,1\}(5\omega+4\kappa^\star)\|\proj_{\mathbb{Q}^\star}(M)\|_{\Phi_\gamma}
\end{aligned}
\end{eqnarray*}
Setting $M = (\proj_{{\mathbb{Q}^\star}}\Gla\I\Gl\proj_{\mathbb{Q}^\star})^{-1}\proj_{\mathbb{Q}^\star}(Z)$ and noting that $\|M\|_{\Phi_\gamma} \leq \frac{1+5\omega+4\kappa^\star}{\beta}$, we have the following bound for $T_1$:
\begin{eqnarray*}
\begin{aligned}
T_1 &\leq \zeta(1+5\omega+4\kappa^\star)+ \|\I\|_2\max\{\gamma,1\}(5\omega+4\kappa^\star)\frac{1+5\omega+4\kappa^\star}{\beta}\\&+  \|\I\|_2\max\{\gamma,1\}(5\omega+4\kappa^\star)(1+d^\star/\gamma)\frac{1+5\omega+4\kappa^\star}{\beta}
\end{aligned}
\end{eqnarray*}
Combining the bounds on $T_1$ and $T_2$, we have the desired result. 
\end{proof}

{\section{{A numerical approach to verifying Assumptions~\ref{ass:1}-\ref{ass:3}}}}
\label{sec:numerical_approach}
In our numerical approach, we obtain lower bound for
\begin{eqnarray}
\begin{aligned}
&\min_{Z \in \mathbb{Q}^\star, \|Z\|_{\Phi_\gamma}=1}\|\proj_{{\mathbb{Q}}^\star}\Gla\I\Gl\proj_{\mathbb{Q}^\star}(Z)\|_{\Phi_\gamma},
\end{aligned}
\label{eqn:temp1_numerical_approach}
\end{eqnarray}
and and an upper bound for
\begin{eqnarray}
\begin{aligned}
\max_{Z \in \mathbb{Q}^\star, \|Z\|_{\Phi_\gamma}=1}\|\proj_{{\mathbb{Q}^\star}^\perp}\Gla\I\Gl\proj_{\mathbb{Q}^\star}(\proj_{{\mathbb{Q}^\star}}\Gla\I\Gl\proj_{\mathbb{Q}^\star})^{-1}(Z)\|_{\Phi_\gamma}. 
\end{aligned}
\label{eqn:temp2_numerical_approach}
\end{eqnarray}
We then can appeal to Lemmas~\ref{lemma:Q'_Qstar}-\ref{lemma:Q'_Q'perp} to quantify the quantities in Assumptions~\ref{ass:1}-\ref{ass:3}. To evaluate \eqref{eqn:temp1_numerical_approach}, consider $Z = (Z_1,Z_2)$ where $Z_1 \in \Omega^\star$ with $\|Z_1\|_\infty = 1$ and $Z_2 \in T^\star \oplus \mathrm{span}(\ones)$ with $\|Z_2\|_2 \leq \gamma$. Then, 
\begin{eqnarray*}
\begin{aligned}
\|\proj_{{\mathbb{Q}}^\star}\Gla\I\Gl\proj_{\mathbb{Q}^\star}(Z)\|_{\Phi_\gamma} &\geq \|\proj_{\Omega^\star}\I(Z_1)\|_{\infty} - \|\proj_{\Omega^\star}\I(Z_2)\|_{\infty} \\&\geq \min_{Z_1 \in\Omega^\star, \|Z_1\|_\infty=1}\|\proj_{\Omega^\star}\I(Z_1)\|_{\infty} - \gamma\max_{Z_2 \in T^\star\oplus \mathrm{span}(\ones), \|Z_2\|_2=1}\|\proj_{\Omega^\star}\I(Z_2)\|_{\infty}.  
\end{aligned}
\end{eqnarray*}
Now consider $Z = (Z_1,Z_2)$ where $Z_1 \in \Omega^\star$ with $\|Z_1\|_\infty \leq 1$ and $Z_2 \in T^\star \oplus \mathrm{span}(\ones)$ with $\|Z_2\|_2 = \gamma$. Then, 
\begin{eqnarray*}
\begin{aligned}
\|\proj_{{\mathbb{Q}}^\star}\Gla\I\Gl\proj_{\mathbb{Q}^\star}(Z)\|_{\Phi_\gamma} &\geq \frac{1}{\gamma}\|\proj_{T^\star \oplus \mathrm{span}(\ones)}\I(Z_2)\|_{2} - \frac{1}{\gamma}\|\proj_{T^\star \oplus \mathrm{span}(\ones)}\I(Z_1)\|_{2} \\&\geq \min_{\substack{Z_2 \in T^\star \oplus \mathrm{span}(\ones)\\ \|Z_2\|_\infty=1}}\|\proj_{T^\star \oplus \mathrm{span}(\ones)}\I(Z_2)\|_{2} \\&- \frac{1}{\gamma}\max_{\substack{Z_1 \in \Omega^\star\\ \|Z_1\|_\infty = 1}}\|\proj_{T^\star \oplus \mathrm{span}(\ones)}\I(Z_1)\|_{2}.  
\end{aligned}
\end{eqnarray*}
Thus, we obtain the following lower bound for \eqref{eqn:temp1_numerical_approach}:
\begin{eqnarray*}
\begin{aligned}
&\min_{Z \in \mathbb{Q}^\star, \|Z\|_{\Phi_\gamma}=1}\|\proj_{{\mathbb{Q}}^\star}\Gla\I\Gl\proj_{\mathbb{Q}^\star}(Z)\|_{\Phi_\gamma} \geq\\
& \min\Bigg\{\min_{Z_1 \in\Omega^\star, \|Z_1\|_\infty=1}\|\proj_{\Omega^\star}\I(Z_1)\|_{\infty} - \gamma\max_{Z_2 \in T^\star\oplus \mathrm{span}(\ones), \|Z_2\|_2=1}\|\proj_{\Omega^\star}\I(Z_2)\|_{\infty},\\
&~~~\min_{\substack{Z_2 \in T^\star \oplus \mathrm{span}(\ones)\\ \|Z_2\|_\infty=1}}\|\proj_{T^\star \oplus \mathrm{span}(\ones)}\I(Z_2)\|_{2}- \frac{1}{\gamma}\max_{\substack{Z_1 \in \Omega^\star\\ \|Z_1\|_\infty = 1}}\|\proj_{T^\star \oplus \mathrm{span}(\ones)}\I(Z_1)\|_{2} \Bigg\}.
\end{aligned}
\end{eqnarray*}
The individual terms above are computed approximately by sampling. To obtain an upper bound for \eqref{eqn:temp2_numerical_approach}, consider $Z = (Z_1,Z_2)$ where $Z_1 \in \Omega^\star$ with $\|Z_1\|_\infty = 1$ and $Z_2 \in T^\star \oplus \mathrm{span}(\ones)$ with $\|Z_2\|_2 \leq \gamma$. Then, 
\begin{eqnarray*}
\begin{aligned}
&\|\proj_{{\mathbb{Q}^\star}^\perp}\Gla\I\Gl\proj_{\mathbb{Q}^\star}(\proj_{{\mathbb{Q}^\star}}\Gla\I\Gl\proj_{\mathbb{Q}^\star})^{-1}(Z)\|_{\Phi_\gamma} \\&\leq\|\proj_{{\Omega^\star}^\perp}\I\Gl\proj_{\mathbb{Q}^\star}(\proj_{{\mathbb{Q}^\star}}\Gla\I\Gl\proj_{\mathbb{Q}^\star})^{-1}(Z_1,0)\|_\infty \\
&+ \|\proj_{{\Omega^\star}^\perp}\Gla\I\Gl\proj_{\mathbb{Q}^\star}(\proj_{{\mathbb{Q}^\star}}\Gla\I\Gl\proj_{\mathbb{Q}^\star})^{-1}(0,Z_2)\|_\infty \\
&\leq \max_{Z_1 \in \Omega^\star, \|Z_1\|_\infty = 1}\|\proj_{{\Omega^\star}^\perp}\I\Gl\proj_{\mathbb{Q}^\star}(\proj_{{\mathbb{Q}^\star}}\Gla\I\Gl\proj_{\mathbb{Q}^\star})^{-1}(Z_1,0)\|_\infty \\
&+ \gamma\max_{Z_2 \in T^\star \oplus \mathrm{span}(\ones), \|Z_2\|=1} \|\proj_{{\Omega^\star}^\perp}\Gla\I\Gl\proj_{\mathbb{Q}^\star}(\proj_{{\mathbb{Q}^\star}}\Gla\I\Gl\proj_{\mathbb{Q}^\star})^{-1}(0,Z_2)\|_\infty  
\end{aligned}
\end{eqnarray*}
Now consider $Z = (Z_1,Z_2)$ where $Z_1 \in \Omega^\star$ with $\|Z_1\|_\infty \;e 1$ and $Z_2 \in T^\star \oplus \mathrm{span}(\ones)$ with $\|Z_2\|_2 = \gamma$. Then, 
\begin{eqnarray*}
\begin{aligned}
&\|\proj_{{\mathbb{Q}^\star}^\perp}\Gla\I\Gl\proj_{\mathbb{Q}^\star}(\proj_{{\mathbb{Q}^\star}}\Gla\I\Gl\proj_{\mathbb{Q}^\star})^{-1}(Z)\|_{\Phi_\gamma} \\&\leq\frac{1}{\gamma}\|\proj_{{T^\star\oplus\mathrm{span}(\ones)}^\perp}\I\Gl\proj_{\mathbb{Q}^\star}(\proj_{{\mathbb{Q}^\star}}\Gla\I\Gl\proj_{\mathbb{Q}^\star})^{-1}(0,Z_2)\|_2 \\
&+ \frac{1}{\gamma}\|\proj_{{T^\star\oplus\mathrm{span}(\ones)}^\perp}\Gla\I\Gl\proj_{\mathbb{Q}^\star}(\proj_{{\mathbb{Q}^\star}}\Gla\I\Gl\proj_{\mathbb{Q}^\star})^{-1}(Z_1,0)\|_2\\
&\leq \max_{Z_2 \in T^\star\oplus\mathrm{span}(\ones), \|Z_2\|_2 = 1}\|\proj_{{T^\star\oplus\mathrm{span}(\ones)}^\perp}\I\Gl\proj_{\mathbb{Q}^\star}(\proj_{{\mathbb{Q}^\star}}\Gla\I\Gl\proj_{\mathbb{Q}^\star})^{-1}(0,Z_2)\|_2 \\
&+ \frac{1}{\gamma}\max_{Z_1 \in \Omega^\star, \|Z_1\|_\infty=1} \|\proj_{{T^\star\oplus\mathrm{span}(\ones)}^\perp}\Gla\I\Gl\proj_{\mathbb{Q}^\star}(\proj_{{\mathbb{Q}^\star}}\Gla\I\Gl\proj_{\mathbb{Q}^\star})^{-1}(Z_1,0)\|_2
\end{aligned}
\end{eqnarray*}
Thus, we obtain the following upper bound for \eqref{eqn:temp2_numerical_approach}:
\begin{eqnarray*}
\begin{aligned}
&\max_{Z \in \mathbb{Q}^\star, \|Z\|_{\Phi_\gamma}=1}\|\proj_{{\mathbb{Q}^\star}^\perp}\Gla\I\Gl\proj_{\mathbb{Q}^\star}(\proj_{{\mathbb{Q}^\star}}\Gla\I\Gl\proj_{\mathbb{Q}^\star})^{-1}(Z)\|_{\Phi_\gamma}\\
&\leq \max\Bigg\{\max_{Z_1 \in \Omega^\star, \|Z_1\|_\infty = 1}\|\proj_{{\Omega^\star}^\perp}\I\Gl\proj_{\mathbb{Q}^\star}(\proj_{{\mathbb{Q}^\star}}\Gla\I\Gl\proj_{\mathbb{Q}^\star})^{-1}(Z_1,0)\|_\infty \\
&+ \gamma\max_{Z_2 \in T^\star \oplus \mathrm{span}(\ones), \|Z_2\|=1} \|\proj_{{\Omega^\star}^\perp}\Gla\I\Gl\proj_{\mathbb{Q}^\star}(\proj_{{\mathbb{Q}^\star}}\Gla\I\Gl\proj_{\mathbb{Q}^\star})^{-1}(0,Z_2)\|_\infty\\&, \max_{Z_2 \in T^\star\oplus\mathrm{span}(\ones), \|Z_2\|_2 = 1}\|\proj_{{T^\star\oplus\mathrm{span}(\ones)}^\perp}\I\Gl\proj_{\mathbb{Q}^\star}(\proj_{{\mathbb{Q}^\star}}\Gla\I\Gl\proj_{\mathbb{Q}^\star})^{-1}(0,Z_2)\|_2 \\
&+ \frac{1}{\gamma}\max_{Z_1 \in \Omega^\star, \|Z_1\|_\infty=1} \|\proj_{{T^\star\oplus\mathrm{span}(\ones)}^\perp}\Gla\I\Gl\proj_{\mathbb{Q}^\star}(\proj_{{\mathbb{Q}^\star}}\Gla\I\Gl\proj_{\mathbb{Q}^\star})^{-1}(Z_1,0)\|_2\Bigg\}
\end{aligned}
\end{eqnarray*}
Again, the individual terms above are computed approximately by sampling. Finally, we note that appealing to Lemmas~\ref{lemma:Q'_Q'perp} involves computing an upper-bound for:
\begin{eqnarray}
\max_{Z \in \mathbb{Q}^\star, \|Z\|_{\Phi_\gamma}=1}\|\proj_{{{\mathbb{Q}}^\star}^\perp}\Gla\I\Gl\proj_{\mathbb{Q}^\star}(Z)\|_{\Phi_\gamma}. 
\label{eqn:temp3_numerical_approach}
\end{eqnarray} 
To obtain an upper bound, consider $Z = (Z_1,Z_2)$ where $Z_1 \in \Omega^\star$ with $\|Z_1\|_\infty = 1$ and $Z_2 \in T^\star \oplus \mathrm{span}(\ones)$ with $\|Z_2\|_2 \leq \gamma$. Then, 
\begin{eqnarray*}
\begin{aligned}
\|\proj_{{{\mathbb{Q}}^\star}^\perp}\Gla\I\Gl\proj_{\mathbb{Q}^\star}(Z)\|_{\Phi_\gamma} &\leq \|\proj_{{\Omega^\star}^\perp}\I(Z_1)\|_{\infty} + \|\proj_{{\Omega^\star}^\perp}\I(Z_2)\|_{\infty} \\&\leq \max_{Z_1 \in\Omega^\star, \|Z_1\|_\infty=1}\|\proj_{{\Omega^\star}^\perp}\I(Z_1)\|_{\infty} + \gamma\max_{\substack{Z_2 \in T^\star\oplus \mathrm{span}(\ones)\\ \|Z_2\|_2=1}}\|\proj_{{\Omega^\star}^\perp}\I(Z_2)\|_{\infty}.  
\end{aligned}
\end{eqnarray*}
Now consider $Z = (Z_1,Z_2)$ where $Z_1 \in \Omega^\star$ with $\|Z_1\|_\infty \leq 1$ and $Z_2 \in T^\star \oplus \mathrm{span}(\ones)$ with $\|Z_2\|_2 = \gamma$. Then, 
\begin{eqnarray*}
\begin{aligned}
\|\proj_{{{\mathbb{Q}}^\star}^\perp}\Gla\I\Gl\proj_{\mathbb{Q}^\star}(Z)\|_{\Phi_\gamma} &\leq \frac{1}{\gamma}\|\proj_{{T^\star \oplus \mathrm{span}(\ones)}^\perp}\I(Z_2)\|_{2} +\frac{1}{\gamma}\|\proj_{{T^\star \oplus \mathrm{span}(\ones)}^\perp}\I(Z_1)\|_{2} \\&\geq \max_{\substack{Z_2 \in T^\star \oplus \mathrm{span}(\ones)\\ \|Z_2\|_\infty=1}}\|\proj_{{T^\star \oplus \mathrm{span}(\ones)}^\perp}\I(Z_2)\|_{2} \\&+ \frac{1}{\gamma}\max_{\substack{Z_1 \in \Omega^\star\\ \|Z_1\|_\infty = 1}}\|\proj_{{T^\star \oplus \mathrm{span}(\ones)}^\perp}\I(Z_1)\|_{2}.  
\end{aligned}
\end{eqnarray*}
Thus, we obtain the following lower bound for \eqref{eqn:temp3_numerical_approach}:
\begin{eqnarray*}
\begin{aligned}
&\max_{Z \in \mathbb{Q}^\star, \|Z\|_{\Phi_\gamma}=1}\|\proj_{{\mathbb{Q}^\star}^\perp}\Gla\I\Gl\proj_{\mathbb{Q}^\star}(Z)\|_{\Phi_\gamma} \geq\\
& \max\Bigg\{\max_{Z_1 \in\Omega^\star, \|Z_1\|_\infty=1}\|\proj_{{\Omega^\star}^\perp}\I(Z_1)\|_{\infty} + \gamma\max_{\substack{Z_2 \in T^\star\oplus \mathrm{span}(\ones)\\ \|Z_2\|_2=1}}\|\proj_{{\Omega^\star}^\perp}\I(Z_2)\|_{\infty},\\&\max_{\substack{Z_2 \in T^\star \oplus \mathrm{span}(\ones)\\ \|Z_2\|_\infty=1}}\|\proj_{{T^\star \oplus \mathrm{span}(\ones)}^\perp}\I(Z_2)\|_{2} + \frac{1}{\gamma}\max_{\substack{Z_1 \in \Omega^\star\\ \|Z_1\|_\infty = 1}}\|\proj_{{T^\star \oplus \mathrm{span}(\ones)}^\perp}\I(Z_1)\|_{2} \Bigg\}.
\end{aligned}
\end{eqnarray*}

\section{Sufficient Hessian conditions and choice of $\gamma$ that satisfies Assumptions~\ref{ass:1}-\ref{ass:3}}
\label{sufficient_hessian}
\emph{Behavior of $\I$ with respect to $\Omega^\star$}. Let
\begin{eqnarray*}
\begin{aligned}
\alpha_\Omega &:=\min_{N \in \Omega^\star, \|N\|_\infty = 1}\|\proj_{\Omega^\star}\I\proj_{\Omega^\star}(N)\|_\infty, \\
\delta_{\Omega^\perp} &: \max_{N \in \Omega^\star, \|N\|_\infty = 1}\|\proj_{{\Omega^\star}^\perp}\I\proj_{\Omega^\star}(N)\|_\infty,\\
\beta_\Omega &:=\max_{N \in \Omega^\star, \|N\|_2 = 1}\|\I(N)\|_2,
%\alpha_T &:=\min_{N \in T', \rho(T',T^\star) \leq \frac{\min\{\kappa,\xi(T^\star)\}}{2}}\|\proj_{T^\star}\I\proj_{T^\star}(N)\|_2\\
%\delta_{T^\perp} &:=\max_{N \in T', \rho(T',T^\star) \leq \frac{\min\{\kappa,\xi(T^\star)\}}{2}, \|N\|_2=1}\max\left\{\|\proj_{{T'}^\perp}\I\proj_{{T'}}(N)\|_2, \|\proj_{{T'}}\I\proj_{{T'}^\perp}(N)\|_2\right\}\\
%\beta_T &:=\max_{N \in T', \rho(T',T^\star) \leq \frac{\min\{\kappa,\xi(T^\star)\}}{2}, \|N\|_\infty = 1}\|\I(N)\|_\infty\\
%\alpha'_{T+} &:=\min_{N \in T' \oplus T(\ones), \rho(T',T^\star) \leq \frac{\min\{\kappa,\xi(T^\star)\}}{2}}\|\proj_{T' \oplus T(\ones)}\I\proj_{T' \oplus T(\ones)}(N)\|_2 
\end{aligned}
\end{eqnarray*}
be functions $\I$ with respect to $\Omega^\star$. Here, $\alpha_\Omega$ quantifies the {minimum gain} of $\I$ restricted to subspace $\Omega^\star$ and with respect to the $\ell_\infty$ norm (the minimum gain of a matrix $M$ restricted to subspace $S$ and with respect to norm $\|\cdot\|$ is $\min_{x \in S, \|x\|=1}\|P_S{M}P_S(x)\|$); the quantity $\delta_\Omega$ computes the inner-product between elements in $\Omega^\star$ and ${\Omega^\star}^\perp$ as quantified by the metric induced by $\I$; and finally, $\beta_\Omega$ quantifies the behavior of $\I$ restricted to $\Omega^\star$ in spectral norm.

\emph{Behavior of $\I$ with respect to $T^\star$}. Similar to $\Omega^\star$, we control the behavior of $\I$ associated with the subspace $T^\star$. We control the behavior of $\I$ for tangent spaces $T'$ close to the tangent space $T^\star$:
\begin{eqnarray*}
\begin{aligned}
\alpha_T &:=\min_{N \in T' \oplus \mathrm{span}(\ones), \rho(T',T^\star) \leq \omega,\|N\|_2=1}\|\proj_{T'\oplus \mathrm{span}(\ones)}\I\proj_{T'\oplus \mathrm{span}(\ones)}(N)\|_2,\\
\delta_{T^\perp} &:= \max_{\substack{N \in T'\oplus\mathrm{span}(\ones), \rho(T',T^\star) \leq \omega,\\ \|N\|_2\leq{1}}}\|\proj_{{T'\oplus\mathrm{span}(\ones)}^\perp}\I\proj_{{T'\oplus\mathrm{span}(\ones)}}(N)\|_2,\\
\beta_T &:=\max_{N \in T' \oplus \mathrm{span}(\ones), \rho(T',T^\star) \leq \omega, \|N\|_\infty = 1}\|\I(N)\|_\infty.
\end{aligned}
\end{eqnarray*}
Here, $\alpha_T$ quantifies the minimum gain of $\I$ restricted to tangent spaces $T'\oplus \mathrm{span}(\ones)$ that are close to $T^\star$ with respect to the spectral norm; the quantify $\delta_T$ computes the inner-product between elements in $T'$ and ${T'}^\perp$ as quantified by the metric induced by $\I$; and finally, $\beta_T$ quantifies the behavior of $\I$ restricted to $T'\oplus\mathrm{span}(\ones)$ and ${T'}^\perp$ in infinity norm. 

With these above quantities defined, and letting $\tilde{\alpha} := \min\{\alpha_\Omega,\alpha_T\}$, $\tilde{\delta} := \max\{\delta_{\Omega^\perp},\delta_{T^\perp}\}$, and $\tilde{\beta} := \max\{\beta_\Omega,\beta_T\}$, the main assumptions are the following. Recall that $d^\star := \max_{i} \sum_{j=1}^p \mathbb{I}[|S^\star_{ij|>0}$ represents the maximal degree of the conditional graphical structure of the observed variables conditioned on the latent variables and $\inc^\star := \max_{i} \|\mathcal{P}_{\text{col-space}(L^\star)}e_i\|_2$ represents the denseness of the latent effects with $e_i$ denoting a standard coordinate basis element. 
\begin{assumption}$\tilde{\alpha} > 0$.
\label{ass:min_gain}
\end{assumption}

\begin{assumption} There exists $\tilde{\nu} \in (2\omega,1/2)$ such that $\tilde{\delta}/\tilde{\alpha} \leq 1-2\tilde{\nu}$.
\label{assumption:interpretable_main}
\end{assumption}
{
\begin{assumption}
\label{assumption:prod_inc_degree}
The product of degree of sparsity of $S^\star$, $d^\star$, and the diffuseness of the latent effects, $\inc^\star$, is bounded as follows: $d^\star(6\inc^\star+\omega) \leq \frac{\tilde{\nu}^2\tilde{\alpha}^2}{2\tilde{\beta}^2(2-\tilde{\nu})^2}$ where $\tilde{\nu}$ and $d^\star \leq \frac{\tilde{\alpha}^2\tilde{\nu}}{32\omega\tilde{\beta}(2-\tilde{\nu}(\|\I(F)\|_2 + \|\I\|_2\omega)}$.
\end{assumption}
}

\begin{assumption} The regularization parameter $\gamma$ chosen in the following range:
$$\gamma \in \left[\frac{2\tilde{\beta}{d^\star}(2-\tilde{\nu})}{\tilde{\nu}\tilde{\alpha}},\min\left\{\frac{\tilde{\nu}\tilde{\alpha}(1-\omega)}{\tilde{\beta}(6\mu^\star+\omega)(2-\tilde{\nu})},\frac{\tilde{\alpha}}{16\omega(\|\mathbb{I}^\star(F)\|_2+\|\mathbb{I}^\star\|_2\omega+1)}\right\}\right]$$
\label{assumption:choice_gamma}
\end{assumption}

\begin{assumption} $\kappa^\star:= \|\proj_{{T^\star}^\perp}(\ones/p)\|_2 \in \left(\omega,\min\left\{2\tilde{\nu}, \frac{\tilde{\alpha}}{16\max\{\gamma,1\}(\|\mathbb{I}^\star(F)\|_2+\|\mathbb{I}^\star\|_2\omega+1)}-\omega\right\}\right)$.
\label{assumption:kappa}
\end{assumption}

Assumptions~\eqref{ass:min_gain}-\ref{assumption:prod_inc_degree} are akin to conditions imposed in \cite{Chand2012}, although our conditions the subspace $T' \oplus \mathrm{span}(\ones)$ that arises from the additional zero row sum constraint in our estimator. Assumption~\ref{assumption:prod_inc_degree} ensures that $d^\star$ and $\inc^\star$ are not simultaneously large, and this type of condition was shown to be sufficient for recovering a sparse and low-rank matrix from their sum using mix of $\ell_1$ and nuclear norm regularization \citep{Chandrasekaran2011RankSparsityIF}. Assumption~\ref{assumption:kappa} is a new condition relative to \cite{Chand2012} to deal with the dual parameter $t\ones$ that arises from the zero sum constraint.

\begin{lemma}Under Assumptions~\ref{assumption:interpretable_main}-\ref{assumption:kappa}, we have that Hessian assumptions~\ref{ass:1}-\ref{ass:3} for some $\alpha = \tilde{\alpha}/2$, $\nu = 2\tilde{\nu}$.
\end{lemma} 

\begin{proof}
Our analysis will depend on the following quantities for any pair of subspaces $\Omega,T \subseteq \mathbb{R}^{p \times p}$:
\begin{eqnarray*}
\begin{aligned}
\theta(\Omega) :=\max_{N \in \Omega, \|N\|_\infty = 1}\|N\|_2~~~;~~~ \xi(T) :=\max_{N \in T, \|N\|_2 = 1}\|N\|_\infty.
\end{aligned}
\end{eqnarray*}
When $\Omega = \Omega^\star$ and $T = T^\star$, these quantities are closely connected to the maximal degree $d^\star$ and the incoherence parameter $\inc^\star$ (defined in Section~\ref{sec:fisher_conds}). In particular, \cite{Chand2012} showed that $\mu(\Omega^\star) \in [0,d^\star]$ and $\xi(T^\star) \in [\inc^\star,2\inc^\star]$. 

 We consider the quantity  $\min_{Z\in\mathbb{Q}',\Phi_\gamma(Z)=1} \allowbreak\|\proj_{\mathbb{H}'}\Gla\I\Gl\proj_{\mathbb{Q}'}(Z)\|_{\Phi_\gamma}$. Let $Z = (Z_1,Z_2)$ where $\|Z\|_{\Phi_\gamma} = 1$, $Z \in \mathbb{Q}'$. Suppose $\|Z_1\|_\infty = 1$. Then using Lemmas~\ref{lemma:1} and \ref{lemma:3p}, we have that:
\begin{eqnarray*}
\begin{aligned}
\|\proj_{\Omega^\star}\Gla\I\Gl(Z)\|_{\infty
} &\geq \|\proj_{\Omega^\star}\I\proj_{\Omega^\star}(Z_1)\|_\infty - \|\proj_{\Omega^\star}\I(Z_2)\|_\infty \\
&\geq \alpha' - \|\I(Z_2)\|_\infty\\ &\geq \tilde{\alpha}-\gamma\tilde{\beta}\xi(T'\oplus \mathrm{span}(\ones))\geq\tilde{\alpha}-3\gamma\tilde{\beta}\xi(T')\geq \tilde{\alpha}-\frac{(3\xi(T^\star)+\omega)}{1-\omega}\tilde{\beta}\gamma\\& \geq \tilde{\alpha}-\frac{\tilde{\nu}\tilde{\alpha}}{2-\tilde{\nu}}.
\end{aligned}
\end{eqnarray*}
Now suppose that $\|Z_2\|_2 = \gamma$. Then, we have using Lemma~\ref{lemma:4} that:
\begin{eqnarray*}
\begin{aligned}
\|\proj_{T'\oplus\mathrm{span}(\ones)}\I\Gl(Z)\|_2 &\geq \|\proj_{T'\oplus\mathrm{span}(\ones)}\I\proj_{T'\oplus\mathrm{span}(\ones)}(Z_2)\|_2 \\&-\|\proj_{T' \oplus \mathrm{span}(\ones)}\I(Z_1)\|_2 \\
&\geq \tilde{\alpha}\gamma-2\tilde{\beta}\theta(\Omega^\star) \geq \tilde{\alpha}\gamma - \frac{\tilde{\nu}\tilde{\alpha}\gamma}{2-\tilde{\nu}}.
\end{aligned}
\end{eqnarray*}

Putting the previous bounds together, we have that:
\begin{eqnarray}
\min_{Z\in\mathbb{Q}',\Phi_\gamma(Z)=1} \|\proj_{{\mathbb{Q}'}}\Gla\I\Gl\proj_{\mathbb{Q}'}(Z)\|_{\Phi_\gamma} \geq \tilde{\alpha}-\frac{\tilde{\nu}\tilde{\alpha}}{2-\tilde{\nu}} \geq \tilde{\alpha}/2.
\label{eqn:bound_Q_Q}
\end{eqnarray}

Now we consider the quantity $\max_{Z\in\mathbb{Q}',\Phi_\gamma(Z)=1} \allowbreak\|\proj_{{\mathbb{Q}'}^\perp}\Gla\I\Gl\proj_{\mathbb{Q}'}(Z)\|_{\Phi_\gamma}$. Let $Z = (Z_1,Z_2)$ where $\|Z\|_{\Phi_\gamma} = 1$, $Z \in \mathbb{Q}'$. Suppose $\|Z_1\|_\infty = 1$.
\begin{eqnarray*}
\begin{aligned}
\|\proj_{{\Omega^\star}^\perp}\Gla\I\Gl(Z)\|_{\infty
} &\leq \|\proj_{{\Omega^\star}^\perp}\I\proj_{\Omega^\star}(Z_1)\|_\infty +\|\proj_{{\Omega^\star}^\perp}\I(Z_2)\|_\infty \\
&\leq \tilde{\delta} + \|\I(Z_2)\|_\infty \leq \tilde{\delta}+\frac{\tilde{\beta}\gamma(3\xi(T^\star)+\omega)}{1-\omega} \leq \tilde{\delta} + \frac{\tilde{\nu}\tilde{\alpha}}{(2-\tilde{\nu})} .
\end{aligned}
\end{eqnarray*}
Now suppose that $\|Z_2\|_2 = \gamma$. Then, we have using Lemma~\ref{lemma:4}:
\begin{eqnarray*}
\begin{aligned}
\|\proj_{{T' \oplus \mathrm{span}(\ones)}^\perp}\I\Gl(Z)\|_2 &\leq \|\proj_{{T' \oplus\mathrm{span}(\ones)}^\perp}\I\proj_{T'\oplus\mathrm{span}(\ones)}(Z_2)\|_2 +\|\proj_{{T'\oplus\mathrm{span}(\ones)}^\perp}\I(Z_1)\|_2 \\
&\leq \tilde{\delta}\gamma+\tilde{\beta}\theta(\Omega^\star) \leq \tilde{\delta}\gamma + \frac{\tilde{\nu}\tilde{\alpha}\gamma}{(2-\tilde{\nu})} .
\end{aligned}
\end{eqnarray*}
Combining the last two inequalities, we have: 
\begin{eqnarray}
\max_{Z\in\mathbb{Q}',\Phi_\gamma(Z)=1} \allowbreak\|\proj_{{\mathbb{Q}'}^\perp}\Gla\I\Gl\proj_{\mathbb{Q}'}(Z)\|_{\Phi_\gamma} \leq \tilde{\delta}+\frac{\tilde{\nu}\tilde{\alpha}}{(2-\tilde{\nu})}.
\label{eqn:bound_Qperp_Q}
\end{eqnarray}

Combining \eqref{eqn:bound_Q_Q} and \eqref{eqn:bound_Qperp_Q}, we have that:
\begin{eqnarray*}
\max_{\substack{Z \in \mathbb{Q}'\\\|Z\|_{\Psi}=1}}\|\mathcal{P}_{{\mathbb{Q}'}^\perp}\mathcal{A}^\dagger\mathbb{I}^\star\mathcal{A}\mathcal{P}_{\mathbb{Q}'}(\mathcal{P}_{\mathbb{Q}'}\mathcal{A}^\dagger\mathbb{I}^\star\mathcal{A}\mathcal{P}_{\mathbb{Q}'})^{-1}(Z)\|_{\Psi} \leq \frac{\tilde{\delta}+\frac{\tilde{\nu}\tilde{\alpha}}{(2-\tilde{\nu})}}{\tilde{\alpha}-\frac{\tilde{\nu}\tilde{\alpha}}{2-\tilde{\nu}}} \leq 1-\tilde{\nu}.
\end{eqnarray*}

\end{proof}

\section{Finite sample convergence guarantees of the empirical variogram matrix}
\label{sec:finite_sample_variogram}
In addition to the identifiability assumptions, following \cite{engelke2022b}, we impose conditions to characterize the convergence rate of the empirical variogram matrix to the population variogram matrix. Throughout, we suppose that the random vector $X = (X_O, X_H)$ is in the domain of attraction of the multivariate Pareto distribution $Y$ following a latent \HR{} distribution with parameter matrix $\Gamma$; for details see Section~\ref{sec:mevt} and~\ref{sec:latent_hr_graphs}.
\begin{assumption}\label{assumption2} The marginal distribution functions $F_i$ of $X_i$, $i\in O$,  are continuous and there exists constants $\xi > 0$,  $K < \infty$ such that for all triples of distinct indices $J = (i,j,m) \subset O$ and $q \in (0,1]$,
        \[\sup_{{x} \in [0,q^{-1}]^2 \times [0,1]} \left |q^{-1}\mathbb{P}(F_J({X}_J)>1-q{x})-\frac{\mathbb{P}({Y}_J > 1/{x})}{\mathbb{P}(Y_1>1)}\right| \leq Kq^\xi,\]
        where $F_J({x}) = (F_i({x_i}),F_j({x}_j),F_m(x_m))$.
\end{assumption}
Assumption~\ref{assumption2} is a second-order condition that essentially controls the speed of
convergence of the sample variogram matrix to the population variogram matrix. 

\begin{corollary}\cite[][Theorem 1]{engelke2022b} Let Assumption~\ref{assumption2} hold. Let $\ell\in (0,1]$ be arbitrary. Suppose that $n^{\ell} \leq k \leq n/2$ where $k$ is the effective sample size in computing the sample variogram matrix (see Section~\ref{sec:exmpirical_variogram}). Let $\vartheta \geq 0$ be any scalar satisfying $\vartheta \leq \sqrt{k}/(\log{n})^4$. Then, there exists positive constants $c_5,C_5,\tilde{C}_5$ only depending on $K,\xi,\ell,\epsilon,$ and $G(z)$ such that:
$$\mathbb{P}\left(\|\hat{\Gamma}_O - \Gamma^{\star}_O\|_\infty > C_5\left\{\left(\frac{k}{n}\right)^{\xi}(\log(n/k))^2 + \frac{1+\vartheta}{\sqrt{k}}\right\}\right) \leq \tilde{C}_5p^3e^{-c_5\vartheta^2}.$$
Further, if the random vector $X$ is in the domain of attraction of a max-stable distribution, then $\xi = 1$.
\label{cor:sample_variogram}
\end{corollary}

\section{Proof of Theorem~\ref{thm:main}}
\label{sec:proof_thm2}

\subsection{Implied Hessian conditions}
Combining Lemma~\ref{lemma:Hessian_cond} with Assumptions~\ref{ass:1}-\ref{ass:3}, and letting $m = \max\{\gamma,1\}$, we have that the following three properties:
\begin{eqnarray*}
\begin{aligned}
&\min_{\substack{Z \in \mathbb{H}'\\\|Z\|_{\Phi_\gamma} = 1\\\rho(T',T^\star)\leq \omega}}\|\mathcal{P}_{\mathbb{H}'}\mathcal{A}^\dagger\mathbb{I}^\star\mathcal{A}\mathcal{P}_{\mathbb{H}'}(Z)\|_{\Phi_\gamma} \geq \alpha-2(\kappa^\star+\omega) \in (0,\infty),\\
&\max_{\substack{Z \in \mathbb{H}'\\\|Z\|_{\Phi_\gamma} = 1\\\mathbb{Q}'\in U(\omega)}}\|\mathcal{P}_{{\mathbb{H}'}^\perp}\mathcal{A}^\dagger\mathbb{I}^\star\mathcal{A}\mathcal{P}_{\mathbb{Q}'}(\mathcal{P}_{{\mathbb{H}'}}\mathcal{A}^\dagger\mathbb{I}^\star\mathcal{A}\mathcal{P}_{\mathbb{Q}'})^{-1}(Z)\|_{\Phi_\gamma} \leq 1-(\nu-2(\kappa^\star+\omega)) \in [0,1) ,\\
&\max_{\substack{Z \in \mathbb{Q}'\\\|Z\|_{\Phi_\gamma} = 1\\\mathbb{Q}'\in U(\omega)}}\|(\mathcal{P}_{{\mathbb{H}'}}\mathcal{A}^\dagger\mathbb{I}^\star\mathcal{A}\mathcal{P}_{\mathbb{H}'})^{-1}\mathcal{P}_{{\mathbb{H}'}}\mathcal{A}^\dagger\mathbb{I}^\star\mathcal{A}\mathcal{P}_{{\mathbb{H}'}^\perp}(Z)\|_{\Phi_\gamma} \leq 1-\frac{4(\kappa^\star+\omega)m(\|\mathbb{I}^\star(F)\|_2+\|\mathbb{I}^\star\|_2\omega)}{\alpha-2({\kappa}^\star+\omega)} \in [0,1).
\end{aligned}
\end{eqnarray*}
The first property follows from $\kappa^\star < \frac{\alpha}{4}$ and $\alpha > 4\omega$. The second property follow from $1-\nu + 2(\kappa^\star+\omega) <1$ since $\nu > 2(\kappa^\star+\omega)$. The final property follows from having $\frac{4(\kappa^\star+\omega)\max\{\gamma,1\}(\|\mathbb{I}^\star(F)\|_2+\|\mathbb{I}^\star\|_2\omega+1)}{\alpha} < 1$ or equivalently that $\kappa^\star \leq \frac{\alpha}{8\max\{\gamma,1\}(\|\mathbb{I}^\star(F)\|_2+\|\mathbb{I}^\star\|_2\omega+1)}-\omega$ with $\alpha > 8\omega\max\{\gamma,1\}(\|\mathbb{I}^\star(F)\|_2+\|\mathbb{I}^\star\|_2\omega+1)$. For notational simplicity and with slight abuse of notation, we let:
\begin{eqnarray*}
\begin{aligned}
\alpha' &:=  \alpha - 2(\kappa^\star+\omega),\\
\zeta &:= \max\left\{\frac{1}{\nu-2(\kappa^\star+\omega)},\frac{\alpha-2(\kappa^\star+\omega)}{4(\kappa^\star+\omega)m(\|\mathbb{I}^\star(F)\|_2+\|\mathbb{I}^\star\|_2\omega)}\right\}.
\end{aligned}
\end{eqnarray*}
Then, we have the following Hessian conditions:
\begin{eqnarray}
\begin{aligned}
&p1) \min_{\substack{Z \in \mathbb{H}'\\\|Z\|_{\Phi_\gamma} = 1\\\rho(T',T^\star)\leq \omega}}\|\mathcal{P}_{\mathbb{H}'}\mathcal{A}^\dagger\mathbb{I}^\star\mathcal{A}\mathcal{P}_{\mathbb{H}'}(Z)\|_{\Phi_\gamma} \geq\alpha' \in (0,\infty),\\
&p2) \max_{\substack{Z \in \mathbb{H}'\\\|Z\|_{\Phi_\gamma} = 1\\\mathbb{Q}'\in U(\omega)}}\|\mathcal{P}_{{\mathbb{H}'}^\perp}\mathcal{A}^\dagger\mathbb{I}^\star\mathcal{A}\mathcal{P}_{\mathbb{Q}'}(\mathcal{P}_{{\mathbb{H}'}}\mathcal{A}^\dagger\mathbb{I}^\star\mathcal{A}\mathcal{P}_{\mathbb{Q}'})^{-1}(Z)\|_{\Phi_\gamma} \leq 1-\frac{1}{\zeta} \in [0,1) ,\\
&p3) \max_{\substack{Z \in \mathbb{Q}'\\\|Z\|_{\Phi_\gamma} = 1\\\mathbb{Q}'\in U(\omega)}}\|(\mathcal{P}_{{\mathbb{H}'}}\mathcal{A}^\dagger\mathbb{I}^\star\mathcal{A}\mathcal{P}_{\mathbb{H}'})^{-1}\mathcal{P}_{{\mathbb{H}'}}\mathcal{A}^\dagger\mathbb{I}^\star\mathcal{A}\mathcal{P}_{{\mathbb{H}'}^\perp}(Z)\|_{\Phi_\gamma} \leq 1-\frac{1}{\zeta} \in [0,1),
\end{aligned}
\label{eqn:Hessian_prop}
\end{eqnarray}

\subsection{Full theoretical statement}
Let $c_5,C_5,\tilde{C}_5$ be constants that ensure Corollary~\ref{cor:sample_variogram} is satisfied. Let $\psi= \max\{1,\|(S^\star-L^\star)^{+}\|_2\}$, $C_0 = 8+\frac{32\sqrt{5h}}{\alpha'(1-\sqrt{1- ({\kappa^\star}^2-\omega)^2})(\frac{1}{\zeta}-2(\kappa^\star+\omega))}\allowbreak\left[1+\frac{1}{3\zeta}\right]$, $C_1 = \psi(m+d^\star)$, and $C_2 ={m}\max\{\left(\frac{4C_0}{\alpha'}+\frac{1}{\psi}\right),1\}$. We also define, 
\begin{equation*}
\begin{aligned}
C_4 &=  \min\Bigg\{\min\left\{\frac{8\alpha'}{C_1}, \frac{\min\{\alpha',1\}(\frac{1}{\zeta} - 2(\kappa^\star+\omega))}{16m\psi{C}_2^2}\right\}\frac{\alpha'(\frac{1}{\zeta}-2(\kappa^\star+\omega))}{4(1+\frac{1}{3\zeta})}\\&,\frac{\alpha'(\frac{1}{\zeta}-2(\kappa^\star+\omega))}{64C_1(1+\frac{1}{3\zeta})}, \frac{\alpha'^2(\frac{1}{\zeta}-2(\kappa^\star+\omega))^2}{6144\zeta(1+\frac{1}{3\zeta})^2}\Bigg\}.
\end{aligned}
\end{equation*}

\begin{theorem} Suppose that there exists $\alpha >0$, $\nu \in (0,1]$, $\omega\in (0,1)$ and the choice of the parameter $\gamma$ so that the Hessian $\mathbb{I}^\star$ satisfies Assumptions~\ref{ass:1}-\ref{ass:3}. Let $m := \max\{1,{1}/{\gamma}\}$ and $\bar{m} := \max\{1,\gamma\}$. Let the effective sample size $k$ be chosen such that $k = o(\lfloor n^{2\xi/(1+2\xi)} \rfloor)$. Furthermore, suppose that:
\begin{equation*}
\begin{aligned}
k \geq \max\Bigg\{\frac{C_5^21152m^2\zeta^2p^2\log(\tilde{C}_5p)}{C_4^2c_5^2} + \frac{72m^2\zeta^2}{C_4^2},\left(\frac{2C_5}{0.1^2\sqrt{c_5}}\sqrt{\log(\tilde{C}_5p)}\right)^{-2/(3/2-(2\xi+1)/(2\xi))},\\\log(k)^{2/(3/2-(2\xi+1)/(2\xi))}, 4(3/2-(2\xi+1)/(2\xi))^8\frac{\log(\tilde{C}_5p)}{c_5}\log(k)^8\Bigg\}
\end{aligned}
\end{equation*}
and
\begin{enumerate}
    \item $\lambda_n =C_5\left[\frac{24m\zeta}{\sqrt{c_5}}\sqrt{\frac{p^2\log(\tilde{C}_5p)}{k}}+\frac{6m\zeta}{\sqrt{k}}\right]$,
    \item $\sigma_\mathrm{min}(L^\star) \geq \max\left\{16m\bar{m}\frac{\lambda_n{C}_2}{\omega},\frac{2\psi{C}_2^2\lambda_n}{C_0},\left(mC_2 + \frac{\alpha'(\frac{1}{\zeta}-2(\kappa^\star+\omega))}{4\left[1+\frac{1}{3\zeta}\right]}\right)\lambda_n\right\}$,
    \item $|S^\star_{ij}| \geq 12m\bar{m}{\lambda_n{C}_2}$ whenever $|S^\star_{ij}|>0$.
\end{enumerate}
Then, the estimate $(\hat{S},\hat{L})$ is the unique minimizer of \eqref{eqn:estimator} with
$$ \mathbb P\left(\mathrm{sign}(\hat{S}) = \mathrm{sign}(S^\star), \mathrm{rank}(\hat{L}) = \mathrm{rank}(L^\star), \|(\hat{S}-\hat{L})-\tilde{\Theta}^\star\|_2 \leq 2mC_2\lambda_n \right) \geq 1-\frac{1}{p}.$$
\label{thm:main_full}
\end{theorem}

To arrive at the scalings provided in Theorem~\ref{thm:main}, note that, $\zeta = \mathcal{O}(1/\nu)$, $\zeta = \mathcal{O}(1/\nu)$, $C_0  = \mathcal{O}(\sqrt{h}\nu/{\alpha'})$, $C_1 = \mathcal{O}(md^\star)$, $C_2 = \mathcal{O}(m\sqrt{h}/{\alpha'}^2)$, $C_4 = \mathcal{O}(\alpha'^3\nu/(d^\star{m}^3\sqrt{h})$. This scaling allows us to conclude that: $k \gtrsim \frac{m^3h{d^\star}^2}{\alpha'^6\nu^2}m\nu{p}\log(p)$, $\lambda_n  = \frac{m}{\nu}\sqrt{\frac{p^2\log(p)}{k}}$, $\sigma_\mathrm{min}(L^\star) \gtrsim \frac{m^4h\bar{m}}{\nu\alpha'^4}\sqrt{\frac{p^2\log(p)}{k}}$, $S^\star_{ij} \gtrsim \frac{m^3\bar{m}\sqrt{h}}{\nu\alpha'^2}\sqrt{\frac{p^2\log(p)}{k}}$, and finally $\|(\hat{S}-\hat{L})-\tilde{\Theta}^\star\|_2 \lesssim \frac{m^3\sqrt{h}}{\nu\alpha'^2}\sqrt{\frac{p^2\log(p)}{k}}$.

\subsection{Proof strategy}
\label{sec:proof_strategy}
The high-level proof strategy is similar in spirit to the proofs of consistency results for sparse graphical model recovery and latent variable graphical model recovery \citep{Chand2012}, although our convex program and the conditions required for its success are different from these previous results. Consider the following convex program
\begin{equation}
\begin{aligned}
    (\hat{S},\hat{L}) = \arg\min_{S,L\in\mathbb{S}^p}&~~-\log{\det}(U^T(S-L)U) - \mathrm{tr}((S-L)\hat{\Gamma}_O/2) + \lambda_n(\|S\|_{1} + \gamma\|L\|_\star).\\
    \text{subject-to}&~~~S-L \in \mathrm{span}(\ones)
\end{aligned}
\label{eqn:estimator_no_psd}
\end{equation}
Comparing \eqref{eqn:estimator_no_psd} with the convex program \eqref{eqn:estimator}, the differences are: i) we have removed the positive-definite constraints, ii) we have replaced $\mathrm{tr}(L)$ with $\|L\|_\star$ which is valid for positive semi-definite $L$, iii) we have replaced the constraint $(S-L)\textbf{1}_p = 0$ with $S-L \in \mathrm{span}(\ones)$ which is equivalent since the matrices $S,L$ are symmetric. Regarding item i), the positive definiteness of $\hat{S}-\hat{L}$ is automatically met due to the log-det term. We show with high probability that $\hat{L} \succeq 0$. 

Note that due to the log-det term, we have that $UU^T(S-L)UU^T = S-L$. Appealing to Lemma \ref{lemma:new_2}, we conclude that $U(U^T(S-L)U)^{-1}U^T$, which is the gradient of the negative log-determinate term with respect to $S$ is equivalent to $(S-L)^+$. Similarly, since $\mathrm{tr}((S-L)\hat{\Gamma}_0/2) =\mathrm{tr}(UU^T(S-L)UU^T\hat{\Gamma}_0/2) = \mathrm{tr}((S-L)UU^T\hat{\Gamma}_0/2UU^T)$, the gradient of the trace term in the objective with respect to $S$ is given by $UU^T\hat{\Gamma}_0/2UU^T$. Standard convex analysis states that $(\hat{S},\hat{L})$ is the solution of the convex program \eqref{eqn:estimator_no_psd} if there exists a dual variable $t \in \mathbb{R}$ with the following conditions being satisfied:
\begin{equation}
\begin{aligned}
-UU^T(\hat{\Gamma}_O/2){U}U^T - (\hat{S}-\hat{L})^{+} + t\ones &= -\lambda\partial\|\hat{S}\|_{1},\\
UU^T(\hat{\Gamma}_O/2){U}U^T + (\hat{S}-\hat{L})^{+}-t\ones &= -\lambda\gamma\partial\|\hat{L}\|_{\star},\\
\hat{S}-\hat{L} &\in \mathrm{span}(\ones).
\label{eqn:optimality_temp}
\end{aligned}
\end{equation}

Recall that elements of the subdifferential with respect to nuclear norm at a matrix $M$
have the key property that they decompose with respect to the tangent space $T(M)$.
Specifically, the subdifferential with respect to the nuclear norm at a matrix $M$ with
(reduced) SVD given by $M = U_lQU_r^T$ is as follows:
$$N \in \partial\|M\|_\star \Leftrightarrow \proj_{T(M)}(N) = U_lV_r^T, \|\proj_{T(M)^\perp}(N)\|_2 \leq 1,$$
where $\proj$ denotes a projection operator. Similarly, we have the following for the subdifferential of $\ell_1$ norm:
$$N \in \partial\|M\|_{1} \Leftrightarrow \proj_{\Omega(M)}(N) = \text{sign}(N), \|\proj_{\Omega(M)^\perp}(N)\|_\infty \leq 1.$$
Let SVD of $\hat{L}$ be $\hat{U}\hat{D}\hat{V}^T$ and let $Z = (-\lambda\text{sign}(\hat{S}), -\lambda\gamma\hat{U}\hat{V}^T)$. Then, letting $\Hs = \Omega(\hat{S})\times T(\hat{L})$ the optimality conditions of \eqref{eqn:estimator_no_psd} reduce to:
\begin{equation}
\begin{aligned}
\proj_{\Hs}\Gla(-UU^T\hat{\Gamma}_O/2UU^T -(\hat{S}-\hat{L})^+-t\ones) &= Z,\\
\Phi_\gamma(\proj_{\Hs^\perp}\Gla(-UU^T\hat{\Gamma}_O/2UU^T - (\hat{S}-\hat{L})^+-t\ones)) &\leq \lambda_n,\\
\hat{S}-\hat{L} &\in \mathrm{span}(\ones).
\end{aligned}
\end{equation}
To ensure that the estimates $(\hat{S},\hat{L})$ are close to their respective population parameters, the quantity $\Delta_S = \hat{S}-S^\star$ and $\Delta_L=\hat{L}-L^\star$ must be small. Since the optimality conditions of \eqref{eqn:estimator_no_psd} are stated in terms of $(\hat{S}-\hat{L})^+$, we bound the deviation between $(\hat{S}-\hat{L})^+$ and $({S}^\star-{L}^\star)^+$. Specifically, the Taylor Series expansion of $(\hat{S}-\hat{L})^+$ around $({S}^\star-{L}^\star)^+$ is:
$$(\hat{S}-\hat{L})^+ = (S^\star-L^\star + \Gl(\Delta_S,\Delta_L))^+ = (S^\star-L^\star)^{+} + (S^\star-L^\star)^{+}\Gl(\Delta_S,\Delta_L)(S^\star-L^\star)^++\mathcal{R}_{\Gamma^\star_0}\Gl(\Delta_S,\Delta_L).$$
where some algebra yields the following representation for the remainder term $\mathcal{R}_{\Gamma^\star_0}(\Gl(\Delta_S,\Delta_L))$:
\begin{equation}
\mathcal{R}_{\Gamma^\star_0}(\Gl(\Delta_S,\Delta_L))=U(S^\star-L^\star+\ones/p)^{-1}\left[\sum_{k=2}^\infty (-\Gl(\Delta_S,\Delta_L)(S^\star-L^\star+\ones/p)^{-1})^k\right]U^T.
\label{remainder:eq}
\end{equation}
From Theorem~\ref{thm:main1}, we have that $(S-L)^{+} = UU^T(-\Gamma^\star/2)UU^T$. Since $UU^T(S^\star-L^\star)UU^T = S^\star-L^\star$, we appeal to Lemma~\ref{lemma:new_2} to conclude that $(U^T(S^\star-L^\star)U)^{-1} = U^T(-\Gamma^\star_O)U$. Let $E_n := UU^T(\hat{\Gamma}_O-\Gamma^\star)/2UU^T$. Then, we have the following equivalent characterization of the optimality conditions \eqref{eqn:optimality_temp}:
\begin{eqnarray}
\begin{aligned}
&\proj_{\Hs}\Gla((S^\star-L^\star)^{+}\Gl(\Delta_S,\Delta_L)(S^\star-L^\star)^++\mathcal{R}_{\Gamma^\star_0}\Gl(\Delta_S,\Delta_L)+E_n+t\ones)= Z,\\
&\Phi_\gamma(\proj_{\Hs^\perp}\Gla((S^\star-L^\star)^{+}\Gl(\Delta_S,\Delta_L)(S^\star-L^\star)^++\mathcal{R}_{\Gamma^\star_0}\Gl(\Delta_S,\Delta_L)+E_n+t\ones)) \leq \lambda_n,\\
\hat{S}-\hat{L} &\in \mathrm{span}(\ones).
\end{aligned}
\label{eqn:optimality_temp2}
\end{eqnarray}
Finally, Since $(S^\star-L^\star)\ones = 0$ and $\Gl(\Delta_S,\Delta_L)\ones = 0$, we have the following formulation of the optimality condition \eqref{eqn:optimality_temp2} in terms of the matrix $\I$
\begin{eqnarray}
\begin{aligned}
&\proj_{\Hs}\Gla(\I(\Gl(\Delta_S,\Delta_L+t\ones))+\mathcal{R}_{\Gamma^\star_0}\Gl(\Delta_S,\Delta_L+t\ones)+E_n) = Z,\\
&\Phi_\gamma(\proj_{\Hs^\perp}\Gla(\I(\Gl(\Delta_S,\Delta_L+t\ones))+\mathcal{R}_{\Gamma^\star_0}\Gl(\Delta_S,\Delta_L+t\ones)+E_n)) \leq \lambda_n,\\
\hat{S}-\hat{L} &\in \mathrm{span}(\ones).
\end{aligned}
\label{eqn:optimality_orig}
\end{eqnarray}
It is straightforward to show that if for some $(\hat{S},\hat{L})$, the second condition in \eqref{eqn:optimality_orig} is satisfied with strict inequality, that is: 
$$\Phi_\gamma(\proj_{\Hs^\perp}\Gla(\I(\Gl(\Delta_S,\Delta_L+t\ones))+\mathcal{R}_{\Gamma^\star_0}\Gl(\Delta_S,\Delta_L+t\ones)+E_n)) < \lambda_n.$$

%Throughout, we let $\psi := \|\I\|_2$.

\subsection{Constrained optimization problem}
We consider the following non-convex optimization problem:
\begin{equation}
\begin{aligned}
    \argmin_{S \in \mathbb{S}^{p},L \in \mathbb{S}^{p}} &~~-\log{\det}(U^T(S-L)U) - \mathrm{tr}((S-L)\hat{\Gamma}_O/2) + \lambda_n(\|S\|_{1} + \gamma\|L\|_\star),\\
    \text{subject-to}&~~~S-L \in \mathrm{span}(\ones)~;~ (S,L) \in \mathcal{M},
\end{aligned}
\label{eqn:estimator_nonconvex}
\end{equation}
where:
\begin{eqnarray*}
\begin{aligned}
    \mathcal{M} &= \Bigg\{S,L \in \mathbb{S}^p: S \in \Omega^\star, \text{rank}(L) \leq \text{rank}(L^\star)\\ &~~~~~~~~~~~~\|\proj_{{T^\star}^\perp}(L-L^\star)\|_2 \leq \frac{C_0\lambda_n}{\psi}~~;~~ \Phi_\gamma(\Gla\I\Gl(S-S^\star,L-L^\star)) \leq C_0\lambda_n \Bigg\},
\end{aligned}
\end{eqnarray*}
with $C_0 = 10+\frac{32\sqrt{5h}}{\alpha'(1-\sqrt{1- ({\kappa^\star}^2-\omega)^2})(\frac{1}{\zeta}-2(\kappa^\star+\omega))}\left[1+\frac{1}{3\zeta}\right]$. The optimization program \eqref{eqn:estimator_nonconvex} is non-convex due to the rank constraint $\text{rank}(L) \leq \text{rank}(L^\star)$ in the set $\mathcal{M}$. These constraints ensure that the matrix $L$ belongs to an appropriate variety. The constraints in $\mathcal{M}$ along ${T^\star}^\perp$ ensure that the tangent space $T(L)$ is close to $T^\star$. Finally, the last condition roughly controls the error. We begin by proving the following useful proposition:
\begin{proposition}
Let $(S,L)$ be a set of feasible variables of \eqref{eqn:estimator_nonconvex}. Let $\Delta = (S-S^\star,L-L^\star)$. Then, $\Phi_\gamma(\Delta) \leq {C}_2\lambda_n$ where $C_2 ={m}\max\{\left(\frac{4C_0}{\alpha'}+\frac{1}{\psi}\right),1\}$.
\label{prop:nonconvex}
\end{proposition}
\begin{proof}[Proof of Proposition~\ref{prop:nonconvex}]
Let $\mathbb{H}^\star = \Omega^\star \times T^\star$. Then:
\begin{eqnarray*}
\begin{aligned}
\Phi_\gamma[\Gla\I\Gl\proj_{\Hs^\star}(\Delta)] &\leq \Phi_\gamma[\Gla\I\Gl(\Delta)]+\Phi_\gamma[\Gla\I\Gl\proj_{{\Hs^\star}^\perp}(\Delta)] \\
&\leq C_0\lambda_n + mC_0\lambda_n \leq 2mC_0\lambda_n.
\end{aligned}
\end{eqnarray*}
Since $\Phi_\gamma[\proj_{\Hs^\star}(\cdot)] \leq 2\Phi_\gamma[\cdot]$, we have that: $\Phi_\gamma[\proj_{\Hs^\star}\Gla\I\Gl\proj_{\Hs^\star}(\Delta)] \leq 4mC_0\lambda_n$. Then, appealing to Property $p1$ in \eqref{eqn:Hessian_prop}, we have that: $\Phi_\gamma[\proj_{\Hs^\star}(\Delta)]\leq \frac{4C_0\lambda_n}{\alpha'}$. Moreover, $\Phi_\gamma(\Delta) \leq \Phi_\gamma[\proj_{\Hs^\star}(\Delta)] +\Phi_\gamma[\proj_{{\Hs^\star}^\perp}(\Delta)] \leq  \lambda_n{m}\left(\frac{4C_0}{\alpha'}+\frac{1}{\psi}\right)$.
\end{proof}
Proposition~\ref{prop:nonconvex} leads to powerful implications. In particular, under additional conditions on the minimum nonzero singular values of $L^\star$, any feasible set of variables $(S,L)$ of \eqref{eqn:estimator_nonconvex} has two key properties: (a) The variables $(S,L)$ are smooth points of their underlying varieties with $L \succeq 0$ and $S-L \succeq 0$, and (b) The constraints in $\mathcal{M}$ along ${T^\star}^\perp$ are locally inactive at $L$. These properties, among others, are proved in the following corollary. 
\begin{corollary}Consider any feasible variables $(S,L)$ of \eqref{eqn:estimator_nonconvex}. Let $T' = T(L)$. Let $\sigma$ be the smallest nonzero singular value of $L^\star$ and $s$ be the smallest in magnitude nonzero value of $S^\star$. Let $\Hs' = \Omega^\star \times T'$, $C_{T'} = \proj_{{T'}^\perp}(L^\star)$ and $C_{T'\oplus\mathrm{span}(\ones)} = \proj_{({T' \oplus \mathrm{span}(\ones)})^\perp}(L^\star)$. Suppose that the following inequalities are met: $\sigma \geq \max\left\{16m\bar{m}\frac{\lambda_n{C}_2}{\omega},\frac{2\psi{C}_2^2\lambda_n}{C_0},\left(mC_2 + \frac{\alpha'(\frac{1}{\zeta}-2(\kappa^\star+\omega))}{4\left[1+\frac{1}{3\zeta}\right]}\right)\lambda_n\right\}$ and $s \geq 12m\bar{m}\lambda_n{C}_2$. Then, 
\begin{enumerate}
\item $L$ and $S$ are smooth points of their underlying varieties so that $\mathrm{support}(\hat{S}) = \mathrm{support}(S^\star)$ and $\mathrm{rank}(\hat{L})= \mathrm{rank}(L^\star)$. Furthermore, $L \succeq 0$, and $S-L \succeq 0$
\item $\|\proj_{{T^\star}^\perp}(\hat{L}-L^\star)\|_2 \leq \frac{C_0\lambda_{n}}{2\psi}$,
\item $\rho(T',T^\star) \leq \omega$,
\item $\max\{\Phi_\gamma(\Gla\I{C}_{T'}),\Phi_\gamma(\Gla\I{C}_{T'\oplus\ones})\} \leq \frac{\lambda_n}{6\zeta}$,
\item $\Phi_\gamma[\Gla{C}_{T'}]\leq \frac{4\lambda_n}{\alpha'(\frac{1}{\zeta}-2(\kappa^\star+\omega))}\left[1+\frac{1}{3\zeta}\right]$.
\end{enumerate}
\label{corr:nonconvex}
\end{corollary}

\begin{proof}[Proof of Corollary~\ref{corr:nonconvex}] We appeal to the results regarding the perturbation analysis of the low-rank matrix variety. %Throughout, we will let $\sigma = \sigma_\mathrm{min}(L^\star)$.
\begin{enumerate}
\item Based on assumptions regarding the minimum nonzero singular value of $L^\star$ and minimum nonzero entry in magnitude of $S^\star$, one can check that since $\omega\leq 1$
\begin{eqnarray*}
\begin{aligned}
\sigma &\geq 12m\bar{m}\frac{\lambda_n{C}_2}{\omega} \geq 12m\bar{m}{\lambda_n{C}_2} \geq 8\|L-L^\star\|_2,\\
s&\geq 12m\bar{m}{\lambda_n{C}_2} \geq 12m\bar{m}{\lambda_n{C}_2} \geq 2\|S-S^\star\|_2.
\end{aligned}
\end{eqnarray*}
Combining these results, we conclude that $S,L$ are smooth points of their varieties, namely that $\text{rank}(L) = \text{rank}(L^\star)$ and $\text{support}(S) = \text{support}(S^\star)$. The fact that $L \succeq 0$ follows from $\sigma \geq 2\|L-L^\star\|_2$. Furthermore, to check that $S-L \succeq 0$, first note that $\sigma_\text{min}(S^\star-L^\star) \geq \frac{1}{\sqrt{\psi}}$. Then, $\|{S}-L-(S^\star-L^\star)\|_2 \leq 2mC_2\lambda_n$. From the choice of $\lambda_n$ and the condition on the sample size, we have that $4mC_2\lambda_n < \frac{1}{\sqrt{\psi}}$. Thus, $S-L \succeq 0$. 
\item Since $\sigma \geq  8\|L-L^\star\|_2$, we can appeal to Proposition 2.2 of \cite{Chand2012} to conclude that the constrain5 in $\mathcal{M}$ along $\proj_{{T^\star}^\perp}$ is strictly feasible:
$$\|\proj_{{T^\star}^\perp}(L-L^\star)\|_2 \leq \frac{\|L-L^\star\|_2^2}{\sigma} \leq \frac{C_2^2\lambda_n^2}{\sigma} < \frac{C_0\lambda_n}{\psi}.$$
\item Appealing to Proposition 2.1 of \cite{Chand2012}, we prove that the tangent space $T'$ is close to $T^\star$:
$$\rho(T',T^\star) \leq \frac{2\|L-L^\star\|_2}{\sigma} \leq \frac{2m\bar{m}\lambda_n{C}_2\omega}{12{m}\bar{m}\lambda_n{C}_2} \leq \omega.$$
\item Letting $\sigma'$ be the minimum nonzero singular value of $L$. One can check that:
$$\sigma' \geq \sigma - \|L-L^\star\|_2 \geq \sigma - mC_2\lambda_n \geq 10mC_2\lambda_n \geq 8\|L-L^\star\|_2.$$
One can also obtain the following lower bounds for $\sigma'$:
\begin{eqnarray*}
\begin{aligned}
\sigma' &\geq \sigma - \|L-L^\star\|_2 \geq \sigma - mC_2\lambda_n \geq 6\zeta{m}C_2^2\psi\lambda_n - mC_2\lambda_n \geq 6\zeta{m}\psi{C}_2^2\lambda_n\\
\sigma' &\geq \sigma - \|L-L^\star\|_2 \geq \sigma - mC_2\lambda_n \geq \frac{\alpha'(\frac{1}{\zeta}-2(\kappa^\star+\omega))\lambda_n}{4\left[1+\frac{1}{3\zeta}\right]}
\end{aligned}
\end{eqnarray*}
where we have used $C_2\psi \geq 1$. Once again appealing to Proposition 2.2 of \cite{Chand2012} and simple algebra, we have:
\begin{eqnarray*}
\Phi_\gamma[\Gla\I{C}_{T'}] \leq m\psi\|C_{T'}\|_2\leq m\psi\frac{\|L-L^\star\|_2^2}{\sigma'} \leq m\psi\frac{C_2^2\lambda_n^2}{6\zeta{m}\psi{C}_2^2\lambda_n} \leq \frac{\lambda_n}{6\zeta}.
\end{eqnarray*}
From Lemma~\ref{lemma:5}, we have that $\|{C}_{T'\oplus\ones}\|_2 \leq \|C_{T'}\|_2$. Following the same logic as above, we can then show that:
$\Phi_\gamma[\Gla\I{C}_{T'\oplus\ones}] \leq \frac{\lambda_n}{6\zeta}$.
\item Finally, we show that:
\begin{eqnarray*}
\Phi_\gamma[C_{T'}] \leq m\|\proj_{{T'}^\perp}(L-L^\star)\|_2 \leq m\frac{\|L-L^\star\|_2^2}{\sigma'} \leq \frac{mC_2^2\lambda_n^2}{\sigma'} \leq \frac{4\lambda_n}{\alpha'(\frac{1}{\zeta}-2(\kappa^\star+\omega))}\left[1+\frac{1}{3\zeta}\right].
\end{eqnarray*}
\end{enumerate}
\end{proof}

Consider any optimal solution $(\hat{S}^\mathcal{M},\hat{L}^\mathcal{M})$ of \eqref{eqn:estimator_nonconvex}. We will show that $(\hat{S}^\mathcal{M},\hat{L}^\mathcal{M})$ is the unique solution of the nonconvex program \eqref{eqn:estimator_nonconvex}, as well as the unique solution of \eqref{eqn:estimator_no_psd}.

\subsection{Variety constrained program to tangent space constrained program}
Let $(\hat{S}^\mathcal{M},\hat{L}^\mathcal{M})$ be any optimal solution of \eqref{eqn:estimator_nonconvex}. In Corollary~\ref{corr:nonconvex}, we conclude that the variables  $(\hat{S}^\mathcal{M},\hat{L}^\mathcal{M})$ are smooth points of their respective varieties. As a result, the rank constraint $\text{rank}(L) \leq \text{rank}(L^\star)$ can be linearized to $L \in T(\hat{L}^\mathcal{M})$. Since all the remaining constraints are convex, the optimum of the linearized program is also the optimum of \eqref{eqn:estimator_nonconvex}. Moreover, we once more appeal to Corollary~\ref{corr:nonconvex} to conclude that the constraints in $\mathcal{M}$ along ${T^\star}^\perp$ are strictly feasible at $\hat{L}^\mathcal{M}$. As a result, these constraints are inactive and can be removed in this ``linearized program". We now argue that the constraint $\Phi_\gamma[\Gla\I\Gl(\hat{S}^\mathcal{M}-S^\star,\hat{L}^\mathcal{M}-L^\star)]$ is inactive. For notational simplicity, we let $T' = T(\hat{{L}}^\M)$ and $\Hs' = \Omega^\star \times T'$, we consider the following optimization problem:
\begin{equation}
\begin{aligned}
    (\tilde{S},\tilde{L}) = \argmin_{S \in \mathbb{S}^{p},L \in \mathbb{S}^{p}} &~~-\log{\det}(U^T(S-L)U) - \mathrm{tr}((S-L)\hat{\Gamma}_O/2) + \lambda_n(\|S\|_{1} + \gamma\|L\|_\star),\\
    \text{subject-to}&~~~ (S,L) \in \Hs', S-L \in \mathrm{span}(\ones).
\end{aligned}
\label{eqn:estimator_tangent}
\end{equation}
We prove that under conditions imposed on the regularization parameter $\lambda_n$, the pair of variables $(\hat{S}^\mathcal{M},\hat{L}^\mathcal{M})$ is the unique optimum of \eqref{eqn:estimator_tangent}. First, note that the optimum of \eqref{eqn:estimator_tangent} is unique since it is a strictly convex program convex because the negative log-likelihood terms have a strictly positive-definite Hessian due to property $p1)$ in \eqref{eqn:Hessian_prop}. To show that $(\hat{S}^\mathcal{M},\hat{L}^\mathcal{M})$ is the optimum of \eqref{eqn:estimator_tangent}, it suffices to show strict feasibility of the constraint, that is: $\Phi_\gamma[\Gla\I\Gl(\tilde{S}-S^\star,\tilde{L}-L^\star)] < C_0\lambda_n$. 

From optimality conditions of \eqref{eqn:estimator_tangent}, there exists $Q_\Omega \in {\Omega^\star}^\perp$, $Q_{T} \in {T'}^\perp$, $t \in \mathbb{R}$ such that:
\begin{equation}
\begin{aligned}
-\hat{\Gamma}_O/2 - (\tilde{S}-\tilde{L})^{+} + t\ones+Q_\Omega &= -\lambda\partial\|\tilde{S}\|_{1},\\
\hat{\Gamma}_O/2 + (\tilde{S}-\tilde{L})^{+}-t\ones+Q_{T} &= -\lambda\gamma\partial\|\tilde{L}\|_{\star},\\
 \tilde{S}-\tilde{L} &\in \mathrm{span}(\ones).
\label{eqn:optimality_tangent_temp}
\end{aligned}
\end{equation}
Let the reduced SVD of $\tilde{L}$ be given by $\tilde{L} = \bar{U}\bar{D}\bar{V}^T$ and $Z = (\lambda\text{sign}(\tilde{S}),\lambda\gamma\bar{U}\bar{V}^T)$. Following a similar logic as in Section \ref{sec:proof_strategy} and restricting the optimality conditions to the space of $\Hs$, we have the following equivalent characterization of the optimality conditions:
\begin{eqnarray}
\begin{aligned}
\proj_{\Hs'}\Gla(\I(\Gl(\Delta_S,\Delta_{L}+t\ones))+\mathcal{R}_{\Gamma^\star_0}\Gl(\Delta_S,\Delta_{L}+t\ones)+E_n) &= Z,\\
 \tilde{S}-\tilde{L} &\in \mathrm{span}(\ones).
\end{aligned}
\label{eqn:optimality}
\end{eqnarray}
Here, $\Delta_S = \tilde{S}-S^\star$, $\Delta_{L} = \tilde{L}-L^\star$. In the remaining, we will denote $\Delta_{L+} = \tilde{L}-L^\star+t\ones$. Our result relies on the following propositions to control the remainder term. 
\begin{proposition}
Suppose $\Phi_\gamma(\Delta_S,\Delta_{L+}) \leq \frac{1}{2C_1}$ for $C_1 = \psi(m+d^\star)$ and any $\Delta_S \in \Omega^\star$. Then, $\Phi_\gamma[\Gla\mathcal{R}_{\Gamma^\star_0}(\Gl(\Delta_S,\Delta_{L+}))] \leq {2m\psi{C}_1^2\Phi_\gamma(\Delta_S,\Delta_{L+})^2}$.
\label{remainder_control}
\end{proposition}
\begin{proof}[Proof of Proposition~\ref{remainder_control}] We have that:
\begin{eqnarray*}
\begin{aligned}
\|\Gl(\Delta_{S},\Delta_{L+})\|_2 &\leq \|\Delta_S\|_2 + \|\Delta_{L+}\|_2\leq \theta(\Omega^\star)\|\Delta_S\|_\infty + \gamma\frac{\|\Delta_{L+}\|_2}{\gamma}\leq (\gamma+\theta(\Omega^\star))\Phi_\gamma(\Delta_S,\Delta_{L+}) \\&\leq (m+d^\star)\Phi_\gamma(\Delta_S,\Delta_{L+}) \leq \frac{1}{2\psi}.
\end{aligned}
\end{eqnarray*}
Therefore, 
\begin{eqnarray*}
\begin{aligned}
\|\mathcal{R}_{\Gamma^\star_0}(\Gl(\Delta_S,\Delta_{L+}))\|_2 &\leq \psi\sum_{k=2}^\infty(\|\Delta_S+\Delta_{L+}\|_2\psi)^k \leq \psi^3\|\Delta_S+\Delta_{L+}\|_2^2\frac{1}{1-\|\Delta_S+\Delta_{L+}\|_2\psi} \\
&\leq 2\psi^3\left(1+\frac{\alpha'}{6\zeta}\right)^2\Phi_\gamma(\Delta_S,\Delta_{L+})^2 = 2\psi{C}_2^2\Phi_\gamma(\Delta_S,\Delta_{L+})^2.
\end{aligned}
\end{eqnarray*}
Putting everything together, we have the desired result.
\end{proof}
Notice that the bound on the remainder term is dependent on the error term $\Phi_\gamma(\Delta_S,\Delta_{L+})$. In the following proposition, we bound this error so we can control the remainder term.  %Let $C_{T} = \proj_{{T'}^\perp}(L^\star)$.
\begin{proposition} Let $\tilde{S},\tilde{L}$ be the solution of convex program \eqref{eqn:estimator_tangent}. Define 
$$ r = \max\left\{\frac{4}{\alpha'(\frac{1}{\zeta}-2(\kappa^\star+\omega))}[\Phi_\gamma(\Gla{E}_n) + \Phi_\gamma(\Gla\I{C}_{T'}) + \lambda_n],\Phi_\gamma[(0,C_{T'})]\right\}.$$
If we have that $r \leq \min\left\{\frac{8\alpha'}{C_1}, \frac{\min\{\alpha',1\}(\frac{1}{\zeta} - 2(\kappa^\star+\omega))}{16m\psi{C}_2^2}\right\}$, then $\Phi_\gamma(\Delta_{S},\Delta_{L}) \leq \frac{4r\sqrt{5h}}{1-\sqrt{1- ({\kappa^\star}^2-\omega)^2}}$\\ and $\Phi_\gamma(0,t\ones) \leq \frac{4r\sqrt{5h}}{1-\sqrt{1- ({\kappa^\star}^2-\omega)^2}}$.% and $\max\{t\ones,\|\Delta_{L}\|_2 \leq \frac{8r{m}}{\kappa^\star}\}.$
\label{brouwer}
\end{proposition}
The proof of the proposition relies on the following lemma which we state and prove first.
\begin{lemma}Consider the following optimization:
\begin{equation}
\begin{aligned}
    \argmin_{S \in \mathbb{S}^{p},L \in \mathbb{S}^{p}} &\log{\det}(U^T(S-L)U) - \mathrm{tr}((S-L)\hat{\Gamma}_O/2) + \mathrm{tr}(\ones(S-L))+\lambda_n(\|S\|_{1} + \gamma\|L\|_\star).\\
    \text{subject-to}&~~~(S,L) \in \mathbb{H}'
\end{aligned}
\label{eqn:estimator_auxillary}
\end{equation}
Then, the solution of \eqref{eqn:estimator_auxillary} is unique and is equal to $\tilde{S},\tilde{L}$ (i.e., the solution of \eqref{eqn:estimator_tangent}).
\label{lemma:auxillary}
\end{lemma}
\begin{proof}[Proof of Lemma~\ref{lemma:auxillary}] Note that by property $p1)$ in \eqref{eqn:Hessian_prop}, the estimator \eqref{eqn:estimator_auxillary} is strictly convex. We will denote the optimal solution of \eqref{eqn:estimator_auxillary} by $(\tilde{S},\tilde{L})$. We are using the same notation as the optimal solution of \eqref{eqn:estimator_tangent} as we will show momentarily that these optimal solutions are identical. Specifically, define $Z$ as is done before Proposition ~\ref{remainder_control}. Let $\Delta_S = \tilde{S}-S^\star$ and $\Delta_L= \tilde{L}-L^\star$. The optimality condition of \eqref{eqn:estimator_auxillary} is given by:
\begin{eqnarray}
\begin{aligned}
\proj_{\Hs'}\Gla(\I\Gl(\Delta_S,\Delta_{L}+t\ones)+\mathcal{R}_{\Gamma^\star_0}\Gl(\Delta_S,\Delta_{L}+t\ones)+E_n) &= Z.%,\\
%(\tilde{S}-\tilde{L})\ones &= 0.
\end{aligned}
\label{eqn:optimality_aux}
\end{eqnarray}
Notice that the optimality condition \eqref{eqn:optimality_aux} is identical to the first condition in \eqref{eqn:optimality}. Since \eqref{eqn:estimator_auxillary} has a unique solution, then, the optimal solutions of \eqref{eqn:estimator_tangent} and \eqref{eqn:estimator_auxillary} coincide. 
\end{proof}
\begin{proof}[Proof of Proposition~\ref{brouwer}] 
Since $T'$ is a tangent space such that $\rho(T',T^\star)\leq \omega$, we have from Property~$p1$ in \eqref{eqn:Hessian_prop} that the operator $\mathcal{B} = (\proj_{\mathbb{H}'}\Gla\I\Gl\proj_{\mathbb{H}'})^{-1}$ is bijective and is well-defined. Consider the following function taking as input $(\delta_S,\delta_{L+}) \in \mathbb{Q}'$ where $\mathbb{Q}' = \Omega^\star \times (T' \oplus t\ones)$:
\begin{equation*}
F(\delta_S,\delta_{L+}) = (\delta_S,\delta_{L+})-\mathcal{B}\left\{\proj_{\Hs'}\Gla[\I\Gl(\delta_S,\delta_{L+})+\mathcal{R}_{\Gamma^\star_0}(\Gl(\delta_S,\delta_{L+}+C_{T'}))+ \I{C}_{T'}+E_n-Z\right\}.
\end{equation*}
Here, $C_{T'} = \proj_{{T'}^\perp}(L^\star)$. Now a point $(\delta_S,\delta_{L+})$ is a fixed point of $F$ if and only if $\proj_{\Hs'}\Gla[\I\Gl(\delta_S\allowbreak,\delta_{L+})+\mathcal{R}_{\Gamma^\star_0}(\Gl(\delta_S,\delta_{L+}+C_{T'}))+\I{C}_{T'}+E_n]=Z$. Further, a fixed point $(\delta_S,\delta_{L+})$ provides certificates of optimality for \eqref{eqn:estimator_auxillary}. Specifically, let $\tilde{S} = S^\star+\delta_S$. By Lemma~\ref{lemma:6}, find a unique decomposition of $\delta_{L+} = L+t\ones$ where $L \in T'$. Then, let $\tilde{L} = \proj_{T'}(L^\star)+L$. By construction, the parameters $(\tilde{S},\tilde{L})$ then satisfy the optimality condition for \eqref{eqn:optimality_aux} and thus also the optimality condition of \eqref{eqn:optimality} after appealing to Lemma~\ref{lemma:auxillary}. In other words, the fixed point of the function $F$ is $\proj_{\Hs'}(\Delta_S,\Delta_{L})+(0,t\ones)$. 

Next, using Brouwer's fixed point theorem, we show that $F$ has a fixed point that lies in the ball $\mathbb{B}_r = \{(\delta_S,\delta_{L+}) \in \mathbb{Q}' | \Phi_\gamma(\delta_S,\delta_{L+}) \leq r\}$. An equivalent formulation of $F$ is:
\begin{eqnarray*}
F(\delta_S,\delta_{L+}) = \proj_{{\Hs'}^\perp}(\delta_S,\delta_{L+}) - \mathcal{B}\Big\{\proj_{\Hs'}\Gla[\mathcal{R}_{\Gamma^\star_0}(\Gl(\delta_S,\delta_{L+}+C_{T'}))+\I[C_{T'}+\Gl\proj_{{\Hs'}^\perp}(\delta_S,\delta_{L+})]\\+E_n-Z\Big\}.
\end{eqnarray*}
First, note that by appealing to Lemma~\ref{lemma:4}, we have that: $\Phi_\gamma\left[\proj_{{\Hs'}^\perp}(\delta_S,\delta_{L+})\right] \leq 2r({\kappa^\star}+\omega).$
Similarly, we have from Property~$p3$ in \eqref{eqn:Hessian_prop} that: $\Phi_\gamma\left[\mathcal{B}\left\{\proj_{\Hs'}\Gla{I}\Gl\proj_{{\Hs'}^\perp}(\delta_S,\delta_{L+}) \right\}\right] \leq r\left(1-\frac{1}{\zeta}\right)$. Finally, we note that:
\begin{eqnarray*}
\begin{aligned}
&\Phi_\gamma\left[\mathcal{B}\left\{\proj_{\Hs'}\Gla[\mathcal{R}_{\Gamma^\star_0}(\Gl(\delta_S,\delta_{L+}+C_{T'}))+\I{C}_T'+E_n-Z \right\}\right]\\
&\leq \frac{2}{\alpha'}\left(\Phi_\gamma[\Gla\mathcal{R}_{\Gamma^\star_0}(\Gl(\delta_S,\delta_{L+}+C_{T'}))] + \Phi_\gamma[\I{C}_{T'}] + \Phi_\gamma[E_n] + \lambda_n\right)\\
&\leq \frac{r(\frac{1}{\zeta} - 2(\kappa^\star-\omega))}{2} + \frac{2}{\alpha'}\left(\Phi_\gamma[\Gla\mathcal{R}_{\Gamma^\star_0}(\Gl(\delta_S,\delta_{L+}+C_{T'}))]\right)
\end{aligned}
\end{eqnarray*}
where the last inequality is by the definition of $r$. By the assumption on $r$, we have that $\Phi_\gamma((\delta_S,\delta_{L+})+(0,{C}_{T'})) \leq \frac{1}{2C_1}$. And so we can appeal to Proposition~\ref{remainder_control} to conclude that:
\begin{eqnarray*}
\frac{2}{\alpha'}\Phi_\gamma[\Gla\mathcal{R}_{\Gamma^\star_0}(\Gl((\delta_S,\delta_{L+}+{C}_{T'}) \leq \frac{8m\psi{C}_1^2{r}^2}{\alpha'} \leq \frac{16m\psi{C}_2^2{r}}{\alpha'(\frac{1}{\zeta} - 2(\kappa^\star+\omega))}\frac{r(\frac{1}{\zeta} - 2(\kappa^\star+\omega))}{2} \leq r/2,
\end{eqnarray*}
where the last inequality uses the bound on $r$. So by Brouwer's fixed point theorem, we conclude that:  $\Phi_\gamma[\proj_{\mathbb{H}'}(\Delta_{S},\Delta_{L})+(0,t\ones)]\leq r$. Finally, note that: $\Phi_\gamma[\proj_{{\mathbb{H}'}^\perp}(\Delta_{S},\Delta_{L})] \leq r$. Thus, $\Phi_\gamma[(\Delta_{S},\Delta_{L})+(0,t\ones)]\leq 2r$. Finally, appealing to Lemma~\ref{lemma:new} and some manipulations, we have the bound $\max\{\Phi_\gamma(\Delta_S,\Delta_L),t\ones\} \leq \frac{4r\sqrt{5h}}{1-\sqrt{1- ({\kappa^\star}^2-\omega)^2}}$.
\end{proof}
\begin{proposition}Suppose that $\Phi_\gamma[\Gla{E}_n] \leq \frac{\lambda_n}{6\zeta}$ and suppose that:
\begin{eqnarray*}
    \begin{aligned}
\lambda_n \leq \min\Bigg\{\min\left\{\frac{8\alpha'}{C_1}, \frac{\min\{\alpha',1\}(\frac{1}{\zeta} - 2(\kappa^\star+\omega))}{16m\psi{C}_2^2}\right\}\frac{\alpha'(\frac{1}{\zeta}-2(\kappa^\star+\omega))}{4(1+\frac{1}{3\zeta})}\\,\frac{\alpha'(\frac{1}{\zeta}-2(\kappa^\star+\omega))}{64C_1(1+\frac{1}{3\zeta})}, \frac{\alpha'^2(\frac{1}{\zeta}-2(\kappa^\star+\omega))^2}{6144\zeta(1+\frac{1}{3\zeta})^2}\Bigg\}.
\end{aligned}
\end{eqnarray*}
Then, we have that: $\tilde{S} = \hat{S}^\mathcal{M}$, $\tilde{L} = \hat{L}^\mathcal{M}$. 
\label{prop:tangent_space_non_convex}
\end{proposition}
\begin{proof}
From Corollary~\ref{corr:nonconvex}, we have that $\Phi_\gamma[\Gla\I{C}_{T'}] \leq \frac{\lambda_n}{6\zeta}$. We then have that:
\begin{eqnarray*}
\begin{aligned}
&\frac{4}{\alpha'(\frac{1}{\zeta}-2(\kappa^\star+\omega))}[\Phi_\gamma(\Gla{E}_n) + \Phi_\gamma(\Gla\I{C}_{T'}) + \lambda_n]\\ &\leq \frac{4\lambda_n}{\alpha'(\frac{1}{\zeta}-2(\kappa^\star+\omega))}\left[1+\frac{1}{3\zeta}\right] \leq \min\left\{\frac{8\alpha'}{C_1}, \frac{\min\{\alpha',1\}(\frac{1}{\zeta} - 2(\kappa^\star+\omega))}{16m\psi{C}_2^2}\right\}.
\end{aligned}
\end{eqnarray*}
We also have from Corollary~\ref{corr:nonconvex} that $\Phi_\gamma(\Gla{C}_{T'}) \leq \frac{4\lambda_n}{\alpha'(\frac{1}{\zeta}-2(\kappa^\star+\omega))}\left[1+\frac{1}{3\zeta}\right]$. Let $r = \frac{4\lambda_n}{\alpha'(\frac{1}{\zeta}-2(\kappa^\star+\omega))}\allowbreak\left[1+\frac{1}{3\zeta}\right]$. We can appeal to Proposition~\ref{brouwer} to conclude that: 
$$\Phi_\gamma[\Delta_S,\Delta_{L}] \leq \frac{16\lambda_n\sqrt{5h}}{\alpha'(1-\sqrt{1- ({\kappa^\star}^2-\omega)^2})(\frac{1}{\zeta}-2(\kappa^\star+\omega))}\left[1+\frac{1}{3\zeta}\right].$$ From the bound on $\lambda_n$, we have that: $\Phi_\gamma[\Delta_S,\Delta_{L}]\leq \frac{1}{2C_1}$. So we can appeal to Proposition~\ref{remainder_control} to conclude that: 
\begin{eqnarray}
\Phi_\gamma[\Gla\mathcal{R}_{\Gamma_O^\star}\Gl(\Delta_S,\Delta_{L})] \leq 2m\psi{C}_1^2\Phi_\gamma[\Delta_S,\Delta_{L}]^2 \leq \frac{\lambda_n}{6\zeta},
\label{eqn:remainder_bound_final}
\end{eqnarray}
where here again we use the bound on $\lambda_n$. 
Note that $\Delta_{L+} = \Delta_{L} + t\ones$. We have from Corollary~\ref{corr:nonconvex} that $\Phi_\gamma[\Gla\I{C}_{T'}] \leq \frac{\lambda_n}{6\zeta}$. From the optimality conditions of \eqref{eqn:estimator_tangent}, we have that:
\begin{eqnarray*}
\begin{aligned}
&\Phi_\gamma(\proj_{\Hs'}\Gla\I\Gl\proj_{\Hs'}(\Delta_S,\Delta_{L}))\\ &\leq 2\lambda_n + 2\Phi_\gamma(0,t\ones) + \Phi_\gamma[\Gla\mathcal{R}_{\Gamma_O^\star}\Gl(\Delta_S,\Delta_{L})]+ \Phi_\gamma[\proj_{\Hs'}\Gla\I{C}_{T'}] + \Phi_\gamma[\Gla{E}_n], \\
&\leq  2\lambda_n+\frac{\lambda_n}{2\zeta} + \frac{16\lambda_n\sqrt{5h}}{\alpha'(1-\sqrt{1- ({\kappa^\star}^2-\omega)^2})(\frac{1}{\zeta}-2(\kappa^\star+\omega))}\left[1+\frac{1}{3\zeta}\right] ,
\end{aligned}
\end{eqnarray*}
where the second inequality follows from bound on $\Phi_\gamma((0,t\ones))$ in Proposition~\ref{brouwer}. Appealing to property $p2$ in \eqref{eqn:Hessian_prop}: $\Phi_\gamma(\proj_{{\Hs'}^\perp}\Gla\I\Gl\proj_{\Hs'}(\Delta_S,\Delta_{L})) \leq \Phi_\gamma(\proj_{{\Hs'}^\perp}\Gla\I\Gl\proj_{\Hs'}(\Delta_S,\Delta_{L}))$. Thus
\begin{eqnarray*}
\begin{aligned}
\Phi_\gamma(\Gla\I\Gl(\Delta_S,\Delta_{L})) &\leq \Phi_\gamma(\proj_{\Hs'}\Gla\I\Gl\proj_{\Hs'}(\Delta_S,\Delta_{L})) +\Phi_\gamma(\proj_{{\Hs'}^\perp}\Gla\I\Gl\proj_{\Hs'}(\Delta_S,\Delta_{L}))\\&+\Phi_\gamma[\Gla\I{C}_{T'}] \leq  8\lambda_n+\frac{32\lambda_n\sqrt{5h}}{\alpha'(1-\sqrt{1- ({\kappa^\star}^2-\omega)^2})(\frac{1}{\zeta}-2(\kappa^\star+\omega))}\left[1+\frac{1}{3\zeta}\right]\\
&<  C_0\lambda_n.
\end{aligned}
\end{eqnarray*}
\end{proof}

\subsection{Removing the tangent space constraint}
It remains to connect the estimator \eqref{eqn:estimator_tangent} with \eqref{eqn:estimator}. In particular, we check that $\tilde{S} = \hat{S}$ and $\tilde{L} = \hat{L}$ where $(\tilde{S},\tilde{L})$ is the solution of \eqref{eqn:estimator_tangent} and $(\hat{S},\hat{L})$ is the solution of \eqref{eqn:estimator}. We formalize this in the following proposition.
\begin{proposition}Suppose that $\Phi_\gamma[\Gla{E}_n] \leq \frac{\lambda_n}{6\zeta}$. Then, $\tilde{S} = \hat{S}$ and $\tilde{L} = \hat{L}$.
\label{prop:removing_tangent_space}
\end{proposition}
\begin{proof}[Proof of Proposition~\ref{prop:removing_tangent_space}]
We must show that $(\tilde{S},\tilde{L})$ satisfy the optimality conditions of \eqref{eqn:estimator_no_psd}
 in \eqref{eqn:optimality_orig}, namely that there exists a dual variable $t$ such that
\begin{eqnarray}
\begin{aligned}
\proj_{\Hs}\Gla(\I(\Gl(\Delta_S,\Delta_L+t\ones))+\mathcal{R}_{\Gamma^\star_0}\Gl(\Delta_S,\Delta_L)+E_n) &= Z,\\
\Phi_\gamma(\proj_{\Hs^\perp}\Gla(\I(\Gl(\Delta_S,\Delta_L+t\ones))+\mathcal{R}_{\Gamma^\star_0}\Gl(\Delta_S,\Delta_L)+E_n)) &<1,\\
\tilde{S}-\tilde{L} &\in \mathrm{span}(\ones),
\end{aligned}
\label{eqn:optimality_orig_again}
\end{eqnarray}
where $\Delta_S = \tilde{S}-S^\star$ and $\Delta_{L} = \tilde{L}-L^\star$. Notice that the first and third optimality conditions are the same as \eqref{eqn:optimality}. It remains to show the second inequality where the strict inequality is to ensure that $(\tilde{S},\tilde{L})$ is the unique solution. It suffices to show that:
\begin{eqnarray}
\begin{aligned}
&\Phi_\gamma(\proj_{\Hs^\perp}\Gla(\I\proj_{{\mathbb{Q}'}}(\Gl(\Delta_S,\Delta_L+t\ones)) \\&<\lambda_n - \Phi_\gamma[\mathcal{R}_{\Gamma^\star_0}\Gl(\Delta_S,\Delta_L)] - \Phi_\gamma[\Gla\I{C}_{T'\oplus\ones}] - \Phi_\gamma[\Gla{E}_n].
\end{aligned}
\label{eqn:what_to_show}
\end{eqnarray}
Manipulating the first optimality condition, we have that:
\begin{eqnarray*}
\begin{aligned}
\Phi_\gamma(\proj_{\Hs}\Gla(\I\proj_{{\mathbb{Q}'}}(\Gl(\Delta_S,\Delta_L+t\ones)) &\leq\lambda_n + 2(\Phi_\gamma[\mathcal{R}_{\Gamma^\star_0}\Gl(\Delta_S,\Delta_L)] +\Phi_\gamma[\Gla\I{C}_{T'\oplus\ones}] \\&+\Phi_\gamma[\Gla{E}_n]) \leq \lambda_n + \frac{\lambda_n}{\zeta} = \lambda_n\left(1+\frac{1}{\zeta}\right),
\end{aligned}
\end{eqnarray*}
where we have here used the bound $\Phi_\gamma[\Gla\I{C}_{T'\oplus\ones}] \leq \frac{\lambda_n}{6\zeta}$ from Corollary~\ref{corr:nonconvex} and the bounds $\Phi_\gamma[\mathcal{R}_{\Gamma^\star_0}\Gl(\Delta_S,\Delta_L)] \leq \frac{\lambda_n}{6\zeta}$ from \eqref{eqn:remainder_bound_final} and $\Phi_\gamma[\Gla{E}_n] \leq \frac{\lambda_n}{6\zeta}$ from proposition statement. Appealing to property $p2$ in \eqref{eqn:Hessian_prop}, we then have that:
\begin{eqnarray*}
\begin{aligned}
\Phi_\gamma(\proj_{\Hs^\perp}\Gla(\I\proj_{{\mathbb{Q}'}}(\Gl(\Delta_S,\Delta_L+t\ones)) &\leq \lambda_n\left(1+\frac{1}{\zeta}\right)\left(1-\frac{1}{\zeta}\right)\\&= \lambda_n\left(1-\frac{1}{\zeta^2}\right) < \lambda_n\left(1-\frac{1}{2\zeta}\right). 
\end{aligned}
\end{eqnarray*}
Since $\Phi_\gamma[\mathcal{R}_{\Gamma^\star_0}\Gl(\Delta_S,\Delta_L)] +\Phi_\gamma[\Gla\I{C}_{T'\oplus\ones} + \Phi_\gamma[\Gla{E}_n] \leq \frac{\lambda_n}{2\zeta}$, \eqref{eqn:what_to_show} holds. 

\end{proof}

\subsection{Bounding the error term $\Phi_\gamma[\Gla{E}_n]$} 
Let $\lambda_n = C_5\left[\frac{24m\zeta}{\sqrt{c_5}}\sqrt{\frac{p^2\log(\tilde{C}_5p)}{k}}+\frac{6m\zeta}{\sqrt{k}}\right]$ where $c_5,C_5,\tilde{C}_5$ are defined in Theorem~\ref{cor:sample_variogram}. 
\begin{lemma} Under the conditions of Theorem~\ref{thm:main}, we have:
$$\mathbb{P}\left(\Phi_\gamma[\Gla{E}_n] \leq \frac{\lambda_n}{6\zeta}\right) \geq 1-p^{-1}.$$
%where $M$ is the constant from Corollary~\ref{cor:sample_variogram}.   
\end{lemma}
\begin{proof}
Note that  $\Phi_\gamma[\Gla{E}_n] \leq m\|\Gamma_O^\star-\hat{\Gamma}_O\|_2 \leq pm\|\Gamma_O^\star-\hat{\Gamma}_O\|_\infty$. To show that, $\Phi_\gamma[\Gla{E}_n] \leq \frac{\lambda_n}{6\zeta}$, it suffices to show that 
\begin{equation}
\|\Gamma_O^\star - \hat{\Gamma}_O\|_\infty \leq \frac{4C_5}{\sqrt{c_5}}\sqrt{\frac{\log(\tilde{C}_5p)}{k}} + \frac{C_5}{\sqrt{k}}.
\label{eqn:desired_for_sample}
\end{equation}
Based on the condition on $k$, it is straightforward to show that:
$$C_5\left\{\left(\frac{k}{n}\right)^\xi (\log(n/k))^2 + \frac{1+\vartheta}{\sqrt{k}}\right\} \leq \frac{4C_5}{\sqrt{c_5}}\sqrt{\frac{\log(\tilde{C}_5p)}{k}} + \frac{C_5}{\sqrt{k}}.$$
for $\vartheta = 2\sqrt{\log(\tilde{C}_5p)}/\sqrt{c_5}$. Note that $\vartheta \leq \sqrt{k}/\log(n)^4$. Furthermore, $k \leq n/2$. Appealing to Corollary~\ref{cor:sample_variogram}, we have that with probability greater than $1-\tilde{C}_5p^3e^{-c_5\vartheta^2} = 1-p^{-1}$ that the bound in \eqref{eqn:desired_for_sample} is satisfied. 
\end{proof}
\subsection{Summary and putting things together}
Combining Propositions \ref{prop:tangent_space_non_convex}-\ref{prop:removing_tangent_space}, we conclude that under the conditions of Theorem~\ref{thm:main}, with probability greater than $1-1/p$, the optimal solution $(\hat{S},\hat{L})$ of \eqref{eqn:estimator_no_psd} is unique and equal to an optimal solution $(\hat{S}^\mathcal{M},\hat{L}^{\mathcal{M}})$ of \eqref{eqn:estimator_nonconvex}. From Corollary~\ref{corr:nonconvex}, we have that $\hat{S}-\hat{L} \succeq 0$, $\hat{L}\succeq 0$. Thus, $(\hat{S},\hat{L}) = (\hat{S}^\mathcal{M},\hat{L}^{\mathcal{M}})$ is also the unique minimizer of \eqref{eqn:estimator}. The guarantees on the closeness of $(\hat{S},\hat{L})$ to the population parameters $(S^\star,L^\star)$ follow from Corollary~\ref{corr:nonconvex} and Proposition~\ref{prop:nonconvex}.

\section{Refitting for \texttt{eglatent}}
\label{sec:refit}
Suppose $(\hat{S}, \hat{L})$ is the solution of \eqref{eqn:estimator} in the first step. We then obtain refitted parameters $(\tilde{S},\tilde{L})$ as the second step by solving the following convex optimization program:
\begin{equation*}
\begin{aligned}
    (\tilde{S},\tilde{L}) = \argmin_{S \in \mathbb{S}^{p}, L \in \mathbb{S}^{p}} &~~-\log{\det}(U^T(S-L)U) - \mathrm{tr}((S-L)\hat{\Gamma}_O/2),\\
    \text{s.t.}&~~~S-L \succeq 0, L \succeq 0, (S-L)\mathbf{1}_p = 0, \\
    &~~~\text{support}(S) \subseteq \text{support}(\hat{S}), \text{col-space}(L) \subseteq \text{col-space}(\hat{L}).
\end{aligned}
\label{eqn:estimator_refitted}
\end{equation*}
Here, the constraint $\text{support}(S) \subseteq \text{support}(\hat{S})$ restricts the graph structure of our refitted solution to be contained in the graph estimated in the first step. Similarly, the constraint $\text{col-space}(L) \subseteq \text{col-space}(\hat{L})$ restricts the row/column space of the refitted low-rank term to be contained in the row/column space estimated in the first step. 

\section{Additional experimental results} 

\subsection{Synthetic experiments on different graph structure}
\label{sec:additional_exp_different_graph_structure}
We consider the exact same setup as in the simulation study in Section~\ref{sec:struc_recov}. The only difference is that we specify the sub-graph $\mathcal{G}_0 = (E_O,O)$ among the observed variables to be an Erd\H{o}s--R{\"e}nyi with edge probability $0.08$ and set $\Theta_{ij}^\star$ to $-2$ for every $(i,j) \in E_O$ and zero otherwise. The rest of the simulation study is carried out as described in Section~\ref{sec:struc_recov}. Figure~\ref{fig:ads3} summarizes the performance of all the methods on 50 independent results. We again observe that our approach outperforms \texttt{eglearn}, and accurately recovers the graphical structure among the observed variables as well as the number of latent variables. In terms of validation likelihood, \texttt{eglatent} is a bit weaker than in the simulation with the cycle graph.
\FloatBarrier

\begin{figure}
 \centering {\includegraphics[width = 1\textwidth]{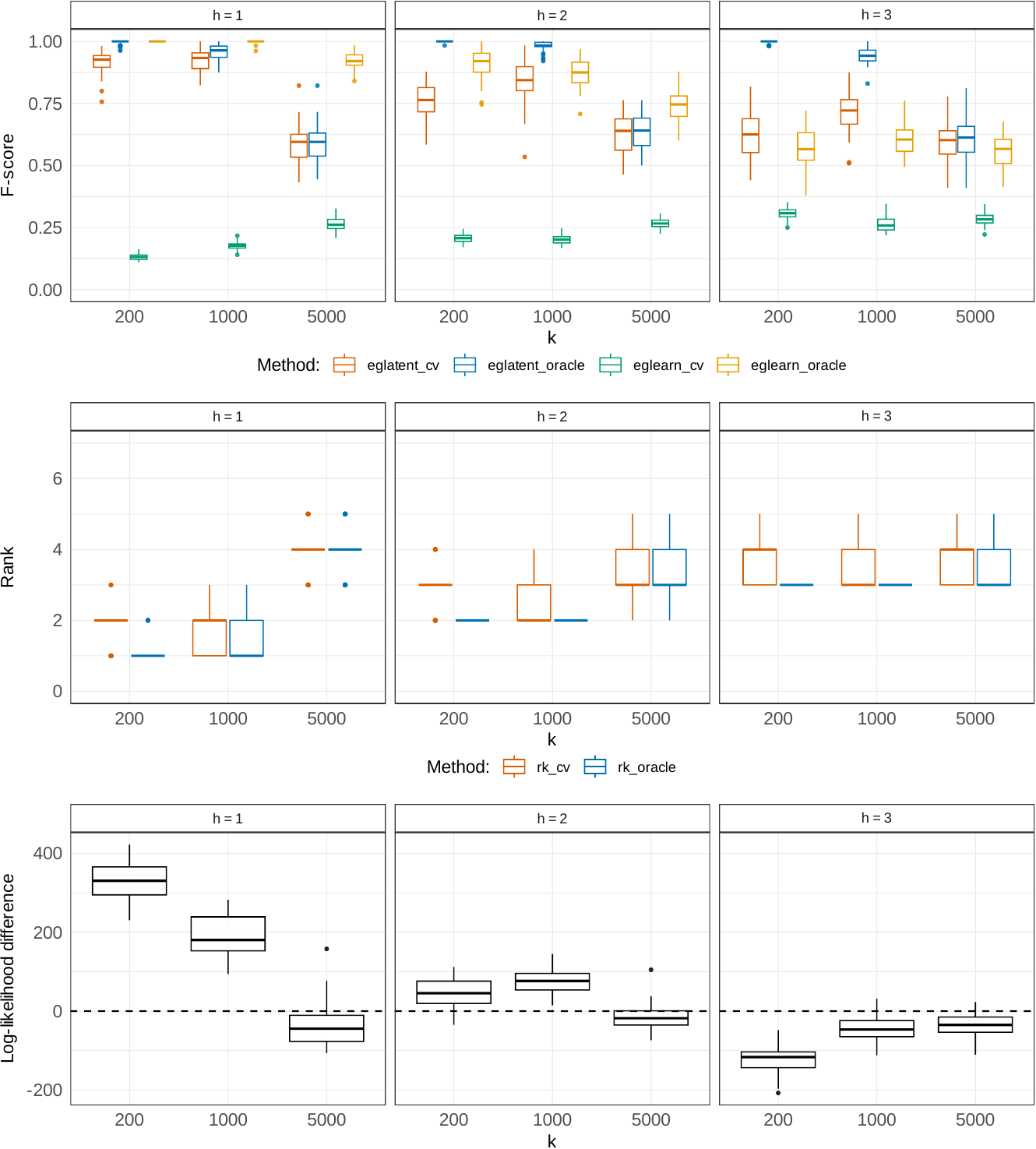}}
       \caption{$F$-score (top row) and estimated number of latent variables (middle row) of \texttt{eglatent} method with the selection of the tuning parameter based on the oracle and validation on the $F$-score for the random graph with $h=1,2,3$ latent variables and different effective sample sizes $k=200,1000,5000$. The bottom row shows the difference between best \texttt{eglatent} and best \texttt{eglearn} log-likelihoods on the validation set.}

    \label{fig:ads3}
\end{figure}
\FloatBarrier

\subsection{Synthetic experiments on different values of $\gamma$}
\label{sec:additional_exp_gamma}
{
We consider the exact same setup as in the simulation study in Section~\ref{sec:struc_recov}. The only difference is the values of $\gamma$ that are used in the \texttt{eglatent} estimator. We generate $k = 1000$ effective samples. Figure~\ref{fig:ads4} shows the performance of \texttt{eglatent} for $\gamma \in \{2,4,6\}$. We observe that the performance of \texttt{eglatent} does not vary drastically with changes in $\gamma$, and continues to perform better than  \texttt{eglearn}, especially for $h \in \{1,2\}$. We also notice that $\gamma = 4$ yields the best-validated model for $h \in \{1,2,3\}$, hence why this value was chosen in our experiments in Section~\ref{sec:struc_recov}.}
\FloatBarrier
\begin{figure}
    \centering {\includegraphics[width = 1\textwidth]{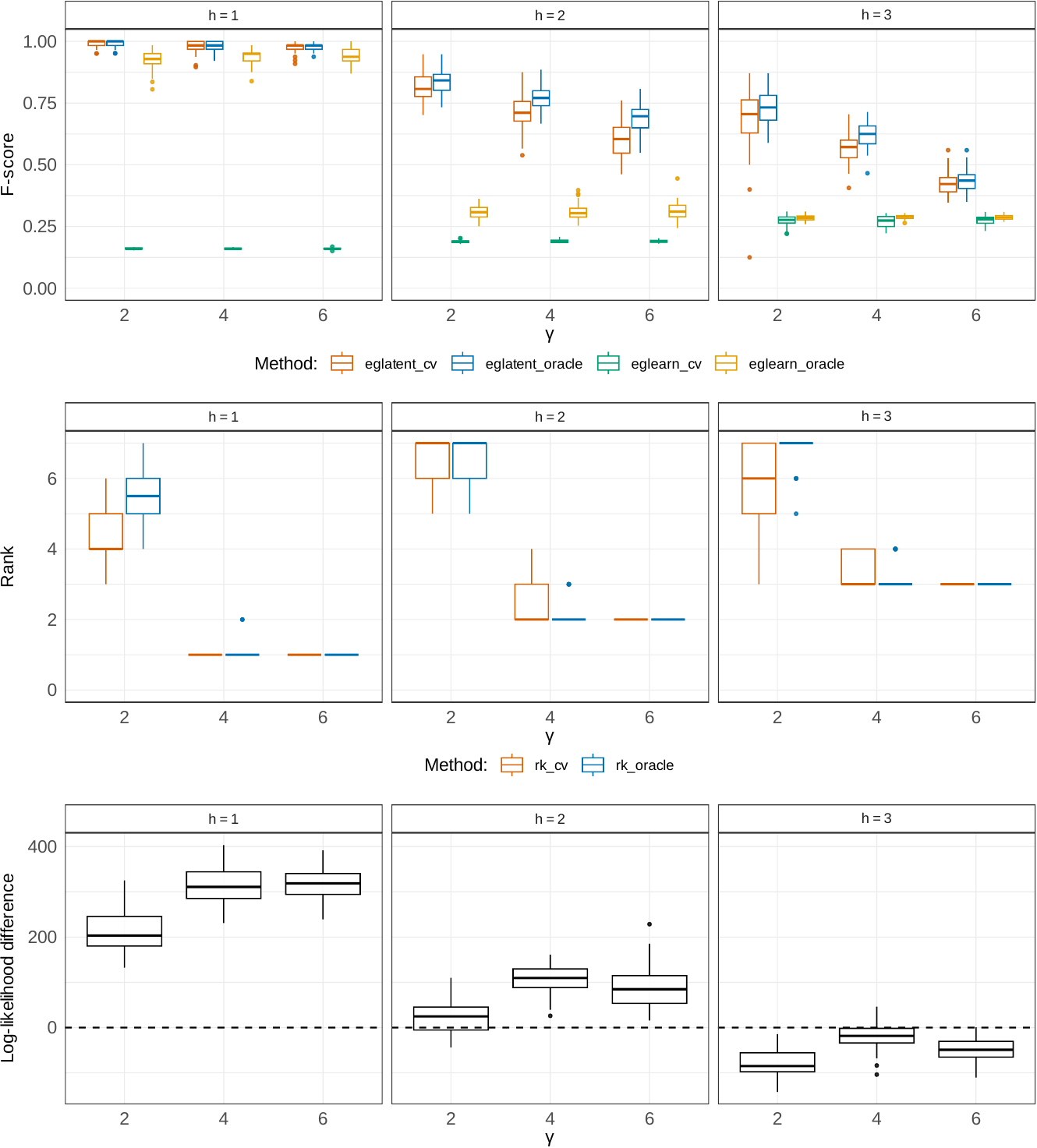}}
    \caption{$F$-score (top row) and estimated number of latent variables (middle row) of \texttt{eglatent} method with the selection of the tuning parameter based on the oracle and validation on the $F$-score for the cycle graph with $h=1,2,3$ latent variables and different regularization parameter $\gamma = 2,4,6$. The bottom row shows the difference between best \texttt{eglatent} and best \texttt{eglearn} log-likelihoods on the validation set. The effective sample size is set to $k = 1000$.}
    \label{fig:ads4}
\end{figure}
\FloatBarrier

\subsection{Synthetic experiments on comparison to the performance of Gaussian latent variable graphical model estimator}
\label{sec:additional_exp_gaussian}
{
We compare the performance of our \texttt{eglatent} estimator to the Gaussian latent variable graphical model estimator in \cite{Chand2012} (denoted by \texttt{LVGM}). We generate the data according to the setting in Appendix~\ref{sec:additional_exp_different_graph_structure}. As the approach in \cite{Chand2012} assumes Gaussian data, we transform the marginal distributions of each variable to standard normal distribution, before supplying the data to the Gaussian estimator. The following table compares the performance of the two estimators, where `CV' is when the regularization parameters are chosen via the validation set, and `Oracle' is when the regularization parameters are chosen to obtain the best $F$-score. 
}
\begin{table}[ht]
\caption{Perfomance of \texttt{eglatent} compared with Gaussian estimator in \cite{Chand2012}}
    \centering
    \resizebox{\columnwidth}{!}{\begin{tabular}{|c||cc||cc||cc||cc|}
    
    \hline
        & \multicolumn{2}{c||}{Oracle \texttt{eglatent}} & \multicolumn{2}{c||}{CV \texttt{eglatent}} & \multicolumn{2}{c||}{Oracle \texttt{LVGM}} & \multicolumn{2}{c|}{CV \texttt{LVGM}}\\
        \hline
      $\#$ latents ($h$) & $F$-score  & $\hat{h}$ & $F$-score  & $\hat{h}$ &$F$-score  & $\hat{h}$ &$F$-score  & $\hat{h}$ \\
       \hline
       $h=1$ & $0.94(\pm 0.02)$ & $1.58(\pm 0.53)$ & $0.92(\pm 0.04)$ & $1.68(\pm 0.55)$ & $0.08 (\pm 0.04)$ & $1.68 (\pm 1.88)$ & $0.06 (\pm 0.03)$ & $8.1 (\pm 0.83)$\\
              $h=2$ & $0.97(\pm 0.01)$ & $2(\pm 0)$ & $0.84(\pm 0.07)$ & $2.48(\pm 0.54)$ & $0.07 (\pm 0.04)$ & $4.94 (\pm 3.01)$ & $0.05 (\pm 0.04)$ & $7.94 (\pm 0.86)$\\
       $h=3$ & $0.93(\pm 0.03)$ & $3(\pm 0)$ & $0.70(\pm 0.08)$ & $3.42(\pm 0.57)$ & $0.06 (\pm 0.03)$ & $4.58(\pm 2.39)$ & $0.05 (\pm 0.03)$ & $8.41 (\pm 0.94)$\\
       \hline
    \end{tabular}}
    \label{tab:compare_benchmarks}
\end{table}

%{\subsection{Smaller effective size for robustness to zero latent variables}}
%\label{sec:additional_zero_latent_vars}
%We consider the setup in Section~\ref{sec:robust}, but let the effective sample size be $k = \lfloor n^{0.65}\rfloor = 200$. The following figure compares the performance of \texttt{eglatent} to \texttt{eglearn}. Similar to the $k = 2000$ case, we observe that \texttt{eglatent} is robust to the presence of no latent variables.

\subsection{{Additional results concerning the application}}
\label{sec:additional_application}
{
We report here the results of the application in Section~\ref{sec:application}. For thresholds $q=0.85$ and $q=0.95$,  Figures~\ref{fig:likelihood_85} and~\ref{fig:likelihood_95} show the number of edges of \texttt{eglatent} and of \texttt{eglearn} and the validation log-likelihood values as a function of the tuning parameter $\lambda_n$.
Figure~\ref{fig:graphs_comparison} compares the different estimated graphs among the observed variables for the three thresholds. }
\FloatBarrier
\begin{figure}
    \centering
    \includegraphics[width = .9\textwidth]{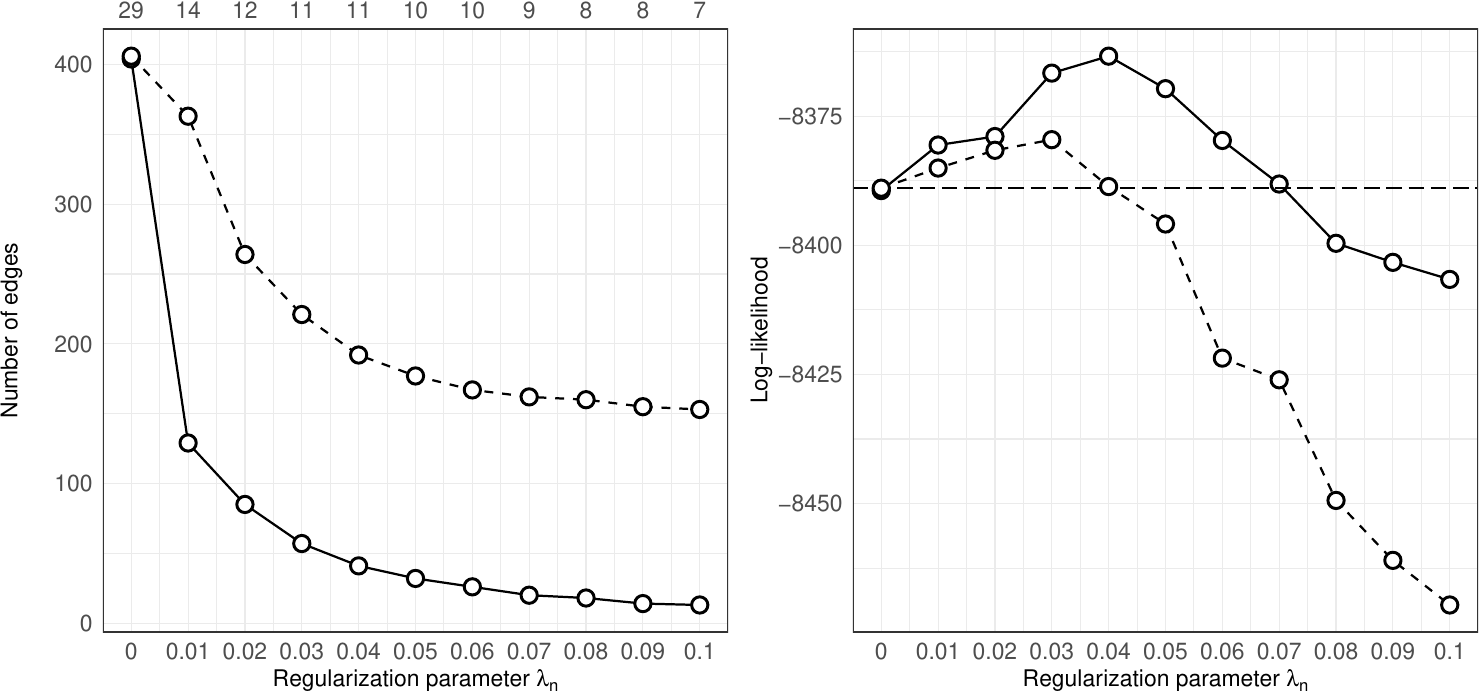}
    \caption{Results for threshold $q=0.85$. Left: number of edges of the estimated graph of \texttt{eglearn} (dashed line) and the estimated sub-graph of observed variables of \texttt{eglatent} (solid line) as functions of the regularization parameter $\rho$; top axis shows the number of latent variables in \texttt{eglatent}. Right: corresponding log-likelihoods; horizontal line is the validation log-likelihood of the fully connected graph.}
    \label{fig:likelihood_85}
\end{figure}

\begin{figure}
    \centering
    \includegraphics[width = .9\textwidth]{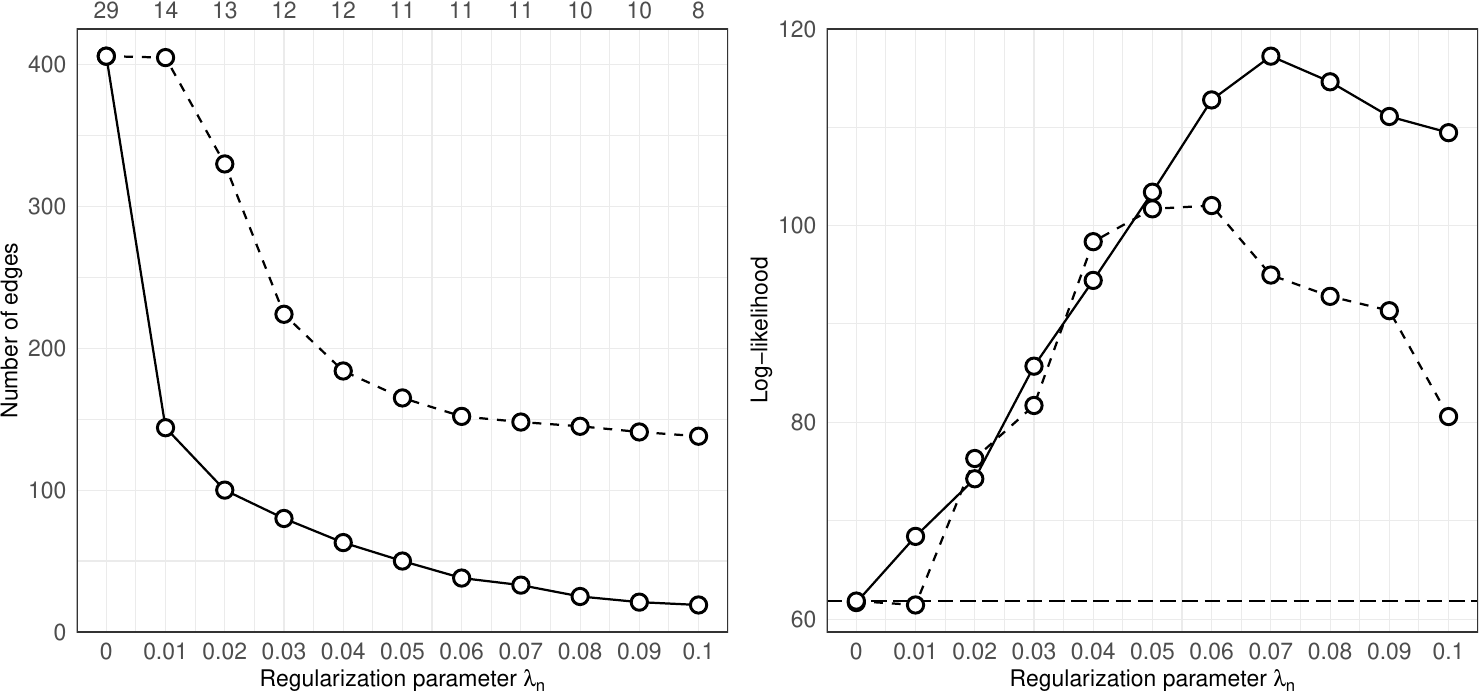}
    \caption{Results for threshold $q=0.95$. Left: number of edges of the estimated graph of \texttt{eglearn} (dashed line) and the estimated sub-graph of observed variables of \texttt{eglatent} (solid line) as functions of the regularization parameter $\rho$; top axis shows the number of latent variables in \texttt{eglatent}. Right: corresponding log-likelihoods; horizontal line is the validation log-likelihood of the fully connected graph.}
    \label{fig:likelihood_95}
\end{figure}

\begin{figure}
    \centering
    \includegraphics[width = .32\textwidth]{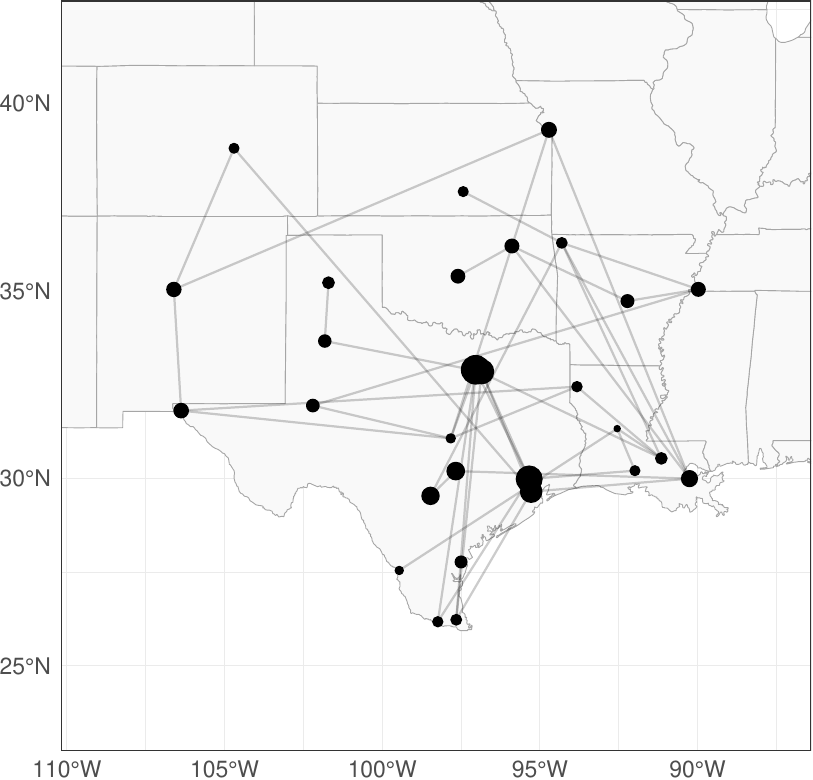}
    \includegraphics[width = .32\textwidth]{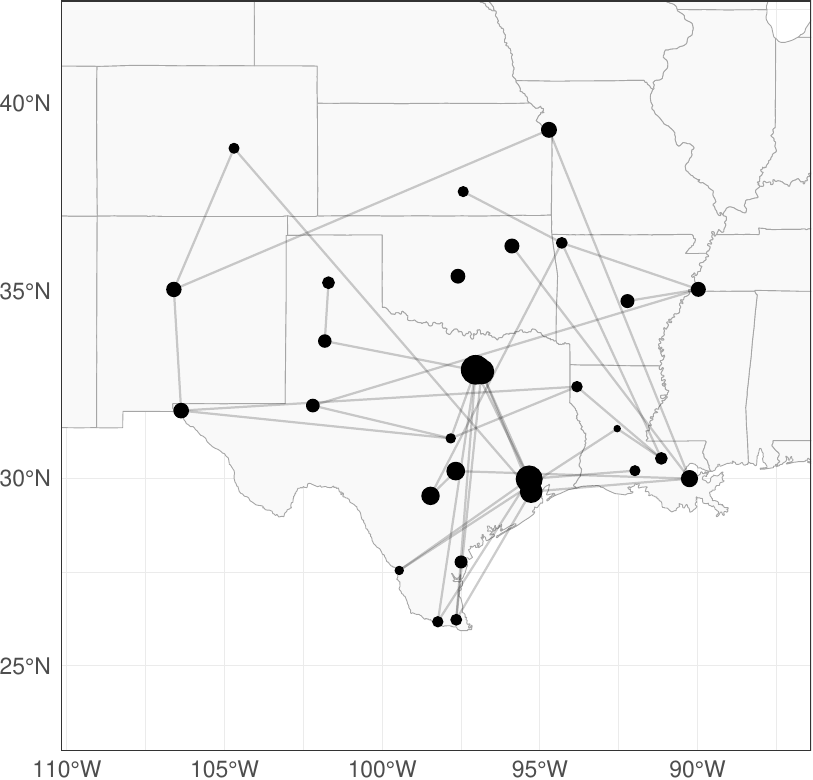}
    \includegraphics[width = .32\textwidth]{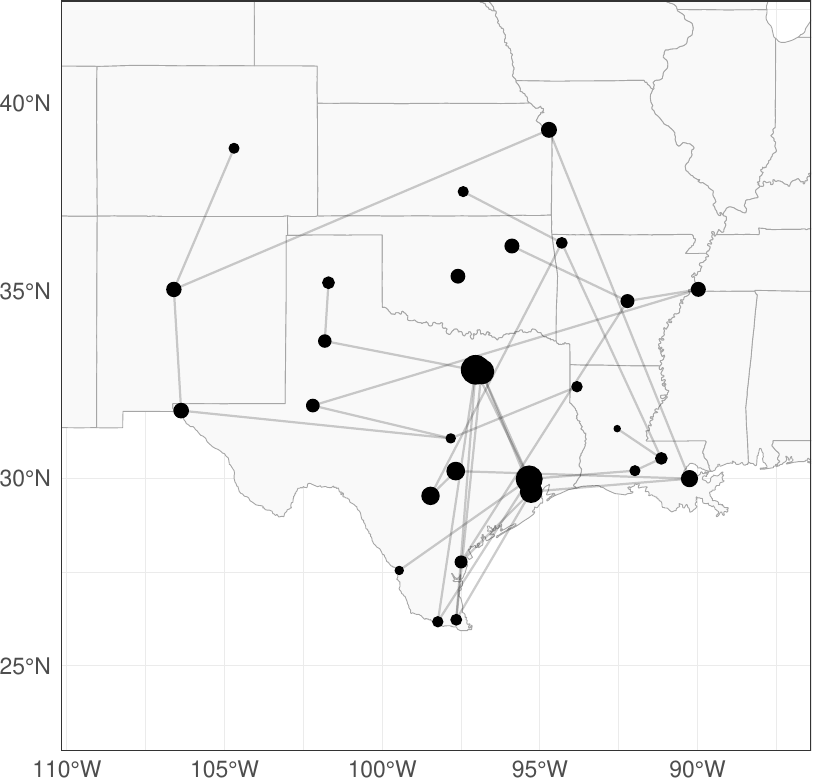}
    \caption{Airports in the Southern U.S.~(dots) and flight connections, where the thickness of the nodes indicates the average number of daily flights at the airports. Estimated sub-graphs corresponding to observed variables of optimal \texttt{eglatent} models for exceendance thresholds 0.85 (left), 0.90 (center) and 0.95 (right).}
    \label{fig:graphs_comparison}
\end{figure}
\FloatBarrier

\end{document}